\documentclass[showpacs,amsmath,amssymb,twocolumn,pra,superscriptaddress,notitlepage,floatfix]{revtex4-2}

\usepackage{qcircuit}
\usepackage{graphics}
\usepackage{amsmath,amssymb,amsthm,mathrsfs,amsfonts,dsfont}
\usepackage{subfigure, epsfig}
\usepackage{braket}
\usepackage{bm}
\usepackage{enumerate}
\usepackage{algorithm}
\usepackage{physics}
\usepackage{diagbox}

\usepackage{pgfplots}
\pgfplotsset{width=7cm,compat=1.9}

\usepgfplotslibrary{external}

\usepackage{algpseudocode}
\usepackage{color}
\usepackage{comment}

\usepackage[colorlinks=true, linkcolor=blue]{hyperref}
\graphicspath{{./figure/}}

\newtheorem{theorem}{Theorem}

\newtheorem{prop}{Proposition}
\newtheorem{fact}{Fact}

\newcommand{\renyi}{R$\mathrm{\acute{e}}$nyi }
\newcommand{\mc}{\mathcal}
\newcommand{\mb}{\mathbf}
\newcommand{\mbb}{\mathbb}
\newcommand{\id}{\mathbb{I}}

\newcommand{\comments}[1]{}

\begin{document}

\let\oldacl\addcontentsline
\renewcommand{\addcontentsline}[3]{}

\title{A hybrid framework for estimating nonlinear functions of quantum states}

\begin{abstract}
Estimating nonlinear functions of quantum states, such as the moment $\tr(\rho^m)$, is of fundamental and practical interest in quantum science and technology. Here we show a quantum-classical hybrid framework to measure them, where the quantum part is constituted by the generalized swap test, and the classical part is realized by postprocessing the result from randomized measurements. This hybrid framework utilizes the partial coherent power of the intermediate-scale quantum processor and, at the same time, dramatically reduces the number of quantum measurements and the cost of classical postprocessing. We demonstrate the advantage of our framework in the tasks of state-moment estimation and quantum error mitigation.
\end{abstract}
 \date{\today}
\author{You Zhou}
\email{you\_zhou@fudan.edu.cn}
\affiliation{Key Laboratory for Information Science of Electromagnetic Waves (Ministry of Education), Fudan University, Shanghai 200433, China}

\author{Zhenhuan Liu}
\email{liu-zh20@mails.tsinghua.edu.cn}
\affiliation{Center for Quantum Information, Institute for Interdisciplinary Information Sciences, Tsinghua University, Beijing 100084, China}

\maketitle

Quantum measurement is one of the fundamental building blocks of quantum physics, connecting the quantum world to its classical counterpart. The linear expectation value of the quantum state $\rho$ in the form $\tr(O\rho)$ can be measured directly in the basis of the observable $O$ by Born's rule. However, the measurement of nonlinear functions, such as the \renyi entropy and the moment $P_m=\tr(\rho^m)$, generally involves quantum circuits interfering among $m$ copies of the state $\rho$ by the generalized swap test \cite{Ekert2002Direct,Horodecki2002Method,Todd2004Polynomial,Garcia2013equivalent},  which transform nonlinear functions to linear ones on all that copies. Although there are significant experimental advances for the $m=2$ case \cite{Islam2015Measuring,Kaufmanen2016tanglement,cotler2019cooling,huang2022quantum}, it is still challenging to extend to a larger degree $m$ with a moderate system size on current quantum platforms \cite{preskill2018quantum}.

Recently alternative approaches based on the randomized measurement (RM) \cite{elben2023randomized}, such as the shadow estimation \cite{aaronson2019shadow,huang2020predicting}, are proposed. By postprocessing the results from random basis measurements of sequentially prepared states,
the shadow estimation is efficient for measuring local observables and the fidelity to some entangled states. However, when measuring nonlinear functions, such as the purity $P_2$, RM protocol inevitably needs an exponential number of measurements and postprocessings \cite{elben2020mixedstate,huang2020predicting,singlezhou}, which hinders its further applications for large systems. 

By trading off the swap test and the RM protocol, here we propose a hybrid framework for estimating nonlinear functions of quantum states to fill the gap between them, which inherits their advantages and reduces weaknesses. Specifically, one can conduct RM on a few jointly prepared copies of the state by utilizing the partial coherent power of the quantum processor, and then estimate many nonlinear functions in a more efficient way, i.e., less demanding on both the quantum hardware and classical post-processing.

The nonlinear function of quantum states $\{\rho_i\}_{i=1}^m$ with some observables $\{O_i\}_{i=1}^m$ of interest reads
\begin{equation}\label{eq:nonlinear}
    \begin{aligned}
       \tr(O_1\rho_1 O_2\rho_2 \cdots O_m \rho_m),
    \end{aligned}
\end{equation} 
which is general and includes, such as the state overlap and fidelity \cite{elben2020cross,Zhenhuan2022correlation}, the purity and higher-order moments \cite{van2012Measuring,Brydges2019Probing}, quantum Fisher information \cite{rath2021Fisher}, out-of-time-ordered correlators \cite{Vermersch2019Scrambling,garcia2021quantum}, and topological invariants \cite{Elben2020topological,Cian2020Chern}.
For simplicity, hereafter we mainly adopt the moment function $P_m$ to illustrate the framework, where $O_i=\id_d$ and $\rho_i=\rho$ in Eq.~\eqref{eq:nonlinear}, and assume $\rho$ is an $n$-qubit state, i.e., $\rho\in \mc{H}_d=\mc{H}_2^{\otimes n}$ with $d=2^n$.

\begin{figure}[htbp!]
\centering
\resizebox{8cm}{!}{\includegraphics[scale=0.8]{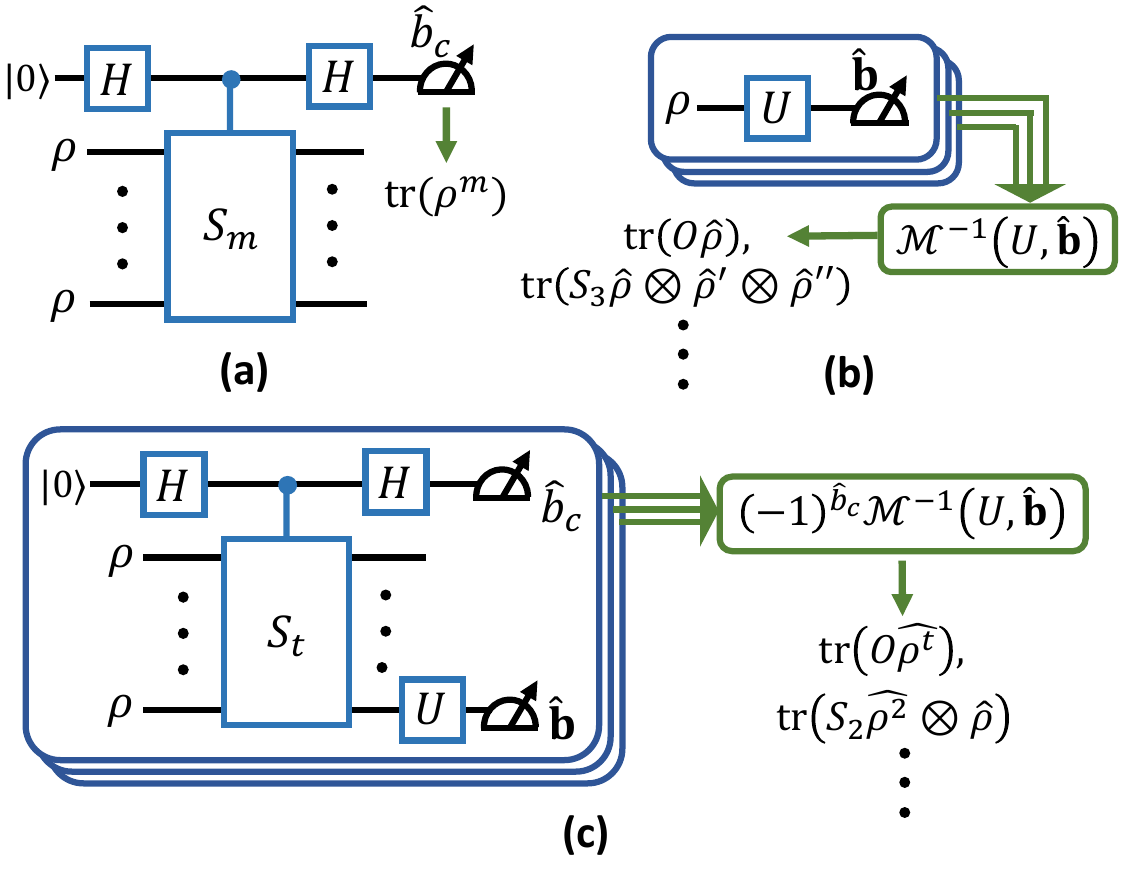}}
\caption{Illustrations of (a) generalized swap test, (b) shadow protocol, and (c) hybrid shadow protocol.  All these three protocols can be divided into two phases, the quantum experiment and the classical postprocessing labeled by green arrows. 
In (c), one measures not only the control qubit but also the last copy of $\rho$ after the random evolution $U$. The information of $U$ and measurement results $\hat{b_c}$ and $\widehat{\mb{b}}$ are used to construct the unbiased estimator $\widehat{\rho^t}$ according to Eq.~\eqref{eq:K1shadow}, and thus estimate $t$-degree function $o_t=\tr(O\rho^t)$ directly. One can also patch them together to estimate higher-degree functions like $P_m=\tr(\rho^m)$ by Eq.~\eqref{eq:HshadowM}. 
}
\label{Fig:original} 
\end{figure}

\emph{Swap test and RM}.---First, let us take a quick review of the swap test and the RM protocol for measuring $P_m$. 
The moment can be written as $P_m=\tr(S_m\rho^{\otimes m})$, with $S_m$ the shift operation on $\mathcal{H}_d^{\otimes m}$ satisfying $S_m\ket{\mb{b}_1,\mb{b}_2\cdots,\mb{b}_m}=\ket{\mb{b}_2,\cdots\mb{b}_{m},\mb{b}_{1}}$, with each $\{\ket{\mb{b}}\}$ the basis of each copy. 
Using the swap test, one can measure $P_m$ efficiently with complete coherent control over multiple copies of the state, as illustrated in Fig.~\ref{Fig:original} (a).
In particular, one initializes a control qubit and prepares the $m$-copy state $\rho^{\otimes m}$, and then conducts the Controlled-shift operation
\begin{equation}\label{}
    \begin{aligned}
        \mathrm{CS}_m=\ket{0}_c\bra{0}\otimes \id_{d^m}+\ket{1}_c\bra{1}\otimes S_m.
    \end{aligned}
\end{equation}
Finally one measures the control qubit to get the value of $P_m$ \cite{Ekert2002Direct,Horodecki2002Method}. The corresponding quantum circuit requires a quantum processor with a total of $N=nm+1$ qubits. Although the shift operation on $m$-copy can be expressed as a product of 2-copy swap operations, this approach would significantly increase the quantum circuit depth. Furthermore, preparing $m$ copies of the state in parallel imposes stringent demands on quantum memories. Therefore, the generalized swap test presents significant challenges for cases where $m\geq 3$.

The RM toolbox \cite{elben2023randomized} was recently developed to ease the experimental challenge mentioned above \cite{Struchalin2021Experimental,zhang2021experimental,Yu2021Fisher}. Compared with swap test, one only needs to control a single-copy state to realize the estimation. For shadow estimation \cite{huang2020predicting}, independent snapshots of the state $\{\widehat{\rho}_{(i)}\}$ 
can be constructed using data conllected in RMs, as shown in Fig.~\ref{Fig:original} (b). 
This is referred to as the shadow set and has the property that the expectation $ \mathbb{E}(\widehat{\rho}_{(i)})=\rho$. Denote the random unitary evolution sampled from some ensemble as $U\in \mc{E}$, and the $Z$-basis measurement result as $\mb{b}=\{b_1b_2\cdots b_n\}\in \{0,1\}^n$, 
\begin{equation}\label{eq:shadowEst}
\begin{aligned}
    \widehat{\rho}=\mc{M}^{-1}\left(U^{\dag}\ketbra{\widehat{\mb{b}}}{\widehat{\mb{b}}}U\right),
\end{aligned}
\end{equation}
where $\widehat{\mb{b}}$ is a random variable with probability $\bra{\mb{b}}U\rho U^{\dag}\ket{\mb{b}}$, and the inverse (classical) postprocessing channel $\mc{M}^{-1}$ is determined by the chosen $\mc{E}$ \cite{huang2020predicting,Hu2022Hamiltonian,hu2022Locally,ohliger2013efficient,bu2022classical}. For instance, $\mc{E}$ can be the global or local Clifford ensemble, denoted as Clifford and Pauli measurements hereafter, respectively. 
After constructing the shadow set, the unbiased estimator of $P_m$ can be constructed as
\begin{equation}\label{eq:shadow}
    \begin{aligned}
       \tr(S_m\widehat{\rho}_{(1)}\otimes \widehat{\rho}_{(2)}\otimes\cdots \widehat{\rho}_{(m)}).
    \end{aligned}
\end{equation}
In principle, any nonlinear function can be obtained with only an $n$-qubit quantum processor using sequential RMs. However, 
the number of measurements needed generally scales exponentially with the qubit number, and the scaling becomes worse for larger $m$ values, such as $m\geq 3$ \cite{huang2020predicting,elben2020mixedstate}.

\emph{Hybrid framework for nonlinear functions}.---By trading off the quantum and classical resources, here we develop a hybrid framework for nonlinear functions. 
All proofs and more detailed discussions are left in Ref.~\cite{supplementary}. 

For a nonlinear function of degree $m$ such as $P_m$, where $m=\sum_{i=1}^L m_i$, we demonstrate that it can be estimated using a quantum processor with only $N=n(\max_i m_i)+1$ qubits. The core idea is to conduct RM on $\rho^t$, where $t=m_i$, by leveraging the coherent operation on $t$ copies of $\rho$. Instead of directly reading the control qubit outcome to obtain $P_t$, as in the swap test, we perform RM on one of the prepared $t$ copies of $\rho$. As the permutation symmetry holds, the RM can be performed on any copy, like the final one shown in Fig.~\ref{Fig:original} (c). By measuring the expectation value of the Pauli-$X$ operator on the control qubit and performing the projective measurement on the final copy, we obtain

\begin{equation}\label{eq:hybrid}
    \begin{aligned}
     &\tr[ X_c\otimes \id_d^{\otimes(t-1)}\otimes \ket{\mb{b}}\bra{\mb{b}}\ U\ \mathrm{CS}_t\ (\ket{+}_c\bra{+}\otimes \rho^{\otimes t})\  \mathrm{CS}_t^{\dag}\ U^{\dag}]\\
     =& \frac1{2}\Big{\{}\tr[X_c\ket{1}_c\bra{0}]\tr[\bra{\mb{b}}U S_t \rho^{\otimes t}\ U^{\dag}\ket{\mb{b}}]\\
     &\ \ \ +\tr[X_c\ket{0}_c\bra{1}]\tr[\bra{\mb{b}}U \rho^{\otimes t}S_t ^\dag \ U^{\dag}\ket{\mb{b}}]  \Big{\}}\\
     =&\bra{\mb{b}}U\rho^tU^{\dag}\ket{\mb{b}}.
    \end{aligned}
\end{equation}
Here the identity operators $\id_d^{\otimes(t-1)}$ are on the first $t-1$ copies, and the projective measurement $\ket{\mb{b}}\bra{\mb{b}}$ and random unitary $U$ is on the final $t$-th copy, as shown in Fig.~\ref{Fig:hybrid}.The result in the last line indicates that one can effectively conduct RM on $\rho^t$. 

\begin{figure}[htbp!]
    \centering
    \includegraphics[scale=0.3]{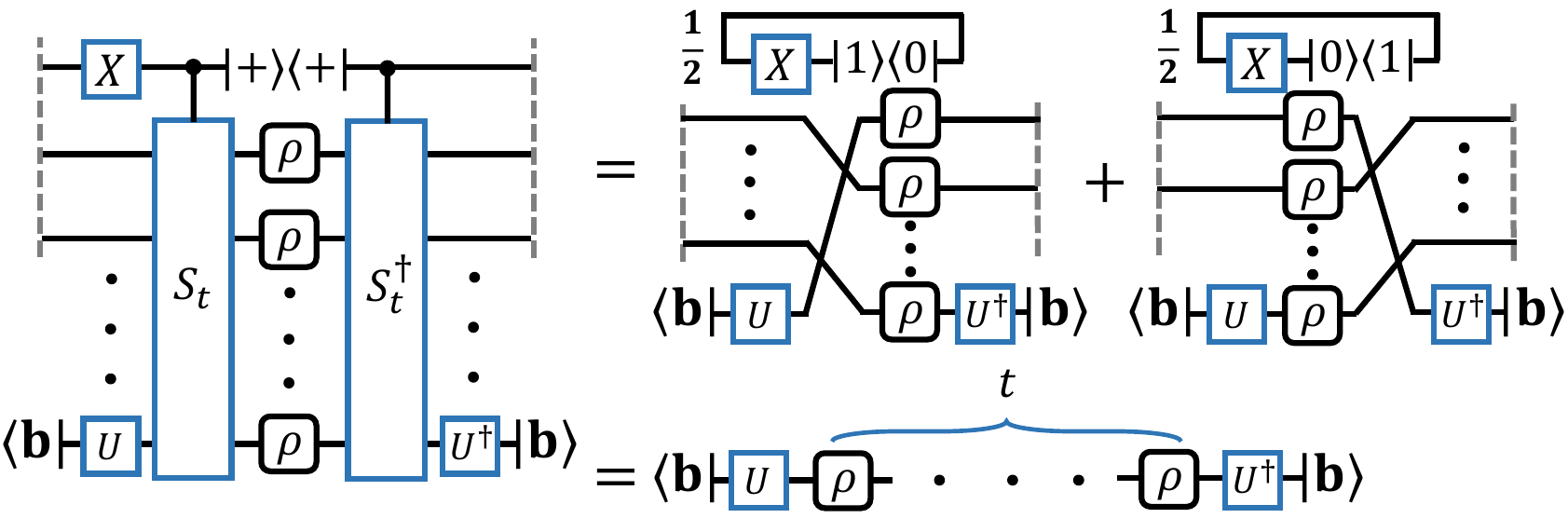}
    \caption{
    The demonstration of Eq.~\eqref{eq:hybrid} by the tensor diagram. The vertical dotted lines denote the periodic boundary condition, i.e., the trace operation. The shift operation $S_t$ and its conjugate $S_t^{\dag}$ are represented by the cyclic permutations of the indices (legs) of the $t$ copies of $\rho$.
    }
    \label{Fig:hybrid}
\end{figure}

The full measurement procedure is listed in Algorithm \ref{algo:tShadow}, which aims to construct the shadow set $\{\widehat{\rho^{t}}_{(i)}\}_{i=1}^M$ of $\rho^t$ from the RM results collected in Fig.~\ref{Fig:original} (c), and these shadow snapshots can be used to estimate more complex nonlinear functions. 

\begin{algorithm}[H]
\caption{Hybrid shadow estimation}\label{algo:tShadow}
\begin{algorithmic}[1]
\Require
$M\times K$ sequentially prepared $\rho^{\otimes t}$ and control qubit initially set as $\ket{0}_c$.
\Ensure
The shadow set $\{\widehat{\rho^{t}}_{(i)}\}_{i=1}^M$. 
\For{$i= 1~\text{\textbf{to}}~M$} 
 \State Randomly choose $U\in\mc{E}$ and record it. 
 \For{$j= 1~\text{\textbf{to}}~K$} 
  \State Conduct the quantum circuit shown in Fig.~\ref{Fig:original} (c).
  \State Measure the control qubit and the final copy of $\rho^{\otimes t}$ in the computational basis $\{\ket{b_c}\}$ and  $\{\ket{\mb{b}}\}$.
  \State Construct the unbiased estimator $\widehat{\rho^t}_{(i)}^{(j)}$ using the results $b_c^{(i,j)}$ and $\mb{b}^{(i,j)}$ by Eq.~\eqref{eq:K1shadow}, where $i$ and $j$ denoting the $j$-th measurement under the $i$-th unitary.
  \EndFor
 \State Average $K$ results under the same unitary to get $\widehat{\rho^t}_{(i)}=\frac1{K}\sum_j\widehat{\rho^t}^{(j)}_{(i)}$.
\EndFor
\State Get the shadow set $\left\{\widehat{\rho^t}_{(1)}, \widehat{\rho^t}_{(2)}, \cdots \widehat{\rho^t}_{(M)}\right\}$, which contains $M$ independent estimators of $\rho^{t}$.
\end{algorithmic}
\end{algorithm}

\begin{theorem}\label{Th:swapshadow}
Suppose one conducts the circuit shown in Fig.~\ref{Fig:original} (c) for once $(M=K=1)$, the unbiased estimator of $\rho^t$ shows
\begin{equation}\label{eq:K1shadow}
    \begin{aligned}
       \widehat{\rho^t}:=(-1)^{\widehat{b}_c}\cdot \mc{M}^{-1}\left(U^{\dag}\ket{\widehat{\mb{b}}}\bra{\widehat{\mb{b}}}U\right),
    \end{aligned}
\end{equation}
such that $\mathbb{E}_{\{U,b_c,\mb{b}\}}\left(\widehat{\rho^t}\right)=\rho^t$. Here $b_c$ and $\mb{b}$ are the measurement results of the control qubit and the final copy from $\rho^{\otimes t}$, respectively; the inverse classical postprocessing $\mc{M}^{-1}$ depends on the random unitary ensemble applied.

Furthermore, to evaluate 
$o_t=\mathrm{tr}(O\rho^{t})$ with $O$ being some observable, the variance shows
\begin{equation}\label{eq:Var}
    \begin{aligned}
      \mathrm{Var}(\widehat{o_t})&=\mathrm{Var}\left[\tr(O\widehat{\rho})\right]+\left[\tr(O\rho)^2-\tr(O\rho^{t})^2\right]\\
       &\leq \|O_0\|^2_{\mathrm{shadow}}+\tr(O\rho)^2
    \end{aligned}
\end{equation}
where $\mathrm{Var}\left[\tr(O\widehat{\rho})\right]$ is the variance of measuring $O$ on the original single-copy shadow snapshot $\widehat{\rho}$, which can be upper bounded by the square of shadow norm $\|O_0\|_{\mathrm{shadow}}$ \cite{huang2020predicting} for the  traceless operator $O_0=O-\tr(O)\id_d/d$.
\end{theorem}

Theorem \ref{Th:swapshadow} is the central result of this work, which gives the unbiased estimator of $\rho^t$ and also relates the statistical variance to the previous single-copy one, i.e., $t=1$ \cite{huang2020predicting}. 
Note that the shadow norm is also related to the chosen random unitary ensemble \cite{huang2020predicting}. 
According to Eq.~\eqref{eq:Var}, the hybrid shadow can dramatically reduce the variance of estimating nonlinear functions compared with the original shadow protocol. Take the Pauli measurement as an example,  for a $k$-local observable $O$, the shadow norm $\|O_0\|^2_{\mathrm{shadow}}\leq 4^k \|O\|_{\infty}^2$ \cite{huang2020predicting} and thus the variance of evaluating $\mathrm{tr}(O\rho^{t})$ is \emph{independent} of the total qubit number $n$ by Eq.~\eqref {eq:Var}. However, the variance of the original shadow protocol shows an exponential scaling with $n$ \cite{huang2020predicting,elben2020mixedstate}.   This point is also clarified by the numerical result in Fig.~\ref{fig:varianceobs} (a). We will discuss this advantage in detail in the application of quantum error mitigation later.

Furthermore, one can repeat the above procedure for all $\widehat{\rho^{m_i}}$ and patch them together to evaluate more complex functions. For $P_m$ the unbiased estimator now shows
\begin{equation}\label{eq:HshadowM}
    \begin{aligned}
       \tr(S_L\widehat{\rho^{m_1}}\otimes \widehat{\rho^{m_2}}\otimes\cdots \widehat{\rho^{m_L}}).
    \end{aligned}
\end{equation}
With the hybrid framework, one can equivalently transform an $m$-degree function in the original shadow protocol in Eq.~\eqref{eq:shadow} to a lower $L$-degree one here. This not only reduces the sampling and postprocessing cost, but also makes other postprocessing strategies \cite{Brydges2019Probing,Elben2019toolbox} available for higher-degree functions. In particular, the postprocessing cost is reduced from $\mathcal{O}(M^m)$ to $\mathcal{O}(M^L)$. We take the moment estimation of $P_3$ as an example to show these advantages.

Besides the functions like $o_m$ and $P_m$ here,  we give the hybrid shadow estimation for more general functions of Eq.~\eqref{eq:nonlinear} in Ref.~\cite{supplementary}, 
by directly extending Theorem \ref{Th:swapshadow}. 

\emph{Application for the moment estimation}.---
We divide $m=3$ to $2+1$ to estimate $P_3=\tr(S_2\rho^2\otimes\rho)$. For $\widehat{\rho^2}$, by following Algorithm \ref{algo:tShadow} ($K=1$, $t=2$) one collects the shadow set
\begin{equation}\label{}
    \begin{aligned}
       \left\{\widehat{\rho^2}_{(1)}, \widehat{\rho^2}_{(2)}, \cdots \widehat{\rho^2}_{(M)}\right\};
    \end{aligned}
\end{equation}
one also collects the shadow set of $\rho$ using the original shadow estimation, 
\begin{equation}\label{}
    \begin{aligned}
       \left\{\widehat{\rho}_{(1)}, \widehat{\rho}_{(2)}, \cdots \widehat{\rho}_{(M')}\right\},
    \end{aligned}
\end{equation}
and then combines two sets to get the estimator of $P_3$. 
\begin{prop}\label{prop:3mEst}
By combing the shadow sets $\{\widehat{\rho^2}_{(i)}\}$ and $\{\widehat{\rho}_{(i')}\}$, one gets the unbiased estimator of $P_3$ as 
\begin{equation}\label{3mEst}
    \begin{aligned}
       \widehat{P_3}=\frac1{MM'}\sum_{i\in[M],i'\in[M']}\tr(S_2\widehat{\rho^{2}}_{(i)}\otimes \widehat{\rho}_{(i')}).
    \end{aligned}
\end{equation}
Suppose one applies the random Pauli measurements, the variance of $\widehat{P_3}$ can be upper bounded by 
\begin{equation}\label{3mEstVar}
    \begin{aligned}
       \mathrm{Var}\left(\widehat{P_3}\right)\leq \tr(\rho^2)[\frac{d+1}{M}+\tr(\rho^2)\frac{d^3}{M^2}],
    \end{aligned}
\end{equation}
with $M=M'$ for simplicity.
\end{prop}

The result of Eq.~\eqref{3mEstVar} is almost the same as that of $P_2$ using the original shadow estimation (Eq.~(D16) in Ref.~\cite{elben2020mixedstate}). This indicates that the hybrid framework reduces the statistical error from a 3-degree problem to a 2-degree one, with coherent access to a limited ($t=2$) quantum hardware.

Moreover, one can adopt another postprocessing protocol \cite{Elben2019toolbox} with the same measurement data collected in Algorithm \ref{algo:tShadow}, and the estimator $\widehat{P_3}'$ is given in Proposition \ref{prop:3mEstZ}. This protocol with Pauli measurements mainly works for 2-degree functions and is proved infeasible for higher-degree ones \cite{singlezhou}. With the hybrid framework here, one can make it feasible for the 3-degree function $P_3$ and also reproduce the same variance scaling of $P_2$ in the original protocol \cite{Zhenhuan2022correlation,Elben2019toolbox}, indicating the advantage again.

\begin{prop}\label{prop:3mEstZ}
Using the RM results, $\{\widehat{b_c}^{(i,j)},\widehat{\mb{b}}^{(i,j)}\}$, collected in Algorithm \ref{algo:tShadow} for $t=2$, one can construct an alternative unbiased estimator of $P_3$ as
\begin{equation}\label{3mEstZ}
\begin{aligned}
  \widehat{P_3}'=\frac{1}{MK(K-1)}\sum_{i\in[M]}\sum_{j\neq j'\in[K] }(-1)^{\widehat{b}_c^{(i,j)}}X_{c\backslash p}(\widehat{\mb{b}}^{(i,j)},\widehat{\mb{b}}^{(i,j')}).
\end{aligned}
\end{equation}
Here the choice of the postprocessing function, $X_c(\mb{b},\mb{b}')=-(-d)^{\delta_{\mb{b},\mb{b}'}}$ or $X_p(\mb{b},\mb{b}')=\prod_{k=1}^n -(-2)^{\delta_{b_k,{b_k}'}}$,
depends on the RM primitives, random Clifford or Pauli measurements. 

For Pauli measurements, the variance is about
\begin{equation}\label{3mEstVarZ}
    \begin{aligned}
       \mathrm{Var}\left(\widehat{P_3}'\right)= \mathcal{O}\left(\frac{d^{\log_2 3}}{MK^2}\right).
    \end{aligned}
\end{equation}
\end{prop}

\begin{figure}
    \centering
    \includegraphics[width=8.3cm]{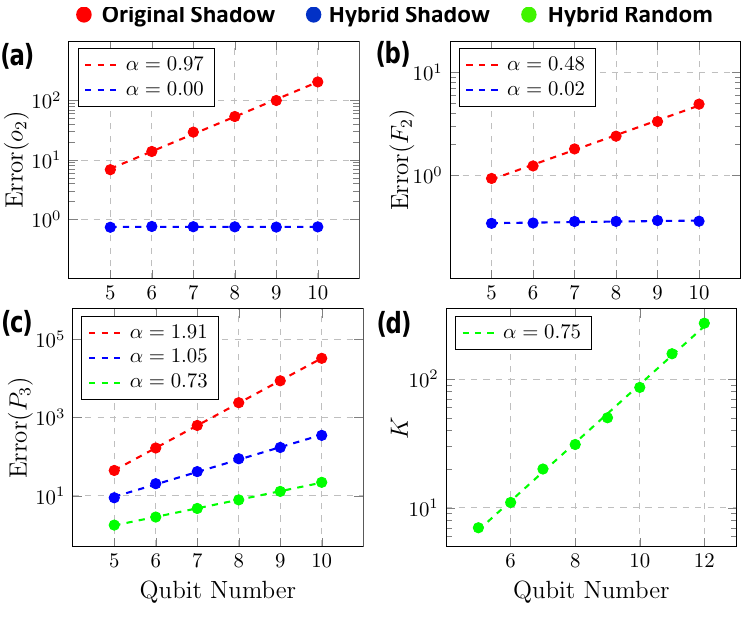}
    \caption{Scaling of errors and the measurement times using OS, HS, and HR protocols. Here the state is the noisy $n$-qubit GHZ state $\rho=0.8\ketbra{\mathrm{GHZ}}{\mathrm{GHZ}}+0.2\mathbb{I}_d/d$. The random Pauli measurements are used for (a), (c), and (d),  and the random Clifford measurements are used for (b). In (a), (b), and (c), we set $M=10$ and $K=1$ for OS and HS, and $M=2$ and $K=5$ for HR. In (a), we estimate local observable $O=\sigma_Z^1\otimes\sigma_Z^2$. In (b), we estimate $F_2=\bra{GHZ}\rho^2\ket{GHZ}$ with $O=\ket{GHZ}\bra{GHZ}$. In (d), we set $M=400$ and find the total measurement is about $MK=200d^{0.75}$ to keep the error less than $0.1$ using HR. }
    \label{fig:varianceobs}
\end{figure}

To complement the above analytical results,
in Fig.~\ref{fig:varianceobs} (c), we numerically study the scaling of the statistical errors for estimating $P_3$, using the estimators $\widehat{P_3}$ in Eq.~\eqref{3mEst} and $\widehat{P_3}'$ in Eq.~\eqref{3mEstZ},
and also the one from the original shadow protocol \cite{elben2020mixedstate}, in the regime $d\gg M(K)$. They  are denoted for short as Hybrid Shadow (HS), Hybrid Random (HR) and Original Shadow (OS), respectively. 
The numerical results, which correspond to the standard variance, are consistent with the analytical ones, and we summarize them together in Table \ref{P3table}.


\begin{table}[htb]
\begin{tabular}
{|c|c|c|c|}
\hline
    $P_3$ & OS & HS & HR  \\
    \hline
     Anal.& $\mathcal{O}(\frac{d^{3}}{M^{1.5}})$ \cite{elben2020mixedstate} & $\mathcal{O}(\frac{d^{1.5}}{M})$ Eq.\eqref{3mEstVar}& $\mathcal{O}(\frac{d^{0.79}}{\sqrt{M}K})$ 
 Eq.\eqref{3mEstVarZ}\\
     \hline
     Numer.& $\mathcal{O}(\frac{d^{1.91}}{M^{1.5}})$ [red]& $\mathcal{O}(\frac{d^{1.05}}{M})$ [blue]& $\mathcal{O}(\frac{d^{0.73}}{\sqrt{M}K})$[green]\\
     \hline
\end{tabular}
\caption{The statistical errors for estimating $P_3$ with different protocols. The numerical results are from Fig.~\ref{fig:varianceobs} (d).
The analytical overestimate of the exponential term on $d$, for two shadow-based protocols OS and HS, due to the fact that the shadow norm is not a tight upper bound for nonlinear functions.}\label{P3table}
\end{table}


Consequently, in practise one needs $M=\mathcal{O}(d^{1.27})$ for OS, $M=\mathcal{O}(d^{1.05})$ for HS, and $K=\mathcal{O}(d^{0.73})$ for HR to make the error less than some constant. It is clear that HS and HR from the hybrid framework both show an advantage compared to OS, and HR is the most efficient one for $P_3$. In Fig.~\ref{fig:varianceobs} (d), we further find the total number of measurements is about $MK=200d^{0.75}$ to make $\mathrm{Error}(P_3)\leq 0.1$, showing great enhancement than OS protocol \cite{elben2020mixedstate,rath2021Fisher}. And the advantage of the hybrid framework is more significant for measuring higher-order moments like $P_4$, and we leave more discussions in Ref.~\cite{supplementary}.

\comments{
\begin{tabular}{|c|c|c|c|}
\hline
    $P_3$ & OS & HS & HR  \\
    \hline
     Anal.& $\mathcal{O}(d^{3}M^{-1.5})$ & $d^{1.5}M^{-1}$ & $d^{\log_2 3/2}M^{-0.5}K^{-1}$\\
     \hline
     Numm.& $d^{1.91}M^{-1.5}$ & $d^{1.05}M^{-1}$ & $d^{0.73}M^{-0.5}K^{-1}$\\
     \hline
\end{tabular}

\begin{tabular}{|c|c|c|c|}
\hline
    $P_3$ & OS & HS & HR  \\
    \hline
     Anal.& $\mathcal{O}(\frac{d^{3}}{M^{1.5}})$ \cite{elben2020mixedstate} & $\mathcal{O}(\frac{d^{1.5}}{M})$ [Eq.~\eqref{3mEstVar}]& $\mathcal{O}(\frac{d^{\log_2 3/2}}{\sqrt{M}K})$ [Eq.~\eqref{3mEstVarZ}]\\
     \hline
     Numm.& $\mathcal{O}(\frac{d^{1.91}}{M^{1.5}})$ [red]& $\mathcal{O}(\frac{d^{1.05}}{M})$ [blue]& $\mathcal{O}(\frac{d^{0.73}}{\sqrt{M}K})$[green]\\
     \hline
\end{tabular}
}


\emph{Application in quantum error mitigation}.---
Recently purified-based methods are proposed \cite{Huggins2021Virtual,Koczor2021Exponential} for quantum error mitigation \cite{temme2017error,Suguru2018Practical,kandala2019error,endo2021hybrid} 
The central task there is to estimate $o_m:=\tr(O\rho^m)$, and use the normalized value $o_m/P_m\rightarrow \tr(O\Psi)$ to approach the target value $\tr(O\Psi)$ with $\Psi$ the noiseless state, which can suppress the error exponentially with the copy-number $m$.  The original protocol is based on the swap test to measure $o_m$, and very recently there have been ones by shadow estimation \cite{Seif2022Shadow,Hu2022Logical}.  There is also an experimental advance for that of $m=2$ \cite{Brien2022purification},  however,  similarly it is still challenging to extend both approaches to $m\geq 3$. Here we show that the hybrid framework gives various advantages on estimating $o_m$.

Suppose one has access to the coherent operation on $m$-copy quantum states. By only adopting the classical postprocessing, the shadow set collected in Algorithm \ref{algo:tShadow} for $t=m$ can be reused for many different observables $\{O_i\}$, say totally $T$ ones, and the estimation error scales like $\mathcal{O}(\log(T))$ using the median-of-mean technique \cite{huang2020predicting}. However, in principle the swap test approach should adopt different quantum circuits for different observables \cite{Koczor2021Exponential}, and also the error would scale linearly $\mathcal{O}(T)$. This advantage is significant for quantum chemistry simulation with the polynomial number of terms in Hamiltonian \cite{McArdle2020chemistry,google2020Hartree,Brien2022purification}.On the other hand, compared to OS protocol \cite{Seif2022Shadow}, HS significantly reduces the statistical variance and thus the sampling cost. For instance, suppose one applies Pauli measurements with OS protocol to estimate $o_m$ for a local observable $O$. Since $o_m=\tr(O_m\rho^{\otimes m})$ and $O_m:=\frac{1}{2}(OS_m+S_m^\dagger O)$ the symmetrized observable is actually a global one with locality about $mn$ \cite{huang2020predicting,elben2020mixedstate}, and thus the shadow norm scales exponentially with $n$. While in HS protocol, the variance is independent of the qubit number $n$ by Eq.~\eqref{eq:Var}. See Fig.~\ref{fig:varianceobs} (a) for this exponential advantage as $t=2$, 
and similar advantage also appears in the fidelity estimation using the Clifford measurements as shown in Fig.~\ref{fig:varianceobs} (b), with more discussions left in Ref.~\cite{supplementary}.

In reality, one generally can not implement Algorithm \ref{algo:tShadow} directly for $t=m$ when $m\ge 3$ due to the hardware limitation. Like the moment estimation, one can alternatively patch low-degree snapshots to measure higher-degree $o_m$. For instance, when $m=3$ and $t=2$, similar as Eq.~\eqref{3mEst}, one can construct the estimator $\widehat{o_3}=(MM')^{-1}\sum_{i,i'}\mathrm{tr}(O_2 \widehat{\rho^{2}}_{(i)}\otimes \widehat{\rho}_{(i')})$,
\comments{
\begin{equation}\label{3mEstO}
    \begin{aligned}
       \widehat{o_3}=\frac1{MM'}\sum_{i\in[M],i'\in[M']}\tr(O_2 \widehat{\rho^{2}}_{(i)}\otimes \widehat{\rho}_{(i')})
    \end{aligned}
\end{equation}
}
and the error also scales logarithmically with the total number of observebles.
In addition, one can construct an alternative estimator $\widehat{o_3}'$ by following Eq.~\eqref{3mEstZ}. Similar as in the moment-estimation of $P_3$, both estimators show statistical error advantages compared to the original shadow. The details of the unbiased estimator for general observable $O$ and  numerical results are left in Ref.~\cite{supplementary}.


\comments{
\begin{figure}
    \centering
    \includegraphics[scale=0.45]{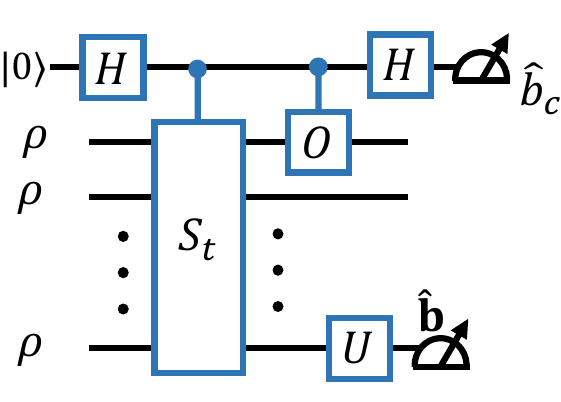}
    \caption{
    Compared with the circuit in Fig.~\ref{Fig:original}, one additionally performs a Controlled-$O$ operation. Here the observable $O$ is also assumed to be a unitary, for instance, the Pauli observable. The quantum circuit and hybrid measurement procedure for more general observable $O$ is given in Ref.~\cite{supplementary}.
    }
    \label{Fig:miti}
\end{figure}
}

\emph{Concluding remarks}.---The hybrid framework proposed here utilizes the partial coherent power of quantum devices and can act as a subroutine for many quantum information tasks, for instance, measuring entanglement \cite{ketterer2019characterizing,neven2021symmetry,Yu2021Optimal,liu2022detecting}, characterizing quantum chaos \cite{Joshi2022Probing,McGinley2022scrambling},  
and constructing quantum algorithms \cite{Xiao2021Hybrid,Lubasch2020nonlinear,Yamamoto2021metrology}.
Note that the framework reduces to previous RM protocols as $t=1$ in Algorithm \ref{algo:tShadow}, and the advantage essentially comes from the coherent processing the few-copy state. So it is intriguing to build the ultimate result of this replica advantage \cite{chen2022exponential,huang2022Science} considering the limited quantum memory. Moreover, it is also appealing to extend the current framework to quantum channel \cite{Chen2021Robust,Helsen2021Estimating,Kunjummen2021process,Levy2021Process} and boson or fermion systems \cite{Elben2018Random,Zhao2021Fermionic}.

\emph{Acknowledgements}---
Y.Z. is supported by National Natural Science Foundation of China(NSFC) Grant No. 12205048 and the start-up funding of Fudan University.  Z.L. is supported by NSFC Grants No.~11875173 and No.~12174216 and the National Key Research and Development Program of China Grants No.~2019QY0702 and No.~2017YFA0303903.


%

\onecolumngrid
\newpage

\renewcommand{\addcontentsline}{\oldacl}
\renewcommand{\tocname}{Appendix Contents}
\tableofcontents


\begin{appendix}
In this Supplementary Material,  we give the proof and more discussions and generalizations of the results in main text. In Sec.~\ref{ap:hybrid}, we mainly prove the central result Theorem 1 in main text. In Sec.~\ref{ap:frame}, we further generalize Theorem 1, and show the hybrid estimation framework for general nonlinear functions defined in Eq.~(1) in main text. Sec.~\ref{sec:ap:Pm} considers the application to the state-moment estimation. Sec.~\ref{ap:sec:miti} discusses the application to the virtual distillation, and the hybrid measurement procedure for the general Hermitian operator $O$ is also given.

\section{RM hybrid with swap test}\label{ap:hybrid}

\subsection{Proof of Theorem 1 $(M=1,K=1)$}\label{ap:Th1}
Here we prove Theorem 1 in main text, related to the single-shot unbiased estimator of $\rho^t$, that is, the $M=1, K=1$ case in Algorithm 1 . And then we extend the result for general $M$ and $K$ in the following subsection.

First, we prove the estimator in Eq.~(6).
\begin{proof}
We take the expectation value of the estimator on the random unitary $U$ and the projective measurement result $\mb{b},b_c$ to get
\begin{equation}\label{eq:ap:est}
    \begin{aligned}
       \mathbb{E}_{\{U,\mb{b},b_c\}}\ \widehat{\rho^t}& =\mathbb{E}_{\{U,\mb{b},b_c\}}\ (-1)^{\widehat{b}_c}\cdot \mc{M}^{-1}\left(U^{\dag}\ket{\widehat{\mb{b}}}\bra{\widehat{\mb{b}}}U\right)\\
       &=\sum_{U,\mb{b},b_c} \mathrm{Pr}(U,\mb{b},b_c)\ (-1)^{b_c}\mc{M}^{-1}\left(U^{\dag}\ket{\mb{b}}\bra{\mb{b}}U\right)\\
       &=\sum_{U} \mathrm{Pr}(U) \sum_{\mb{b},b_c}\mathrm{Pr}(\mb{b},b_c|U)\  (-1)^{b_c}\mc{M}^{-1}\left(U^{\dag}\ket{\mb{b}}\bra{\mb{b}}U\right)
    \end{aligned}
\end{equation}
with $\mathrm{Pr}(U,\mb{b},b_c)$ the joint probability distribution for these random variables.
Given $U$ and $\mb{b}$, one can first sum the index $b_c$ as 
\begin{equation}\label{eq:ap:bc}
    \begin{aligned}
       \sum_{b_c}\mathrm{Pr}(\mb{b},b_c|U)\  (-1)^{b_c}=&\tr[ \bra{+}_c \bra{\mb{b}}U\  \mathrm{CS}_t\left(\ket{+}_c\bra{+}\otimes \rho^{\otimes t}\right)  \mathrm{CS}_t^{\dag}\ U^{\dag}\ket{\mb{b}}\ket{+}_c]\\
       &-\tr[ \bra{-}_c \bra{\mb{b}}U\  \mathrm{CS}_t \left(\ket{+}_c\bra{+}\otimes \rho^{\otimes t}\right)  \mathrm{CS}_t^{\dag}\ U^{\dag}\ket{\mb{b}}\ket{-}_c]\\
     =&\tr[ X_c \bra{\mb{b}}U\  \mathrm{CS}_t\left(\ket{+}_c\bra{+}\otimes \rho^{\otimes t}\right)  \mathrm{CS}_t^{\dag}\ U^{\dag}\ket{\mb{b}}]\\
     =&\bra{\mb{b}}U\rho^tU^{\dag}\ket{\mb{b}}.
    \end{aligned}
\end{equation}
where in the first line we insert the probability $\mathrm{Pr}(\mb{b},0/1|U)$ accounting for the measurement result $0/1$ of the control qubit in the $X$-basis, and the final line is by Eq.~(5) in main text. Inserting the result of Eq.~\eqref{eq:ap:bc} into the final line of Eq.~\eqref{eq:ap:est}, one gets
\begin{equation}\label{}
    \begin{aligned}
       \mathbb{E}_{\{U,\mb{b},b_c\}}\ \widehat{\rho^t}=\sum_{U} \sum_{\mb{b}} \mathrm{Pr}(U) \bra{\mb{b}}U\rho^t U^{\dag}\ket{\mb{b}}  \mc{M}^{-1}\left(U^{\dag}\ket{\mb{b}}\bra{\mb{b}}U\right)=\rho^t.
    \end{aligned}
\end{equation}
The equality holds due to the following fact: by taking $\sigma=\rho^t$, the equation is just the definition of the unbiased estimator for  $\sigma$ in the original shadow protocol \cite{huang2020predicting}.
\end{proof}


Second, we prove the result of variance in Eq.~(7).

\begin{figure}[hbt]
\centering
\resizebox{12cm}{!}{\includegraphics[scale=1]{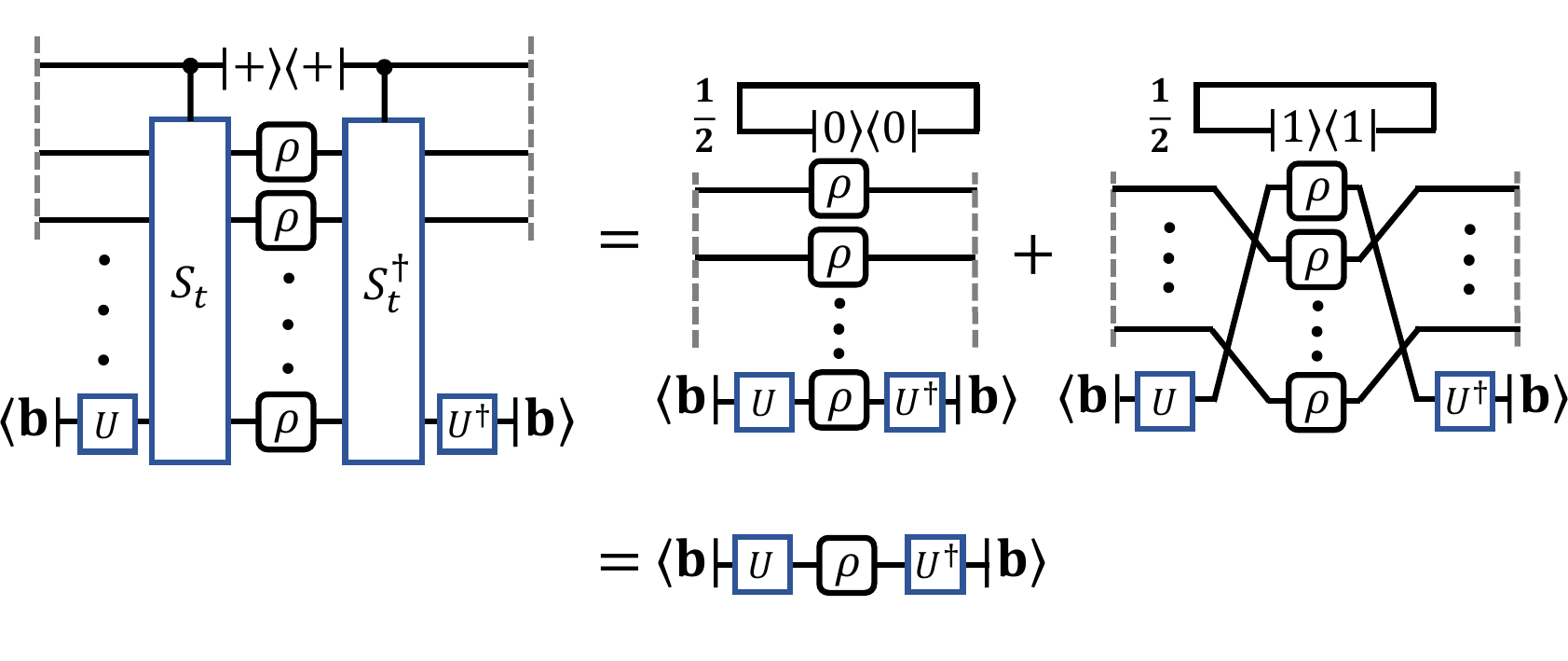}}
\caption{The demonstration of the derivation of Eq.~\eqref{ap:eq:minus} by tensor diagram. The vertical dotted lines denote the periodic boundary condition, i,e., the trace operation. The shift operation $S_t$ and also its conjugate $S_t^{\dag}$ are represented by the permutation of the indices of the $t$-copy. On account of the identity $\id_c$ (partial trace) on the control qubit, there only two terms corresponding to $\ket{0}_c\bra{0}$ and $\ket{1}_c\bra{1}$ from $\ket{+}_c\bra{+}$ survive, and two terms both return $\bra{\mb{b}}U\rho U^{\dag}\ket{\mb{b}}$, i.e., conducting the RM on $\rho$. }\label{Fig:diagonal} 
\end{figure}

\begin{proof}
By definition, the variance of $\tr(O\widehat{\rho^{t}})$ is
\begin{equation}
    \begin{aligned}
      \mathrm{Var}\left[\tr(O\widehat{\rho^{t}})\right]=\mathbb{E}\left[ \tr(O\widehat{\rho^{t}})^2\right]- \tr(O\rho^{t})^2
    \end{aligned}
\end{equation}
with the first term being
\begin{equation}\label{ap:eq:Var}
    \begin{aligned}
       \mathbb{E}\left[ \tr(O\widehat{\rho^{t}})^2\right]&=\mathbb{E}\ \tr\left[O(-1)^{\widehat{b}_c}\cdot \mc{M}^{-1}\left(U^{\dag}\ket{\widehat{\mb{b}}}\bra{\widehat{\mb{b}}}U\right)\right]^2\\
       &=\mathbb{E}\ (-1)^{2\widehat{b}_c} \bra{\widehat{\mb{b}}}U\mc{M}^{-1}(O)U^{\dag}\ket{\widehat{\mb{b}}}^2\\
       &=\sum_{U} \mathrm{Pr}(U) \sum_{\mb{b},b_c}\mathrm{Pr}(\mb{b},b_c|U)\  \bra{\mb{b}}U\mc{M}^{-1}(O)U^{\dag}\ket{\mb{b}}^2,
    \end{aligned}
\end{equation}
where we insert the definition of $\widehat{\rho^{t}}$ of Eq.~(6) in main text in the first line and use the self-adjoint property of $\mc{M}^{-1}$ to equivalently act it on $O$ in the second line. For simplicity, we omit the subscript $\{U,\mb{b},b_c\}$ in the expectation. The interesting point is that the sign with respective to $b_c$ disappears due to the square. By first summing the index $b_c$, which implies taking a partial trace on the control qubit, one has
\begin{equation}\label{ap:eq:minus}
    \begin{aligned}
       \sum_{b_c} \mathrm{Pr}(\mb{b},b_c|U)=&\tr[ \bra{+}_c \bra{\mb{b}}U\  \mathrm{CS}_t\left(\ket{+}_c\bra{+}\otimes \rho^{\otimes t}\right)  \mathrm{CS}_t^{\dag}\ U^{\dag}\ket{\mb{b}}\ket{+}_c]+\tr[ \bra{-}_c \bra{\mb{b}}U\  \mathrm{CS}_t\left(\ket{+}_c\bra{+}\otimes \rho^{\otimes t} \right) \mathrm{CS}_t^{\dag}\ U^{\dag}\ket{\mb{b}}\ket{-}_c]\\
     =&\tr[ \id_c \bra{\mb{b}}U\  \mathrm{CS}_t\left(\ket{+}_c\bra{+}\otimes \rho^{\otimes t}\right)  \mathrm{CS}_t^{\dag}\ U^{\dag}\ket{\mb{b}}]\\
     =&\frac1{2}\tr[ \bra{\mb{b}}U\  \mathrm{CS}_t\left(\ket{0}_c\bra{0}\otimes \rho^{\otimes t}\right)  \mathrm{CS}_t^{\dag}\ U^{\dag}\ket{\mb{b}}]+\frac1{2}\tr[ \bra{\mb{b}}U\  \mathrm{CS}_t\left(\ket{1}_c\bra{1}\otimes \rho^{\otimes t}\right)  \mathrm{CS}_t^{\dag}\ U^{\dag}\ket{\mb{b}}]\\
     =&\frac1{2}\tr[ \bra{\mb{b}}U\ \left(\ket{0}_c\bra{0}\otimes \rho^{\otimes t} \right) U^{\dag}\ket{\mb{b}}]+\frac1{2}\tr[ \bra{\mb{b}}U\left(  \ket{1}_c\bra{1}\otimes S_t\rho^{\otimes t} S_t^{\dag}\right) U^{\dag}\ket{\mb{b}}]\\
     =&\bra{\mb{b}}U\rho U^{\dag}\ket{\mb{b}},
    \end{aligned}
\end{equation}
where in the third line only two terms corresponding to the diagonal terms $\ket{0}_c\bra{0}$ and $\ket{1}_c\bra{1}$ from $\ket{+}_c\bra{+}$ survive due to the partial trace, in the last line we use $S_t\rho^{\otimes t} S_t^{\dag}=\rho^{\otimes t}$. See Fig.~\ref{Fig:diagonal} for a diagram representation of the derivation.

By inserting Eq.~\eqref{ap:eq:minus} into the last line of Eq.~\eqref{ap:eq:Var}, one gets
\begin{equation}\label{ap:eq:Srhot}
    \begin{aligned}
       \mathbb{E}\left[ \tr(O\widehat{\rho^{t}})^2\right]
       &=\sum_{U} \mathrm{Pr}(U) \sum_{\mb{b},b_c}\mathrm{Pr}(\mb{b},b_c|U)\  \bra{\mb{b}}U\mc{M}^{-1}(O)U^{\dag}\ket{\mb{b}}^2\\
       &=\sum_{U} \sum_{\mb{b}} \mathrm{Pr}(U) \bra{\mb{b}}U\rho U^{\dag}\ket{\mb{b}} \bra{\mb{b}}U\mc{M}^{-1}(O)U^{\dag}\ket{\mb{b}}^2\\
       &=\sum_{U} \sum_{\mb{b}} \mathrm{Pr}(U) \mathrm{Pr}(\mb{b}|U)_{(\rho)} \tr[O\ \mc{M}^{-1}(U^{\dag}\ket{\mb{b}}\bra{\mb{b}}U)]^2\\
       &=\mathbb{E}\left[ \tr(O\widehat{\rho})^2\right].
        \end{aligned}
\end{equation}
The last line shows that the expectation value just equals the one in the original shadow protocol on a single-copy of the state $\rho$, no matter what value $t$ takes.

As a result,
\begin{equation}\label{ap:eq:Shnorm1}
    \begin{aligned}
      \mathrm{Var}\left[\tr(O\widehat{\rho^{t}})\right]&=\mathbb{E}\left[\tr(O\widehat{\rho})^2\right]-\tr(O\rho^{t})^2\\
    \end{aligned}
\end{equation}
Note that $\mathbb{E}\ \tr(O\widehat{\rho})^2$ has already been analysed by maximizing it for all possible $\rho$, which serves as an upper bound denoted by the square of the shadow norm $\|O\|_{\mathrm{shadow}}$ \cite{huang2020predicting}. 
Consequently, by ignoring the second term $\tr(O\rho^{t})^2$, one has the upper bound for the variance for any $t$ as
\begin{equation}\label{ap:eq:Shnorm2}
    \begin{aligned}
      \mathrm{Var}\left[\tr(O\widehat{\rho^{t}})\right]\leq \|O\|^2_{\mathrm{shadow}}.
    \end{aligned}
\end{equation}
In addition, since shifting the operator to its traceless part $O_0=O-\tr(O)\id/d$ does not change the variance for $\mathrm{Var}[\tr(O\widehat{\rho})]$, one could further bound Eq.~\eqref{ap:eq:Shnorm1} as
\begin{equation}\label{ap:eq:Shnorm3}
    \begin{aligned}
    \mathrm{Var}\left[\tr(O\widehat{\rho^{t}})\right]
      &=\mathrm{Var}[\tr(O\widehat{\rho})]+\tr(O\rho)^2-\tr(O\rho^{t})^2\\
      &=\mathrm{Var}\left[\tr(O_0\widehat{\rho})\right]+\tr(O\rho)^2-\tr(O\rho^{t})^2\\
      &=\mathbb{E}\ \tr(O_0\widehat{\rho})^2-\tr(O_0\rho)^2+\tr(O\rho)^2-\tr(O\rho^{t})^2\\
      &\leq \|O_0\|^2_{\mathrm{shadow}}+\tr(O\rho)^2.
    \end{aligned}
\end{equation}
\end{proof}

We remark that this reduction of variance of $\widehat{\rho^{t}}$ to the $t=1$ case is since one has the ability to coherently control $t$-copy of the state.
The detailed expression of $\|O\|_{\mathrm{shadow}}$ actually depends on the RM primitives. We show in latter section that by using Eq.~\eqref{ap:eq:Shnorm3} and the results of shadow norm, one can further give upper bounds for variance of more complex estimators, which are combined from a few of $\widehat{\rho^{t}}$.

\subsection{Unbiased estimator for general $M$ and $K$}
In Algorithm 1  in main text, one samples $M$ independent $U\in \mc{H}_d$ from some unitary ensemble $\mc{E}$, and for a given $U$, i.e., the measurement setting, one repeats the quantum circuit in Fig.~1 (c) in main text for $K$ shots to collect the measurement results $b_c^{(i,j)}$ and $\mb{b}^{(i,j)}$. Here the superscript $i$ labels the measurement setting, and $j$ labels the shot. We list the estimator in the following proposition.
\begin{prop}\label{ap:prop:MK}
Following the procedure in Algorithm 1  in main text, one can collect $M$ independent estimators of $\rho^{t}$, i.e., $\left\{\widehat{\rho^t}_{(1)}, \widehat{\rho^t}_{(2)}, \cdots \widehat{\rho^t}_{(M)}\right\}$ with 
\begin{equation}\label{}
    \begin{aligned}
\widehat{\rho^t}_{(i)}=\frac1{K}\sum_j\widehat{\rho^t}^{(j)}_{(i)}.
    \end{aligned}
\end{equation}
Here each $\widehat{\rho^t}^{(j)}_{(i)}$ is constructed according to Eq.~(6) in Theorem 1  by the $j$-th measurement result $b_c^{(i,j)}$ and $\mb{b}^{(i,j)}$ under $i$-th measurement setting. To estimate the value of $\tr(O\rho^t)$, one can further average total $M$ rounds to get the unbiased estimator 
\begin{equation}\label{ap:eq:O}
    \begin{aligned}
\tr(O \frac{\sum_i \widehat{\rho^t}_{(i)}}{M}).
    \end{aligned}
\end{equation}

\end{prop}

\begin{proof}
In Appendix~\ref{ap:Th1}, we have proved the result in Theorem 1  that $\mathbb{E}_{\{U,\mb{b},b_c\}}\ \widehat{\rho^t}^{(j)}_{(i)}=\rho^t$ for the single-shot case, that is, for any $i,j$. Thus the estimators here directly hold since they only involve average operations.
\end{proof}
To evaluate the variance of the estimator in Eq.~\eqref{ap:eq:O}. Note that $\left\{\widehat{\rho^t}_{(1)}, \widehat{\rho^t}_{(2)}, \cdots \widehat{\rho^t}_{(M)}\right\}$ are 
independent and identically distributed (i.i.d.) random variables. As a result,
\begin{equation}\label{ap:eq:MKvar}
    \begin{aligned}
\mathrm{Var}\left[\tr(O \sum_i \widehat{\rho^t}_{(i)}/M)\right]=\frac1{M}\mathrm{Var}\left[\tr(O  \widehat{\rho^t}_{(i_0)})\right]
    \end{aligned}
\end{equation}
for any $i_0$. 

As a result, for the general $M$ and $K=1$ (which is the case considered in the original shadow estimation \cite{huang2020predicting}), one can directly use the single-shot result in Eq.~(7) in main text to bound the variance by dividing it by $M$. We leave the variance for general $K$ to further study. Hereafter, our discussion will mainly focus on the single-shot scenario with shadow-type postprocessing.

\section{Hybrid framework for general nonlinear functions}\label{ap:frame}
In this section, we show that our hybrid framework can work for more general nonlinear functions in the following form,
\begin{equation}\label{ap:Fm}
    \begin{aligned}
       F_m(\{\rho_i\},\{O_i\})=\tr(\rho_1 O_1 \rho_2O_2 \cdots \rho_mO_m ),
    \end{aligned}
\end{equation} 
for quantum states $\{\rho_1,\rho_2\cdots\rho_m\}$ with some observables $\{O_1,O_2\cdots O_m\}$. We call it a $m$-degree function. Here for simplicity we assume that $O_i$ is both Hermitian and unitary (e.g., Pauli operators of practical interest) in the following discussion.

\subsection{Swap test and RM for $F_m$}\label{ap:secSwap}
Like in main text, let us first briefly recast the generalized swap test and RM to get $F_m(\{\rho_i\},\{O_i\})$, respectively.

\begin{figure}[hbt]
\centering
\resizebox{6cm}{!}{\includegraphics[scale=0.6]{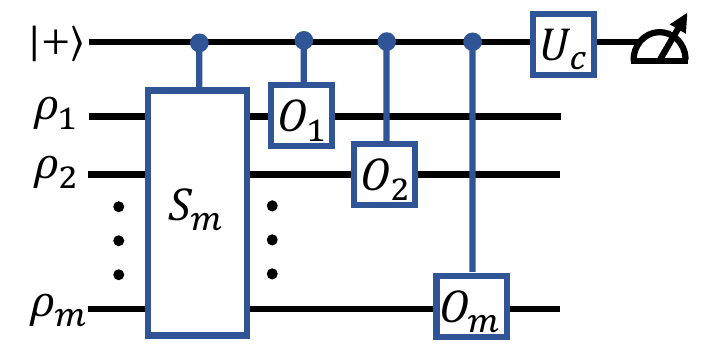}}
\caption{The quantum circuit of the swap test to measure $F_m(\{\rho_i\},\{O_i\})$ in Eq.~\eqref{ap:Fm}. Compared to the circuit in Fig.~1 (a), one additionally adds $\mathrm{\mathrm{CO_i}}$ operations. The unitary $U_c$ on the control qubit before the projective measurement determines the measurement basis. For the $X$-basis, $U_c=H$ as in main text; for the $Y$-basis, $U_c=e^{-i\frac{\pi}{4}Z}H$.}\label{Fig:FmSwap} 
\end{figure}

In the swap test, one prepares the quantum state $\bigotimes_{i=1}^m \rho_i$ in parallel and initializes the control qubit, and further operates Controlled-$O_i$ gate, denoted as $\mathrm{\mathrm{CO_i}}$, for each $\rho_i$ after the Controlled-shift operation $ \mathrm{CS}_m$ in the quantum circuit, as shown in Fig.~\ref{Fig:FmSwap}.  Measuring the expectation value of control qubit using Pauli-$X$ operator, one gets
\begin{equation}\label{ap:eq:mswap}
    \begin{aligned}
     &\tr[ X_c\ \prod_i \mathrm{CO_i}\  \mathrm{CS}_m\left(\ket{+}_c\bra{+}\otimes\bigotimes_{i=1}^m \rho_i\right)  \mathrm{CS}_m^{\dag}\ \prod_i \mathrm{CO_i}]\\
     =& \frac1{2}\tr[X_c\ket{1}_c\bra{0}]*\tr[\left(\bigotimes_i O_i\right) S_m \left(\bigotimes_{i=1}^m \rho_i\right)] +\frac1{2} \tr[X_c\ket{0}_c\bra{1}]*\tr[\left(\bigotimes_{i=1}^m \rho_i\right) S_t^\dag \left(\bigotimes_i O_i\right)]\\
     =&\frac1{2}\tr[\rho_1O_1 \rho_2O_2 \cdots \rho_mO_m ] +\frac1{2}\tr[O_m\rho_m \cdots O_2\rho_2  O_1 \rho_1]=\frac1{2}(F_m+F_m^*),
    \end{aligned}
\end{equation}
which is the average of $F_m$ and its conjugation $F_m^*$.

By additionally introducing the Pauli-$Y$ measurement on the control qubit, one has
\begin{equation}\label{ap:eq:mswapY}
    \begin{aligned}
     &\tr[ Y_c\ \prod_i \mathrm{CO_i}\  \mathrm{CS}_m\left( \ket{+}_c\bra{+}\bigotimes_{i=1}^m \rho_i\right)  \mathrm{CS}_m^{\dag}\ \prod_i \mathrm{CO_i}]\\
     =& \frac1{2}\tr[Y_c\ket{1}_c\bra{0}]*\tr[\left(\bigotimes_i O_i\right) S_m \left(\bigotimes_{i=1}^m \rho_i\right)] +\frac1{2} \tr[Y_c\ket{0}_c\bra{1}]*\tr[\left(\bigotimes_{i=1}^m \rho_i\right) S_m^\dag \left(\bigotimes_i O_i\right)]=\frac{-i}{2}(F_m-F_m^*).
    \end{aligned}
\end{equation}
By repeating the measurements and combining the Pauli-$X$ and $Y$ measurement results, one can obtain $F_m$.

For the RM with shadow estimation, one can construct the following unbiased estimator,
\begin{equation}\label{ap:eq:shadow}
    \begin{aligned}
       \tr(\bigotimes_{i=1}^m O_iS_m\ \widehat{\rho_1}\otimes \widehat{\rho_2}\otimes\cdots \widehat{\rho_m})=\tr(\widehat{\rho_1}O_1\widehat{\rho_2}O_2\cdots \widehat{\rho_m}O_m),
    \end{aligned}
\end{equation}
where $\widehat{\rho_i}$ is the snapshot for $\rho_i$ based on the shadow protocol as shown in Eq.~(3) in main text. And one can further increase the shadow size to reduce the statistical fluctuation of the estimator.

As mentioned in main text, both methods are challenging to measure $F_m$ for a large degree $m$. For the swap test, it is hard to prepare so many copies of the state and realize the Controlled-shift gate among them; for the RM, the measurement times and the postprocessing budgets are quite annoying. The following hybrid framework can trade off them, and thus shows advantages.

\subsection{The hybrid framework}\label{ap:secHybrid}
We first define the $t$-degree multiplication with $t<m$, denoted by $\sigma$ as
\begin{equation}\label{ap:eq:sigma}
    \begin{aligned}
       \sigma=\rho_1O_1\rho_2O_2\cdots \rho_{t-1}O_{t-1}\rho_t,
    \end{aligned}
\end{equation}
and develop a method to perform the RM on it with the help of the partial coherent power of Controlled-shift operation on just t-copy.

\begin{figure}[hbt]
\centering
\resizebox{6cm}{!}{\includegraphics[scale=0.6]{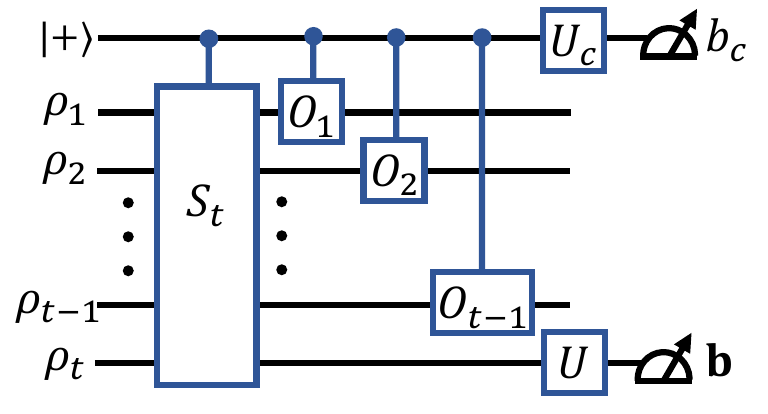}}
\caption{The quantum circuit to effectively conduct the RM on $\sigma$ in Eq.~\eqref{ap:eq:sigma}. Here $U$ is the random unitary applied on the $t$-th state $\rho_t$. $U_c$ on the control qubit is also random, determined by the binary random variable $c$, which corresponds to the $X(Y)$-basis measurement.}\label{Fig:FmHybrid} 
\end{figure}

The quantum circuit is similar to that for the swap test method shown in Fig.~\ref{Fig:FmSwap}, but we further conduct RM on $t$-th state $\rho_t$, as shown in Fig.~\ref{Fig:FmHybrid}. The measurement procedure follows Algorithm 1 . Note here besides the RM on $\rho_t$, we also conduct the Pauli-$X$ or $Y$ measurement on the control qubit, which in some sense is also random. For simplicity we use a random variable $\widehat{c}$ to denote this choice, and $c=0/1$ labels the Pauli-$X$/$Y$ measurement with probability $\frac{1}{2}$. 
The reason for this choice is to separate $\sigma$ and its conjugate $\sigma^{\dag}$, similar to the swap test method discussed in Sec.~\ref{ap:secSwap}. 

For a single-shot of the measurement, one can collect the measurement result $b_{c}$ and $\mb{b}$. Here $b_{c}$ denotes the result for the control qubit, and $c=0/1$ for the chosen basis. We give the unbiased estimator of $\sigma$ as follows.
\begin{prop}\label{}
Suppose one conducts the measurement procedure shown in Fig.~\ref{Fig:FmHybrid} for once, the unbiased estimator of $\sigma$ shows
\begin{equation}\label{ap:esttsigma}
    \begin{aligned}
       \widehat{\sigma}:=2(-1)^{\widehat{b}_{\widehat{c}}}\ i^{\widehat{c}}\cdot \mc{M}^{-1}\left(U^{\dag}\ket{\widehat{\mb{b}}}\bra{\widehat{\mb{b}}}U\right),
    \end{aligned}
\end{equation}
such that $\mathbb{E}_{\{U,c,\mb{b},b_c\}}\left(\widehat{\sigma}\right)=\sigma$. Here $b_c$ and $\mb{b}$ are the measurement results of the control qubit and the final copy $\rho_t$ respectively; $c$ and $U$ denote the measurement basis of the control qubit and the random unitary, respectively; the inverse operation $\mc{M}^{-1}$ depends on the random unitary ensemble one applied. 

To evaluate the value of some observable $O$ on $\sigma$, the variance of the estimator $\widehat{o}=\mathrm{tr}(O\widehat{\sigma})$ shows
\begin{equation}\label{eq:VarSig}
    \begin{aligned}
 \mathrm{Var}\left[\tr(O\widehat{\sigma})\right]
      &=4\left\{\mathrm{Var}\left[\tr(O_0\widehat{\rho'})\right]+\tr(O\rho')^2\right\}-|\tr(O\sigma)|^2\\
      &\leq 4\|O_0\|^2_{\mathrm{shadow}}+4\tr(O\rho')^2,
    \end{aligned}
\end{equation}
where $\mathrm{Var}\left[\tr(O_0\widehat{\rho'})\right]$ is the variance of measuring $O$ on the shadow snaptshot of the state $\rho'=(\rho_1+\rho_t)/2$, which can be upper bounded by the square of shadow norm $\|O_0\|_{\mathrm{shadow}}$ for the traceless part $O_0=O-\tr(O)\id_d/d$. \cite{huang2020predicting}.
\end{prop}

The proof is similar to that of Theorem 1  in Sec.~\ref{ap:Th1} by further considering the expectation on the index $c$.
\begin{proof}
By definition, the expectation value of the estimator shows
\begin{equation}\label{ap:eq:sigmaEst}
    \begin{aligned}
       \mathbb{E}_{\{U,c,\mb{b},b_c\}}\ \widehat{\sigma}& =\mathbb{E}_{\{U,c,\mb{b},b_c\}}\ 2(-1)^{\widehat{b}_{\widehat{c}}}(i)^{\widehat{c}}\cdot \mc{M}^{-1}\left(U^{\dag}\ket{\widehat{\mb{b}}}\bra{\widehat{\mb{b}}}U\right)\\
       &=\sum_{U,c,\mb{b},b_c} \mathrm{Pr}(U,c,\mb{b},b_c)\ 2(-1)^{b_c}(i)^{c}\mc{M}^{-1}\left(U^{\dag}\ket{\mb{b}}\bra{\mb{b}}U\right)\\
       &=2\sum_{U} \mathrm{Pr}(U) \sum_{c}\mathrm{Pr}(c)\sum_{\mb{b},b_c}\mathrm{Pr}(\mb{b},b_c|U,c)\  (-1)^{b_c}(i)^{c}\mc{M}^{-1}\left(U^{\dag}\ket{\mb{b}}\bra{\mb{b}}U\right)\\
       &=2\sum_{U} \mathrm{Pr}(U) \sum_{c}\frac1{2} \sum_{\mb{b},b_c}\mathrm{Pr}(\mb{b},b_c|U,c)\  (-1)^{b_c}(i)^{c}\mc{M}^{-1}\left(U^{\dag}\ket{\mb{b}}\bra{\mb{b}}U\right).
    \end{aligned}
\end{equation}
Here the joint probability distribution $\mathrm{Pr}(U,c,\mb{b},b_c)=\mathrm{Pr}(U)\mathrm{Pr}(c)\mathrm{Pr}(\mb{b},b_c|U,c)$, since the the choices of $U$ and $c$ are independent, and $\mathrm{Pr}(c)=\frac{1}{2}$ for $c=0,1$.
Given $U$ and $\mb{b}$, and for $c=0$ (the $X$-basis), one can first sum the index $b_0$ as 
\begin{equation}\label{}
    \begin{aligned}
       \sum_{b_0}\mathrm{Pr}(\mb{b},b_0|U,0)\  (-1)^{b_0}
     =&\tr[ X_c \bra{\mb{b}}U\ \prod_{i=1}^{t-1} \mathrm{CO_i}\  \mathrm{CS}_t \left(\ket{+}_c\bra{+}\otimes \bigotimes_{i=1}^t \rho_i\right)  \mathrm{CS}_t^{\dag}\ \prod_{i=1}^{t-1} \mathrm{CO_i} \ U^{\dag}\ket{\mb{b}}]\\
     =&\frac1{2}\left(\bra{\mb{b}}U\sigma U^{\dag}\ket{\mb{b}}+\bra{\mb{b}}U\sigma^{\dag} U^{\dag}\ket{\mb{b}}\right)
    \end{aligned}
\end{equation}
with $\sigma$ defined in Eq.~\eqref{ap:eq:sigma}, and the final line follows similarly as in Eq.~\eqref{ap:eq:mswap}. Similarly, for the $c=1$ case, one has
\begin{equation}\label{}
    \begin{aligned}
       \sum_{b_1}\mathrm{Pr}(\mb{b},b_1|U,1)\  (-1)^{b_0}(-i)
     =&(-i)\tr[ Y_c \bra{\mb{b}}U\ \prod_{i=1}^{t-1} \mathrm{CO_i}\  \mathrm{CS}_t \left(\ket{+}_c\bra{+}\otimes \bigotimes_{i=1}^t \rho_i\right)  \mathrm{CS}_t^{\dag}\ \prod_{i=1}^{t-1} \mathrm{CO_i} \ U^{\dag}\ket{\mb{b}}]\\
     =&\frac1{2}\left(\bra{\mb{b}}U\sigma U^{\dag}\ket{\mb{b}}-\bra{\mb{b}}U\sigma^{\dag} U^{\dag}\ket{\mb{b}}\right)
    \end{aligned}
\end{equation}
following Eq.~\eqref{ap:eq:mswapY}. As a result, combing them one has
\begin{equation}\label{}
    \begin{aligned}
       \sum_{c}\frac1{2} \sum_{b_c}\mathrm{Pr}(\mb{b},b_c|U,c)\  (-1)^{b_c}(-i)^{c}=\frac1{2}\bra{\mb{b}}U\sigma U^{\dag}\ket{\mb{b}},
    \end{aligned}
\end{equation}
and by inserting it into Eq.~\eqref{ap:eq:sigmaEst} one further has
\begin{equation}\label{}
    \begin{aligned}
       \mathbb{E}_{\{U,c,\mb{b},b_c\}}\ \widehat{\sigma}=\sum_{U} \sum_{\mb{b}} \mathrm{Pr}(U) \bra{\mb{b}}U\sigma U^{\dag}\ket{\mb{b}}  \mc{M}^{-1}\left(U^{\dag}\ket{\mb{b}}\bra{\mb{b}}U\right)=\sigma.
    \end{aligned}
\end{equation}
The equality holds by the definition of the unbiased estimator of $\sigma$ in the original shadow protocol \cite{huang2020predicting}.

The variance of $\tr(O\widehat{\sigma})$ can be proved by following a similar way as in Sec.~\ref{ap:Th1}. 
\begin{equation}\label{ap:eq:VarSig}
    \begin{aligned}
       \mathbb{E}\ |\tr(O\widehat{\sigma})|^2&=\mathbb{E}\ \left|\tr\left[O\ 2(-1)^{\widehat{b}_c}(i)^{\widehat{c}} \mc{M}^{-1}\left(U^{\dag}\ket{\widehat{\mb{b}}}\bra{\widehat{\mb{b}}}U\right)\right]\right|^2\\
       &=\mathbb{E}\ 4\bra{\widehat{\mb{b}}}U\mc{M}^{-1}(O)U^{\dag}\ket{\widehat{\mb{b}}}^2\\
       &=4\sum_{U} \mathrm{Pr}(U) \sum_{c}\frac1{2}\sum_{\mb{b},b_c}\mathrm{Pr}(\mb{b},b_c|U,c)\  \bra{\mb{b}}U\mc{M}^{-1}(O)U^{\dag}\ket{\mb{b}}^2
    \end{aligned}
\end{equation}
Note that the value of $\tr(O\widehat{\sigma})$ could be a complex number, and thus we take the square of the absolute value here.
As in the expectation value case, here we first sum the index $b_c$. Since the phases from $b_c$ and $c$ are eliminated, similar as Eq.~\eqref{ap:eq:minus}, the summation about $b_c$ is equivalent to taking a partial trace on the control qubit. For $c=0/1$, it holds that
\begin{equation}\label{}
    \begin{aligned}
       \sum_{b_c}\mathrm{Pr}(\mb{b},b_c|U,c)
     =&\tr[ \id_c \bra{\mb{b}}U\ \prod_{i=1}^{t-1} \mathrm{CO_i}\  \mathrm{CS}_t \left(\ket{+}_c\bra{+}\otimes \bigotimes_{i=1}^t \rho_i\right)  \mathrm{CS}_t^{\dag}\ \prod_{i=1}^{t-1} \mathrm{CO_i} \ U^{\dag}\ket{\mb{b}}]\\
     =&\frac1{2}\bra{\mb{b}}U\rho_t U^{\dag}\ket{\mb{b}}+\frac1{2}\tr\left[\bra{\mb{b}}U \left(\bigotimes_{i=1}^{t-1}O_i\right) S_t \left(\bigotimes_{i=1}^t\rho_i\right) S_t^{\dag}\left(\bigotimes_{i=1}^{t-1} O_i\right) U^{\dag}\ket{\mb{b}}\right]\\
     =&\frac1{2}\left[\bra{\mb{b}}U\rho_t U^{\dag}\ket{\mb{b}}+\bra{\mb{b}}U\rho_1 U^{\dag}\ket{\mb{b}}\right]
    \end{aligned}
\end{equation}
where we use that fact that $O_i=O_i^{\dag}$ and $O_i^2=\id_i$. By inserting this result to Eq.~\eqref{ap:eq:VarSig}, one has
\begin{equation}\label{}
    \begin{aligned}
       \mathbb{E}\ |\tr(O\widehat{\sigma})|^2
       &=4\sum_{U} \mathrm{Pr}(U) \sum_{c}\frac1{2}\sum_{\mb{b},b_c}\mathrm{Pr}(\mb{b},b_c|U,c)\  \bra{\mb{b}}U\mc{M}^{-1}(O)U^{\dag}\ket{\mb{b}}^2\\
       &=4\sum_{U} \mathrm{Pr}(U) \sum_{\mb{b}}\bra{\mb{b}}U \frac{\rho_1+\rho_t}{2} U^{\dag}\ket{\mb{b}}\  \bra{\mb{b}}U\mc{M}^{-1}(O)U^{\dag}\ket{\mb{b}}^2\\
       &=4\ \mathbb{E}\ \left|\tr(O\widehat{\frac{\rho_1+\rho_t}{2}})\right|^2
    \end{aligned}
\end{equation}
Here in the last line the expectation value is for the shadow snapshot of the state $\rho':=(\rho_1+\rho_t)/2$ for any $t$. As a result,
\begin{equation}\label{}
    \begin{aligned}
 \mathrm{Var}\left[\tr(O\widehat{\sigma})\right]
      &=4\left\{\mathrm{Var}\left[\tr(O\widehat{\rho'})\right]+\tr(O\rho')^2\right\}-|\tr(O\sigma)|^2\\
      &=4\left\{\mathbb{E}\ \tr(O_0\widehat{\rho'})^2-\tr(O_0\rho')^2+\tr(O\rho')^2\right\}-|\tr(O\sigma)|^2\\
      &\leq 4\|O_0\|^2_{\mathrm{shadow}}+4\tr(O\rho')^2.
    \end{aligned}
\end{equation}
\end{proof}

With the estimator in Eq.~\eqref{ap:esttsigma} for the $t$-degree multiplication in Eq.~\eqref{ap:eq:sigma}, one can patch a few of them to estimate the $m$-th degree function $F_m$ in Eq.~\eqref{ap:Fm} . First, one can partition $F_m$ into $L$ pieces as follows.
\begin{equation}\label{ap:FmP}
    \begin{aligned}
       F_m(\{\rho_i\},\{O_i\})&=\tr(\underbrace{\rho_1O_1 \rho_2O_2\cdots \rho_{m_1}}_{\sigma_1}\ O_{m_1}\ \underbrace{\rho_{m_1+1}O_{m_1+1}
       \cdots \rho_{m_2}}_{\sigma_2}\ O_{m_1+m_2} \ \cdots    \cdots \underbrace{\rho_{m-m_L+1}O_{m-m_L+1}\cdots \rho_m}_{\sigma_L}\ O_m ),\\
       &=\tr(\sigma_1\ O_{m_1}\ 
      \sigma_2\ O_{m_1+m_2} \ \cdots    \cdots \sigma_L\ O_m )\\
      &=\tr(O_{m_1}\otimes O_{m_1+m_2}\cdots \otimes O_m\  S_L\ \sigma_1\otimes\sigma_2 \cdots \otimes \sigma_L)\\
    \end{aligned}
\end{equation} 
with each of $\sigma_i$ a $m_i$-degree multiplication. 
Then the estimator of $F_m$ reads
\begin{equation}\label{ap:FmPEst}
    \begin{aligned}
       \widehat{F_m}
&=\tr(O_{m_1}\otimes O_{m_1+m_2}\cdots \otimes O_m\  S_L\ \widehat{\sigma_1}\otimes\widehat{\sigma_2} \cdots \otimes \widehat{\sigma_L}),  
    \end{aligned}
\end{equation} 
where each of $\widehat{\sigma_i}$ is the estimator of $\sigma_i$ by Eq.~\eqref{ap:esttsigma}. Indeed, one can increase the number of independent snapshots for each $\sigma_i$ to reduce the variance. 

We remark that by applying the hybrid framework, one at most needs to parallel prepare $\max_i m_i$ copies of the states, and then conduct the classical postprocessing on $L$-copy of Hilbert space $\mc{H}_d$, which significantly reduce the experimental realization and also shadow estimation budget. Moreover, the variance can be analysed using Eq.~\eqref{eq:VarSig}. We give more detailed discussions on the advantages of the hybrid framework developed in Sec.~\ref{ap:hybrid} and \ref{ap:frame}, by showing the applications for measuring the moments and quantum error mitigation in the following two sections.

\section{Application for measuring moment $P_m$}\label{sec:ap:Pm}
In this section, we apply the hybrid framework developed in Sec.~\ref{ap:hybrid} and \ref{ap:frame} for measuring the higher-order moments of the quantum state. In particular, we propose new estimators of these moments, and take the third and forth moments $P_3$ and $P_4$ for the illustration.

\subsection{Unbiased estimators of $P_3$ and $P_4$}
We first give the proof of the unbiasedness of the estimator for $P_3$ in Eq.~(11) in Proposition 1.
\begin{proof}
The expectation of the estimator reads
\begin{equation}\label{}
    \begin{aligned}
      \mbb{E} \widehat{P_3}&= \frac1{MM'}\sum_{i\in[M],i'\in[M']}\tr(S_2\ \mbb{E}\widehat{\rho^{2}}_{(i)}\otimes \mbb{E}\widehat{\rho}_{(i')})\\
      &=\frac1{MM'}\sum_{i\in[M],i'\in[M']}\tr(S_2\ \rho^{2} \otimes \rho)=\tr(\rho^3),
    \end{aligned}
\end{equation}
where we use the fact that $\{\widehat{\rho^{2}}_{(i)}\}$ and $\{\widehat{\rho}_{(i')}\}$ are shadow sets for $\rho^2$ and $\rho$ respectively, and the snapshots in these two sets are independent.
\end{proof}
 Similarly, one can construct an unbiased estimator for $P_4$ as follows.
\begin{equation}\label{4mEst}
    \begin{aligned}
       \widehat{P_4}=\frac{2}{M(M-1)}\sum_{1\leq i<j \leq M}\tr(S_2\widehat{\rho^{2}}_{(i)}\otimes \widehat{\rho^2}_{(j)}).
    \end{aligned}
\end{equation}
where $\{\widehat{\rho^2}_{(i)}\}$ is the shadow set for $\rho^2$, and we take $i\neq j$ for independence. It is clear that one can in principle construct the estimator of general $P_m$ in the same way, by patching the shadow snapshots $\widehat{\rho^{t}}$ with $t<m$. One thing should be kept in mind is that one should take distinct snapshots for terms $\widehat{\rho^{t}}$ with the same $t$ in the summation for independence, as for $\widehat{P_4}$ in Eq.~\eqref{4mEst}.

At the end of this subsection, we give some additional remark on $\widehat{P_3}$. Different from $\widehat{P_4}$, there are two shadow sets $\{\widehat{\rho^{2}}_{(i)}\}$ and $\{\widehat{\rho}_{(i')}\}$ involved, which means one should make shadow estimation for $\rho^2$ and $\rho$ respectively. Actually, the measurement data collected in Algorithm 1  for any $t$ can also be used to construct shadow for $\rho$, by ignoring the information of the control qubit. For the single-shot case, one has $\widehat{\rho}= \mc{M}^{-1}\left(U^{\dag}\ket{\widehat{\mb{b}}}\bra{\widehat{\mb{b}}}U\right)$. This can be proved following the proof of Theorem 1  in Sec.~\ref{ap:Th1}. Note that since now there is no sign information from the measurement result $b_c$ of the control qubit, the summation of the index $b_c$ in Eq.~\eqref{eq:ap:bc} should be substituted with Eq.~\eqref{ap:eq:minus}. Consequently, one can construct another estimator of $P_3$ only using the measurement result in Algorithm 1  ($K=1$) for $t=2$ as follows.
\begin{equation}\label{3mEst1}
    \begin{aligned}
   \widehat{P_3}&= \frac{1}{M(M-1)}\sum_{i\neq i'}\tr\left[S_2\  (-1)^{\widehat{b}_c^{(i)}} \mc{M}^{-1}\left(U^{\dag}\ket{\widehat{\mb{b}}^{(i,1)}}\bra{\widehat{\mb{b}}^{(i,1)}}U\right)\otimes \mc{M}^{-1}\left(U^{\dag}\ket{\widehat{\mb{b}}^{(i',1)}}\bra{\widehat{\mb{b}}^{(i',1)}}U\right)\right].
    \end{aligned}
\end{equation}
The advantage of this estimator compared to that in Eq.~(11) is that one needs not to make additional shadow estimation for $\rho$.

\subsection{The variance of the estimators}\label{ap:var}
In this subsection, we give the statistical analysis of the estimators of moments, by applying Eq.~(7) in Theorem 1 , and also the properties of shadow norm. Fist, let us list the previous results on the upper bound of the shadow norm and variance, which are helpful for the later discussion.

\begin{fact}\label{fact:shadownorm}
In the shadow estimation \cite{huang2020predicting}, for a single-shot estimator $\widehat{\rho}$ of the quantum state $\rho\in \mc{H}_d$, and an observable $O$ with its traceless part $O_0$, the shadow norms of the Clifford and Pauli measurement primitives are respectively upper bounded by
\begin{equation}
    \begin{aligned}
      &\|O_0\|^2_{\mathrm{shadow}(c)}\leq 3\tr(O^2)\ (Proposition \ S1); \\
      &\|O_0\|^2_{\mathrm{shadow}(p)}\leq 2^{\mathrm{supp}(O)}\|\Tilde{O}\|_{2}^2\ (Proposition \ S3).
    \end{aligned}
\end{equation}
Here in the second bound for Pauli measurements we assume $O=\Tilde{O}\otimes \id_{2^{n-k}} $ with $\Tilde{O}$ restricted on the actively supporting $k$ qubits, and thus $\mathrm{supp}(O)=k$.

For the nonlinear function estimation, suppose $\widehat{\rho}_1$ and $\widehat{\rho}_2$ are two distinct estimators and $O$ is some operator on $\mc{H}_d^{\otimes 2}$, the variance with Clifford measurements can be upper bounded by
\begin{equation}
    \begin{aligned}
      \mathrm{Var}[\tr(O\widehat{\rho}_1\otimes \widehat{\rho}_2)]\leq 9\tr(O^2)+\frac{6}{d} \|O\|^2_{\infty}\ (Lemma~S6).
      \end{aligned}
\end{equation}
For Pauli measurements with $O=S_2$, the variance has the upper bound (Lemma~1 \cite{garcia2021quantum})
\begin{equation}
    \begin{aligned}
      \mathrm{Var}[\tr(S_2\widehat{\rho}_1\otimes \widehat{\rho}_2)]\leq d^3.
            \end{aligned}
\end{equation}
\end{fact}

We first give the proof of the variance for $\widehat{P_3}$ in Eq.~12 in Proposition 1.
\begin{proof}
Recall that in Eq.~(11) the unbiased estimator of $P_3$ shows
\begin{equation*}
    \begin{aligned}
       \widehat{P_3}=\frac1{M_1M_2}\sum_{i\in[M_1],j\in[M_2]}\tr(S_2\widehat{\rho^{2}}_{i}\otimes \widehat{\rho}_{j}),
    \end{aligned}
\end{equation*}
where $\widehat{\rho^{2}}_i$ and $\widehat{\rho}_{j}$ come from $M_1$ and $M_2$ independent snapshots. By definition, the variance shows
\begin{equation}\label{ap:eq:VarP3}
    \begin{aligned}
       \mathrm{Var}\left(\widehat{P_3}\right)=\frac1{(M_1M_2)^2} \sum_{i\in[M_1],j\in[M_2]} \sum_{i'\in[M_1],j'\in[M_2]} \mathbb{E}\left[\tr(S_2\widehat{\rho^{2}}_i\otimes \widehat{\rho}_{j}) \tr(S_2\widehat{\rho^{2}}_{i'}\otimes \widehat{\rho}_{j'})\right]-P_3^2 \end{aligned}
\end{equation}
If there is no index-coincidence, say $i\neq i'$ and $j\neq j'$, the expectation for this kind of term returns $P_3^2$. Thus one only needs to consider the following cases with the coincidence. 

One coincidence with $i=i'$ and $j\neq j'$:
\begin{equation}
    \begin{aligned}
       &\mathbb{E}\left[\tr(S_2\widehat{\rho^{2}}_i\otimes \widehat{\rho}_{j}) \tr(S_2\widehat{\rho^{2}}_{i}\otimes \widehat{\rho}_{j'})\right]-P_3^2\\
       =& \mathbb{E}\left[\tr(S_2\widehat{\rho^{2}}_i\otimes \rho)\tr(S_2\widehat{\rho^{2}}_i\otimes \rho)\right]-P_3^2\\
       =& \mathbb{E}\left[\tr(\rho\widehat{\rho^{2}}_i)^2\right]-P_3^2\\
       =&\mathrm{Var}\left[\tr(\rho \widehat{\rho^{2}})\right]\leq \|\rho-\id/d\|_\mathrm{shadow}^2+\tr(\rho^2)^2.
    \end{aligned}
\end{equation}
Here in the final line we omit the subscript $i$, and the result is just the variance of $O=\rho$ on the estimator $\widehat{\rho^{2}}$ for any $i$. The final inequality is due to Eq.~(7) in Theorem 1 .

One coincidence $i\neq i'$ and $j=j'$:
\begin{equation}
    \begin{aligned}
       &\mathbb{E}\left[\tr(S_2\widehat{\rho^{2}}_{i}\otimes \widehat{\rho}_{j}) \tr(S_2\widehat{\rho^{2}}_{i'}\otimes \widehat{\rho}_{j})\right]-P_3^2\\
       =& \mathbb{E}\left[\tr(S_2\rho^{2}\otimes \widehat{\rho}_{j})\tr(S_2\rho^{2}\otimes \widehat{\rho}_{j})\right]-P_3^2\\
       =&\mathrm{Var}\left[\tr(\rho^2 \widehat{\rho})\right]\leq \|\rho^2-\tr(\rho^2)\id/d\|_\mathrm{shadow}^2.
    \end{aligned}
\end{equation}
The result is just the variance of $O=\rho^2$ on the estimator $\widehat{\rho}$ for any $j$, and can be bounded by the shadow norm.

Two coincidences $i=i'$ and $j=j'$:
\begin{equation}
    \begin{aligned}
       &\mathbb{E}\left[\tr(S_2\widehat{\rho^{2}}_i\otimes \widehat{\rho}_{j}) \tr(S_2\widehat{\rho^{2}}_{i}\otimes \widehat{\rho}_{j})\right]-P_3^2\\
       =&\mathrm{Var}\left[\tr(S_2\widehat{\rho^{2}}\otimes \widehat{\rho})\right]\\
       =&\tr(\rho^{2})^2 \mathrm{Var}\left[\tr(S_2 \widehat{\sigma} \otimes \widehat{\rho})\right].
    \end{aligned}
\end{equation}
The result is the variance of $O=S_2$ on the composite estimator $\widehat{\rho^{2}}\otimes \widehat{\rho}$, and in the final line $\widehat{\sigma}=\widehat{\rho^{2}}/\tr(\rho^{2})$ is normalized.

By inserting all these cases into Eq.~\eqref{ap:eq:VarP3}, one gets
\begin{equation}
    \begin{aligned}
       \mathrm{Var}(\widehat{P_3})&=\frac1{(M_1M_2)^2}  \left\{ M_1\binom{M_2}{2} \mathrm{Var}\left[\tr(\rho \widehat{\rho^{2}})\right] + M_2\binom{M_1}{2} \mathrm{Var}\left[\tr(\rho^2 \widehat{\rho})\right] + M_1M_2 \mathrm{Var}\left[\tr(S_2\widehat{\rho^{2}}\otimes \widehat{\rho})\right]\right\},\\
      &\leq \frac{\|\rho-\id/d\|_\mathrm{shadow}^2+\tr(\rho^2)^2}{2M_1} + \frac{\|\rho^2-\id \tr(\rho^2)/d\|_\mathrm{shadow}^2}{2M_2} + \frac{\tr(\rho^{2})^2 \mathrm{Var}\left[\tr(S_2 \widehat{\sigma} \otimes \widehat{\rho})\right]}{M_1M_2},\\
       \end{aligned}
\end{equation}

By applying the results of shadow norm and variance for Pauli measurements in Proposition \ref{fact:shadownorm} and taking $M_1=M_2=M$ for simplicity,  one further obtains
\begin{equation}\label{ap:eq:P3Var}
    \begin{aligned}
       \mathrm{Var}_p(\widehat{P_3})&\leq \frac{ d\tr(\rho^2)+\tr(\rho^2)^2+d\tr(\rho^4)}{2M}+\frac{\tr(\rho^2)^2d^3}{M^2}\\
       &\leq \tr(\rho^2)[\frac{d+d\tr(\rho^2)+\tr(\rho^2)}{2M}+\frac{\tr(\rho^2)d^3}{M^2}]\\
       &\leq \tr(\rho^2)[\frac{d+1}{M}+\frac{\tr(\rho^2)d^3}{M^2}]\\
       \end{aligned}
\end{equation}
by the fact $\tr(\rho^4)\leq \tr(\rho^2)^2\leq 1$.

By directly applying Chebyshev’s inequality, one gets that
\begin{equation}\label{3mEstCheb}
    \begin{aligned}
       M\geq 2\max\left\{\frac{P_2(d+1)}{\epsilon^2\delta}, \frac{P_2d^{\frac{3}{2}}}{\epsilon\sqrt{\delta}}\right\}
    \end{aligned}
\end{equation}
is sufficient to let the estimation error less than $\epsilon$ within some confidence level $\delta$, i.e.,
$\mathrm{Prob}(|P_3-\widehat{P_3}|\leq\epsilon)\geq1-\delta$.
\end{proof}

We remark that the sampling complexity in Eq.~\eqref{3mEstCheb} of $P_3$ by the hybrid framework here is almost the same as that for $P_2$ by the original shadow estimation (Lemma 2 in Appendix of Ref.~\cite{elben2020mixedstate}), except that there is also a $P_2$ factor for the second term. This indicates that the hybrid framework reduces the variance and thus the statistical error from a 3-degree problem to a 2-degree one.

Similarly, by applying the results of shadow norm for the Clifford measurement in Proposition \ref{fact:shadownorm},
\begin{equation}\label{ap:eq:P3VarC}
    \begin{aligned}
       \mathrm{Var}_c(\widehat{P_3})&\leq \frac{ 3\tr(\rho^2)+\tr(\rho^2)^2+3\tr(\rho^4)}{2M}+\tr(\rho^2)^2\frac{9d^2+6/d}{M^2}\\
       &\leq \tr(\rho^2)[\frac{3+ 4\tr(\rho^2)}{2M}+\tr(\rho^2)\frac{9d^2+1}{M^2}],
       \end{aligned}
\end{equation}
and one can transform it to a lower bound of $M$ similarly as Eq.~\eqref{3mEstCheb} by using Chebyshev’s inequality.

We remark that the upper bounds in Eq.~\eqref{ap:eq:P3Var} and Eq.~\eqref{ap:eq:P3VarC} may be not tight, but they already imply the advantage over original shadow estimation. For instance, if one directly applies original shadow estimation for $P_3$ with estimator like $\widehat{P_3}^{(\mathrm{OS})}\propto\tr(S_3\widehat{\rho}_{(i)}\otimes \widehat{\rho}_{(j)}\otimes\widehat{\rho}_{(k)})$ with Pauli measurements (Eq.~(D15) in Ref.~\cite{elben2020mixedstate}), the leading order of the variance looks like (related to the 3 coincidence of the indices) 
\begin{equation}
    \begin{aligned}
       \mathrm{Var}_{p}\left[\widehat{P_3}^{(\mathrm{OS})}\right]&\sim \frac{d^3\tr(S_3S_3^{\dag})}{M^3}=\frac{d^6}{M^3},
       \end{aligned}
\end{equation}
as shown in Eq.~(D28) in Ref.~\cite{elben2020mixedstate}, which is worse than Eq.~\eqref{ap:eq:P3Var}. The advantage on the variance of our hybrid framework is also demonstrated by numerical results in Fig.~3 in main text.

The variance of the fourth moment $\widehat{P_4}$ can be bounded analytically in the same way as $\mathrm{Var}(\widehat{P_3})$ and shows a similar scaling, and we do not elaborate it here. We also numerically study the scaling of the statistical error with the qubit number $n$ using the Pauli measurement in Fig.~\ref{fig:varianceobsApp} (a), which is summarized in the second column of Table \ref{P4table}.

\subsection{Alternative estimators for the moments}\label{sec:Alternative}
In this subsection, we show alternative estimators of $P_3$ and $P_4$, by adopting the postprocessing in Ref.~\cite{Elben2019toolbox} using the same RM data collected in Algorithm 1  for $t=2$. We first prove the result of $P_3$ in Proposition 2.

\begin{proof}
In each shot of the measurement, as shown in Fig.~1 (c), the state is initialized to be $\rho_{c,I,II}=\ketbra{+}\otimes \rho^{\otimes 2}$ denoting the joint state. After the Controlled-$S_2$ operation and the random unitary evolution, the state before the projective measurement shows
\begin{equation}\label{eq:fullmatrix}
\begin{aligned}
  \rho_{c,I,II}=&\frac{1}{2}\ketbra{0}{0}\otimes(\mathbb{I}\otimes U)\rho^{\otimes 2}(\mathbb{I}\otimes U^\dagger)+\frac{1}{2}\ketbra{1}{1}\otimes(\mathbb{I}\otimes U)S_2\rho^{\otimes 2}S_2(\mathbb{I}\otimes U^\dagger)\\ 
  &+\frac{1}{2}\ketbra{0}{1}\otimes(\mathbb{I}\otimes U)\rho^{\otimes 2}S_2(\mathbb{I}\otimes U^\dagger)+\frac{1}{2}\ketbra{1}{0}\otimes(\mathbb{I}\otimes U)S_2\rho^{\otimes 2}(\mathbb{I}\otimes U^\dagger)
\end{aligned}
\end{equation}
As there is no measurement on the first copy of $\rho$, by taking a partial trace one gets the remaining density matrix on the control qubit and the second-copy as
\begin{equation}\label{eq:patrace}
\begin{aligned}
  \rho_{c,II}=\frac{1}{2}(\ketbra{0}{0}+\ketbra{1}{1})\otimes U\rho U^\dagger+\frac{1}{2}(\ketbra{0}{1}+\ketbra{1}{0})\otimes U\rho^2 U^\dagger.
\end{aligned}
\end{equation}
After measuring the second state in computational basis and getting the result $\mb{b}$, the (unnormalized) state of the control qubit becomes
\begin{equation}
\begin{aligned}
  \rho_c=\frac{1}{2}\bra{\mb{b}}U\rho U^\dagger\ket{\mb{b}}(\ketbra{0}{0}+\ketbra{1}{1})+\frac{1}{2}\bra{\mb{b}}U\rho^2 U^\dagger\ket{\mb{b}}(\ketbra{0}{1}+\ketbra{1}{0}).
\end{aligned}
\end{equation}
As a result, by measuring the control qubit in the $X$-basis, the probabilities for $0/1$ outcomes are
\begin{equation}\label{eq:probs}
\begin{aligned}
  &\mathrm{Pr}(0,\mb{b}|U)=\bra{+}\rho_c\ket{+}=\frac{1}{2}\bra{\mb{b}}U\rho U^\dagger\ket{\mb{b}}+\frac{1}{2}\bra{\mb{b}}U\rho^2 U^\dagger\ket{\mb{b}},\\
  &\mathrm{Pr}(1,\mb{b}|U)=\bra{-}\rho_c\ket{-}=\frac{1}{2}\bra{\mb{b}}U\rho U^\dagger\ket{\mb{b}}-\frac{1}{2}\bra{\mb{b}}U\rho^2 U^\dagger\ket{\mb{b}}.
\end{aligned}
\end{equation}
For any term in the summation of the estimator $\widehat{P_3}'$ in Proposition 2, say the $i$-the measurement setting, and $j,j'$-th shots, the expectation value shows (here we omit the labels of $i,j,j'$ for simplicity)
\begin{equation}\label{ap:eq:P3Zest}
\begin{aligned}
  \mathbb{E}_{U,b_c,b_c',\mb{b},\mb{b}'}\ (-1)^{b_c}X(\mb{b},\mb{b}')=&\mathbb{E}_{U} \sum_{b_c,b_c',\mb{b},\mb{b}'}\ \mathrm{Pr}(b_c,\mb{b}|U)\mathrm{Pr}(b_c',\mb{b}'|U)\ (-1)^{b_c}X(\mb{b},\mb{b}')\\
  =&\mathbb{E}_U\sum_{\mb{b},\mb{b}'} \ \left[\mathrm{Pr}(0,\mb{b}|U)-\mathrm{Pr}(1,\mb{b}|U)\right]\left[\mathrm{Pr}(0,\mb{b}'|U)+\mathrm{Pr}(1,\mb{b}'|U)\right]\ X(\mb{b},\mb{b}')\\
  =&\mathbb{E}_U\sum_{\mb{b},\mb{b}'}X(\mb{b},\mb{b}')\bra{\mb{b}}U\rho^2  U^\dagger\ket{\mb{b}}\bra{\mb{b}'}U\rho U^\dagger\ket{\mb{b}'}\\
  =&\mathbb{E}_U\tr\left[\sum_{\mb{b},\mb{b}'}X(\mb{b},\mb{b}')\ketbra{\mb{b},\mb{b}'}{\mb{b},\mb{b}'}U^{\otimes 2}\ (\rho^2\otimes\rho)\ U^{\dagger\otimes 2}\right]\\
  =&\tr\left[ \mathbb{E}_U U^{\dagger\otimes 2}XU^{\otimes 2}\ (\rho^2\otimes\rho)\right]\\
  =&\tr\left[S_2(\rho^2\otimes\rho)\right]=\tr(\rho^3).
\end{aligned}
\end{equation}
where the third line is according to Eq.~\eqref{eq:probs}.
Here we define the diagonal operator $X:=\sum_{\mb{b},\mb{b}'}X(\mb{b},\mb{b}')\ketbra{\mb{b},\mb{b}'}{\mb{b},\mb{b}'}$, and the 2-fold twirling channel can reproduce the swap operator $\mathbb{E}_U U^{\dagger\otimes 2}XU^{\otimes 2}=S_2$ \cite{Elben2019toolbox} . Note that the choice of classcial postprocessing $X_{c\backslash p}(\mb{b},\mb{b}')$ and thus the operator $X_{c\backslash p}$ depend on the random unitary ensemble being $n$-qubit Clifford or tensor-product single-qubit ones. Then by averaging all the terms in $\widehat{P_3}'$, one finishes the proof.
\end{proof}

Similarly, one can construct an alternative unbiased estimator for $P_4$ as
\begin{equation}\label{4mEstZ}
\begin{aligned}
  \widehat{P_4}'=\frac{1}{MK(K-1)}\sum_{i=1}^{M}\sum_{j\neq j'}(-1)^{\widehat{b}_c^{(i,j)}+\widehat{b}_c^{(i,j')}}X_{c\backslash p}(\widehat{\mb{b}}^{(i,j)},\widehat{\mb{b}}^{(i,j')}),
\end{aligned}
\end{equation}
by adding the information of the control qubit of $j'$-th shot. The proof is quite similar, for any term in the summation the expectation value shows 
\begin{equation}\label{4mEstZproof}
\begin{aligned}
  \mathbb{E}_{U,b_c,b_c',\mb{b},\mb{b}'}\ (-1)^{b_c+b_c'}X(\mb{b},\mb{b}')=&\mathbb{E}_{U} \ \mathrm{Pr}(b_c,\mb{b}|U)\mathrm{Pr}(b_c',\mb{b}'|U)\ (-1)^{b_c+b_c'}X(\mb{b},\mb{b}')\\
  =&\mathbb{E}_U\sum_{\mb{b},\mb{b}'} \ \left[\mathrm{Pr}(0,\mb{b}|U)-\mathrm{Pr}(1,\mb{b}|U)\right]\left[\mathrm{Pr}(0,\mb{b}'|U)-\mathrm{Pr}(1,\mb{b}'|U)\right]\ X(\mb{b},\mb{b}')\\
  =&\mathbb{E}_U\sum_{\mb{b},\mb{b}'}X(\mb{b},\mb{b}')\bra{\mb{b}}U\rho^2  U^\dagger\ket{\mb{b}}\bra{\mb{b}'}U\rho^2 U^\dagger\ket{\mb{b}'}\\
  =&\mathbb{E}_U\tr\left[\sum_{\mb{b},\mb{b}'}X(\mb{b},\mb{b}')\ketbra{\mb{b},\mb{b}'}{\mb{b},\mb{b}'}U^{\otimes 2}\ (\rho^2\otimes\rho^2)\ U^{\dagger\otimes 2}\right]\\
  =&\tr\left[ \mathbb{E}_U U^{\dagger\otimes 2}XU^{\otimes 2}\ (\rho^2\otimes\rho^2)\right]\\
  =&\tr\left[S_2(\rho^2\otimes\rho^2)\right]=\tr(\rho^4).
\end{aligned}
\end{equation}

At the end of this subsection, we give a brief discussion on the \textbf{statistical variance} of the estimator $\widehat{P_3}'$ in Proposition 2 by following the approach in the previous works \cite{singlezhou,Zhenhuan2022correlation}, and the same analysis works for $\widehat{P_4}'$.
For clearness, here we show the estimator $\widehat{P_3}'$ in Proposition 2 again as
\begin{equation*}\label{}
\begin{aligned}
  \widehat{P_3}'=\frac{1}{MK(K-1)}\sum_{i=1}^{M}\sum_{j\neq j'}(-1)^{\widehat{b}_c^{(i,j)}}X_{c\backslash p}(\widehat{\mb{b}}^{(i,j)},\widehat{\mb{b}}^{(i,j')}).
\end{aligned}
\end{equation*}
It is clear that $\widehat{P_3}'$ is the summation of $M$ i.i.d estimators, so here we first consider the case $M=1$, and omit the index $i$ of $(i,j)$ for the $i$-th setting and also the hat for the random variable for concise.
By definition, the variance shows
\begin{equation}
\begin{aligned}
\mathrm{Var}\left(\widehat{P_3}'\right)=\mathbb{E}\left(\widehat{P_3}'^2\right)-P_3^2,
\end{aligned}
\end{equation}
and the first term is
\begin{equation}\label{ap:varP3a}
\begin{aligned}
\mathbb{E}\left(\widehat{P_3}'^2\right)=\frac{1}{K^2(K-1)^2}\mathbb{E}\sum_{j_1\neq j_1'}\sum_{j_2\neq j_2'}(-1)^{b_c^{(j_1)}+b_c^{(j_2)}}X(\mb{b}^{(j_1)},\mb{b}^{(j_1')})X(\mb{b}^{(j_2)},\mb{b}^{(j_2')}).
\end{aligned}
\end{equation}
Note that $\{b_c,\mb{b}\}$ here are all random variables, and the calculation of each term in the summation depends on the index coincidence, similar as Eq.~\eqref{ap:eq:VarP3} for $\widehat{P_3}$. Here for simplicity, we only consider the two-coincidence case, that is $j_1=j_2, j_1'=j_2'$, or $j_1=j_2', j_1'=j_2$, which contributes to the leading term of the variance \cite{singlezhou,Zhenhuan2022correlation,liu2022detecting}, and one can also figure out the other sub-leading terms with less coincidence in a similar way.

For $j_1=j_2, j_1'=j_2'$, there are $K(K-1)$ terms in the summation of Eq.~\eqref{ap:varP3a}, and one of them would show as
\begin{equation}
\begin{aligned}
&\mathbb{E}\ (-1)^{b_c^{(j_1)}+b_c^{(j_1)}}X(\mb{b}^{(j_1)},\mb{b}^{(j_1')})X(\mb{b}^{(j_1)},\mb{b}^{(j_1')})\\
=&\mathbb{E}\ X^2(\mb{b}^{(j_1)},\mb{b}^{(j_1')})=\tr\left[\Phi^2_{\mc{E}}(X^2)\ \rho\otimes\rho\right]
\end{aligned}
\end{equation}
where the final equality follows similarly as Eq.~\eqref{ap:eq:P3Zest}, and the 2-copy twirling channel is denoted as $\Phi^2_{\mc{E}}(\cdot)=\mathbb{E}_{U\in\mc{E} }U^{\dagger\otimes 2}(\cdot)U^{\otimes 2}$. Note here the phase information of $b_c^{(j_1)}$ is cancelled, and thus one has $\rho$ but not $\rho^2$ in the final result. For  $j_1=j_2', j_1'=j_2$ also with $K(K-1)$ terms,
\begin{equation}
\begin{aligned}
&\mathbb{E}\ (-1)^{b_c^{(j_1)}+b_c^{(j_2)}}X(\mb{b}^{(j_1)},\mb{b}^{(j_2)})X(\mb{b}^{(j_2)},\mb{b}^{(j_1)})\\
=&\mathbb{E}\ (-1)^{b_c^{(j_1)}+b_c^{(j_2)}}X^2(\mb{b}^{(j_1)},\mb{b}^{(j_2)})=\tr\left[ \Phi^2_{\mc{E}}(X^2)\ \rho^2\otimes\rho^2\right]
\end{aligned}
\end{equation}
and the second equality can follow Eq.~\eqref{4mEstZproof}.
By combing both cases, one has the leading term of the variance as
\begin{equation}\label{}
\begin{aligned}
\mathrm{Var}\left(\widehat{P_3}'\right) \sim \frac{\tr[ \Phi^2_{\mc{E}}(X^2)\ (\rho^2\otimes\rho^2+\rho \otimes\rho) ]}{MK(K-1)}\leq \frac{2\tr[ \Phi^2_{\mc{E}}(X^2) \rho \otimes\rho ]}{MK(K-1)}
\end{aligned}
\end{equation}
where the inequality is due to $\rho- \rho^2$ and $\Phi^2_{\mc{E}}(X^2)$ are non-negative. Note that the final bound is actually the leading term of the variance of $P_2$ by using the original RM \cite{Elben2019toolbox}, which indicates that the hybrid framework reduces the variance of a 3-degree problem to a 2-degree one, similar as in the hybrid shadow estimation shown in Sec.~\ref{ap:var}.

The twirling result $\Phi^2_{\mc{E}}(X^2)$ depends on the random unitary ensemble $\mc{E}$. For the Clifford measurement primitive, it shows 
$\Phi^2_{\mc{E}_c}(X_c^2)=d\id_{d^2}+(d-1)S_2$; for the Pauli measurement primitive, it is in the tensor-product form $\Phi^2_{\mc{E}_p}(X_p^2)=\bigotimes_{i=1}^n (2\id_4+S_2^{(i)})$ with $S_2^{(i)}$ the swap operator for the $i$-th qubit pair \cite{Zhenhuan2022correlation}. As a result, for the Clifford measurement primitive,
\begin{equation}\label{ap:randomClvar}
\begin{aligned}
\mathrm{Var}_c\left(\widehat{P_3}'\right) \sim \frac{2[d+(d-1)P_2]}{MK^2}\leq \frac{4d}{MK^2}.
\end{aligned}
\end{equation}
For the Pauli measurement primitive,
\begin{equation}\label{ap:randomLvar}
\begin{aligned}
\mathrm{Var}_p\left(\widehat{P_3}'\right) \sim 
\frac{2\tr[ \bigotimes_{i=1}^n (2\id_4+S_2^{(i)}) \rho \otimes\rho ]}{MK^2}= \frac{2\sum_{A\subseteq[n]} \tr(\rho_A^2)2^{n-|A|}}{MK^2}\leq \frac{2\sum_{A\subseteq[n]} 2^{n-|A|}}{MK^2}= \frac{2d^{\log_2 3}}{MK^2}.
\end{aligned}
\end{equation}
where $A\subseteq[n]$ is any subsystem of the $n$-qubit system, $\rho_A=\tr_{\bar{A}}(\rho)$ the reduced density matrix on $A$ with $\tr(\rho_A^2)\leq 1$, and the final equality is by summing the binomial with $d=2^n$. As a result, to make the statistical error to be some constant, the shotting time should be $K=\mc{O}(d^{(\log_2 3)/2})=\mc{O}(d^{0.79})$.

From Eq.~\eqref{ap:randomLvar}, it is clear that the variance is related to the subsystem purity. For the noisy GHZ state $\rho=q\ket{\mathrm{GHZ}}\bra{\mathrm{GHZ}}+(1-q)\mathbb{I}/d$ used in the numerics, here we give a refined estimation of the scaling. Note that for subsystem $A\subset[n]$ with $d_A=2^{|A|}$,
\begin{equation}\label{}
\begin{aligned}
&\rho_A=\frac{q}{2}\left(\ket{00\cdots0}\bra{00\cdots0}+\ket{11\cdots1}\bra{11\cdots1}\right)+(1-q) \id_{d_A}/d_A,\\
&\tr(\rho_A^2)=q^2/2+(1-q^2)2^{-|A|}.
\end{aligned}
\end{equation}
which is only related to the qubit number in $A$, and for $A=[n]$ the whole system, $\tr(\rho^2)=q^2+(1-q^2)2^{-n}$. By inserting all these purity into Eq.~\eqref{ap:randomLvar}, one has
\begin{equation}\label{ap:randomLvar1}
\begin{aligned}
\mathrm{Var}_p\left(\widehat{P_3}'\right) &\sim 
2M^{-1}K^{-2}\left\{[q^2+(1-q^2)2^{-n}]+\sum_{k=0}^{n-1} \binom{n}{k} [q^2/2+(1-q^2)2^{-k}]2^{n-k}\right\}\\
&=2M^{-1}K^{-2}\left\{\frac{q^2}{2}+2^{n}\sum_{k=0}^n \binom{n}{k} \left[\frac{q^2}{2}2^{-k}+(1-q^2)2^{-2k}\right]\right\}\\
&=2M^{-1}K^{-2}\left\{\frac{q^2}{2}+2^{n}\left[\frac{q^2}{2}(1/2+1)^n+(1-q^2)(1/4+1)^n\right]\right\}\\
&=M^{-1}K^{-2}\left[q^2(1+3^{n})+2(1-q^2)(2.5)^n\right]=M^{-1}K^{-2}\left[q^2(1+d^{\log_2 3})+2(1-q^2)d^{\log_{2} 2.5}\right].
\end{aligned}
\end{equation}
As a result, the final variance is the interpolation between $d^{\log_2 3}$ and $d^{\log_{2}2.5}$, and thus the shot-number $K$ is between $\mc{O}(d^{(\log_{2}2.5)/2})=\mc{O}(d^{0.66})$, and $\mc{O}(d^{0.79})$, which is consistent with our numerical results.

The variance of $\widehat{P_4}'$ in Eq.~\eqref{4mEstZ} can be analysed in a similar way, and the leading term shows the same scaling behaviour as $\widehat{P_3}'$ in Eq.~\eqref{ap:randomClvar} and \eqref{ap:randomLvar}. We also numerically study the scaling of the statistical error with the qubit number $n$ using the Pauli measurement in Fig.~\ref{fig:varianceobsApp} (a), which is summarized in the final column of Table \ref{P4table}.

\begin{table}[htb]
\begin{tabular}
{|c|c|c|c|}
\hline
    $P_4$ & OS & HS & HR  \\
    \hline
     Anal.& $\mathcal{O}(\frac{d^{4}}{M^{2}})$ \cite{elben2020mixedstate} & $\mathcal{O}(\frac{d^{1.5}}{M})$ & $\mathcal{O}(\frac{d^{0.79}}{\sqrt{M}K})$ 
 \\
     \hline
     Numer.& \diagbox{}{}& $\mathcal{O}(\frac{d^{0.99}}{M})$ [blue]& $\mathcal{O}(\frac{d^{0.73}}{\sqrt{M}K})$[green]\\
     \hline
\end{tabular}
\caption{The statistical errors for estimating $P_4$ with different protocols using the Pauli measurement. The numerical results are from Fig.~\ref{fig:varianceobsApp} (a). The analytical results of HS and HR follow similarly as that of $P_3$ in Eq.~\eqref{ap:eq:P3Var} and \eqref{ap:randomLvar}, respectively. The numerical result of OS is not shown since the measurement and post-processing costs are too demanding for the $4$-degree function.
}\label{P4table}
\end{table}

\begin{figure}

    \includegraphics[width=9cm]{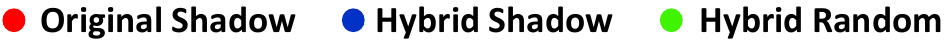}
    
    \includegraphics[width=5cm]{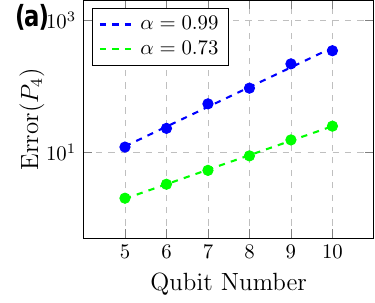}
    \includegraphics[width=5cm]{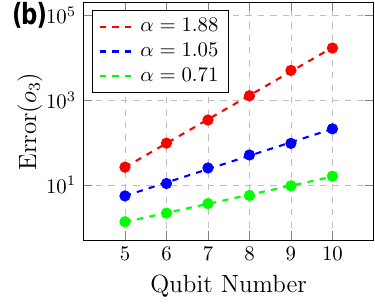}
    \includegraphics[width=5cm]{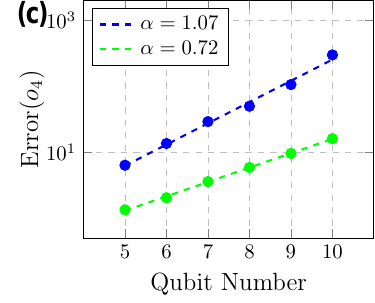}
    \caption{
    Scaling of statistical errors of measuring different quantities with respect to the qubit number $n$ for different protocols using the Pauli measurement.  The processed state is the noisy multi-qubit GHZ state, $\rho=0.8\ketbra{\mathrm{GHZ}}{\mathrm{GHZ}}+0.2\mathbb{I}_d/d$, and we take $M=10$ and $K=1$ for original shadow (OS) and hybrid shadow (HS) protocols, and $M=2$ and $K=5$ for hybrid random (HR) protocol, that is, in the regime $d\gg M(K)$ to study the scaling behavior. 
     (a) shows the error of measuring $P_4$ for HS and HR, using the estimator $\widehat{P_4}$ in Eq.~\eqref{4mEst} and $\widehat{P_4}'$ Eq.~\eqref{4mEstZ}, respectively. In (b) and (c),  $O$ is taken to be the local observable $\sigma_Z^1\otimes\sigma_Z^2$ on the first two qubits.
    (b) shows the error of measuring $o_3$ for OS, HS and HR, using the original shadow estimator \cite{huang2020predicting,elben2020mixedstate},
    $\widehat{o_3}$ shown in main text and $\widehat{o_3}'$ in Eq.~\eqref{eq:o3infty}, respectively.
    (c) shows the error of measuring $o_4$ for HS and HR, using the estimator $\widehat{o_4}$ in Eq.~\eqref{4mEstO} and $\widehat{o_4}'$ Eq.~\eqref{eq:o4Papp}, respectively. All the scalings are quite similar to that of $P_3$ in Fig.~3 (c), which shows the advantage of the current hybrid framework, and we list these results together with analytical bound in Table \ref{P4table}, \ref{o3table}, and \ref{o4table} for compare.
    Note that we do not show the results of OS in (a) and (c) here as in Fig.~3 (c) in main text, since the measurement and post-processing costs scale $\mc{O}(M^4)$ of OS for the $4$-degree functions $P_4$ and $o_4$,  which is demanding for the numerical simulation. Indeed, the budgets also set a big challenge for the application of the original shadow to high-degree functions.} 
    \label{fig:varianceobsApp}
\end{figure}

\section{Application in quantum error mitigation}\label{ap:sec:miti}
In the purification-based quantum error mitigation \cite{Huggins2021Virtual,Koczor2021Exponential}, the central task is to measure $o_m:=\tr(O\rho^m)$ and $P_m$. The measurement of $P_m$ has been discussed in the previous section. Here we give discussions about estimators of $o_m$ and their performance via the developed hybrid framework.

\subsection{Advantage on estimating $\tr(O\rho^t)$}\label{ap:sec:Ot}
In this subsection, we aims to estimate $o_t:=\tr(O\rho^t)$ by directly using the shadow set   $\left\{\widehat{\rho^t}_{(1)}, \widehat{\rho^t}_{(2)}, \cdots \widehat{\rho^t}_{(M)}\right\}$ via running Algorithm 1  in main text. Here we assume one can access a quantum computer controlling $t$-copy of the state. The copy number $t$ here is usually not large, for instance, $t=2$. 

The estimator is just $\widehat{o_t}=M^{-1}\sum_i \tr(O\widehat{\rho^t}_{(i)})$ by averaging total $M$ independent snapshots, and the variance of it is upper bounded by
\begin{equation}\label{ap:eq:VarOt}
    \begin{aligned}
      \mathrm{Var}(\widehat{o_t})\leq M^{-1}\left[\|O_0\|^2_{\mathrm{shadow}}+\tr(O\rho)^2\right]
    \end{aligned}
\end{equation}
on account of Eq.~(7) in Theorem 1 in main text. 
By applying results of the shadow norm in Proposition \ref{fact:shadownorm}, one further has
\begin{equation}\label{ap:eq:VarOt1}
    \begin{aligned}
      &\mathrm{Var}_c(\widehat{o_t})\leq M^{-1}\left[3\tr(O^2)+\tr(O\rho)^2\right],\\
      &\mathrm{Var}_p(\widehat{o_t})\leq M^{-1}\left[2^{\mathrm{supp}(O)}\|\Tilde{O}\|_{2}^2+\tr(O\rho)^2\right],
    \end{aligned}
\end{equation}
for the Clifford and Pauli measurements respectively. 

The second term can be bounded by $\tr(O\rho)^2\leq \|O\|_{\infty}^2$ which is assumed to be less than some constant, thus is not essential. Due to the hybrid framework which utilizes the power of a quantum computer, the variance of $\widehat{o}_t$ is reduced to the original shadow estimation on single-copy, and thus the hybrid framework provides a few advantages. For the Clifford measurement, suppose $O$ has a low rank, for instance $O=\ketbra{\Psi}$ some pure state, then $\tr(O^2)=\mc{O}(1)$; For Pauli measurements, suppose $O$ is a $k$-local operator, $2^k\|\Tilde{O}\|_{2}^2\leq 4^k \|O\|_{\infty}^2$, which is also moderate as $k$ is not very large. In particular, if $O$ is a $k$-local Pauli operator, this first term can be further tightened to $3^k$ \cite{huang2020predicting}.

To compare, we take $t=2$ and the estimator of the original shadow protocol shows \cite{Seif2022Shadow}
\begin{equation}\label{}
    \begin{aligned}
     \widehat{o_2}^{(\mathrm{OS})}=\frac{1}{M(M-1)}\sum_{ i\neq j }\tr(S_2 O\widehat{\rho}_{(i)}\otimes \widehat{\rho}_{(j)}),
    \end{aligned}
\end{equation}
with the shadow set $\left\{\widehat{\rho}_{(1)}, \widehat{\rho}_{(2)}, \cdots \widehat{\rho}_{(M)}\right\}$. For Pauli measurements with any Pauli operator $O$, the leading term of the upper bound of the variance is like $d^3/M^2$ \cite{garcia2021quantum,elben2020mixedstate}, which has nothing to do with the locality of $O$ compared with the hybrid approach; for Clifford measurements, the leading term of the upper bound of the variance is like $d\tr(O^2)/M^2$ \cite{huang2020predicting}. Thus in each cases, the variance would scale with the Hilbert space dimension, thus exponentially with the qubit number. These exponential advantages are also manifested by the numerical results in Fig.~3(a) and (b) in main text.


\subsection{Estimators of $o_3$ and $o_4$}\label{ap:sec:Om}

To measure $o_m$ with a larger $m$, for instance $m\geq 3$, it is challenging to directly collect the shadow snapshot $\widehat{\rho^m}$ due to the hardware limitation. Alternatively, here we combine a few low-degree shadow snapshots, as for estimating the moments. In main text, we have shown the estimator of $o_3$, and we give that of $o_4$ by modifying Eq.~\eqref{4mEst} as follows.

\begin{equation}\label{4mEstO}
    \begin{aligned}
       \widehat{o_4}&=\frac{2}{M(M-1)}\sum_{1\leq i<j \leq M}\tr(O_2\widehat{\rho^{2}}_{(i)}\otimes \widehat{\rho^2}_{(j)})\\
      &=\frac{1}{M(M-1)}\sum_{i\neq j }\tr(S_2O\widehat{\rho^{2}}_{(i)}\otimes \widehat{\rho^2}_{(j)})\\
       &=\frac{1}{M(M-1)}\sum_{i\neq j }\tr(O\widehat{\rho^{2}}_{(i)}\widehat{\rho^2}_{(j)}),
    \end{aligned}
\end{equation}
where $\{\widehat{\rho^{2}}_{(i)}\}$ is the shadow set of $\rho^2$. Here $O_2=(OS_2+S_2O)/2$, and in the second line we put this symmetry in the summation of indices. The variance of $\widehat{o_3}$ and $\widehat{o_4}$ can be analysed in the same way as $\widehat{P_3}$ in Sec.~\ref{ap:var}, and show similar behaviour. We also numerically study the scaling of the statistical error with the qubit number $n$ using the Pauli measurement in Fig.~\ref{fig:varianceobsApp} (b) and (c), which is summarized in the second column of Table \ref{o3table} and Table \ref{o4table}, respectively.

\begin{table}[htb]
\begin{tabular}
{|c|c|c|c|}
\hline
    $o_3$ & OS & HS & HR  \\
    \hline
     Anal.& $\mathcal{O}(\frac{d^{3}}{M^{1.5}})$ \cite{elben2020mixedstate} & $\mathcal{O}(\frac{d^{1.5}}{M})$ & $\mathcal{O}(\frac{d^{0.79}}{\sqrt{M}K})$ 
 \\
     \hline
     Numer.& $\mathcal{O}(\frac{d^{1.88}}{M^{1.5}})$ [red]& $\mathcal{O}(\frac{d^{1.05}}{M})$ [blue]& $\mathcal{O}(\frac{d^{0.71}}{\sqrt{M}K})$[green]\\
     \hline
\end{tabular}
\caption{The statistical errors for estimating $o_3$ with different protocols using the Pauli measurement. The numerical results are from Fig.~\ref{fig:varianceobsApp} (b). The analytical results of HS and HR follow similarly as that of $P_3$ in Eq.~\eqref{ap:eq:P3Var} and \eqref{ap:randomLvar}, respectively. Consequently, in practice, one needs $M=\mathcal{O}(d^{1.25})$ for OS, $M=\mathcal{O}(d^{1.05})$ for HS, and $K=\mathcal{O}(d^{0.71})$ for HR to make the error less than some constant. It is clear that HS and HR from the hybrid framework both show an advantage compared to OS.
}\label{o3table}
\end{table}

Similar to the measurement of moments $P_3$ and $P_4$ in Sec.~\ref{sec:Alternative}, here we give alternative estimators of $o_3$ and $o_4$ with the post-processing strategy in Ref.~\cite{Elben2019toolbox}.  The central idea is to effctively make RMs on $\rho O\rho$.

The first method decompose $O$ and additionally applies control operation. Any Hermitian $O$ can be decomposed into the form $O=\frac{1}{2}\norm{O}_{\infty}\left(V_O+V_O^\dagger\right)$, 
where $\norm{O}_{\infty}=\max_{\mb{b}} |\lambda_{\mb{b}}|$ is the largest absolute eigenvalue of $O$, and $V_O$ is some unitary determined by $O$. In particular, when $O$ is also a unitary such as a Pauli operator, $\norm{O}_\infty=1$ and $V_O=O$, and the quantum circuit is reduced to that in Fig.~\ref{Fig:FmHybrid}.

We show the quantum circuits in Fig.~\ref{fig:Odecomp} (a), where one additionally apply the Controlled-$V_O$ on the first-copy. The unbiased estimator of $o_3$ shows
\begin{equation}\label{eq:o3infty}
\begin{aligned}
  \widehat{o_3}'=\frac{\norm{O}_{\infty}}{MK(K-1)}\sum_{i\in[M]}\sum_{j\neq j'\in[K] }(-1)^{\widehat{b_c}^{(i,j)}}X_{c\backslash p}(\widehat{\mb{b}}^{(i,j)},\widehat{\mb{b}}^{(i,j')}).
\end{aligned}
\end{equation}

To construct the unbiased estimator of $o_4$, besides the RM results collected from the quantum circuit in Fig.~\ref{fig:Odecomp} (a), labeled by $\{b_{c}^{(i,j)},\mb{b}^{(i,j)}\}$, one also needs the measurement results collected in the Algorithm 1 from the quantum circuit in Fig.~1 (c) with $t=2$ in main text, labeled by $\{b_{c}^{'(i,j)},\mb{b}^{'(i,j)}\}$. Then the unbiased estimator of $o_4$ reads
\begin{equation}\label{eq:o4Papp}
\begin{aligned}
  \widehat{o_4}'=\frac{\norm{O}_{\infty}}{MKK'}\sum_{i\in[M]}\sum_{j\in[K],j'\in[K'] }(-1)^{\widehat{b_c}^{(i,j)}+\widehat{b_c'}^{(i,j')}}X_{c\backslash p}(\widehat{\mb{b}}^{(i,j)},\widehat{\mb{b}'}^{(i,j')}),
\end{aligned}
\end{equation}
where $K$ and $K'$ are the number of measurement results collected from these two circuits under the $i$-th same unitary $U$.
\begin{proof}
The proof of the unbiasedness of $\widehat{o_3}'$ is similar to that of $\widehat{P_3}'$ in Sec.~\ref{sec:Alternative}. It is not hard to show that the measurement probabilities satisfy
\begin{equation}
\begin{aligned}
  \mathrm{Pr}(0,\mb{b}|U)+\mathrm{Pr}(1,\mb{b}|U)&=\bra{\mb{b}}U\rho U^\dagger\ket{\mb{b}},\\
  \mathrm{Pr}(0,\mb{b}|U)-\mathrm{Pr}(1,\mb{b}|U)&=\frac{1}{2}\left(\bra{\mb{b}}U\rho V_O\rho U^\dagger\ket{\mb{b}}+\bra{\mb{b}}U\rho V_O^\dag \rho U^\dagger\ket{\mb{b}}\right)=\frac{1}{2}\bra{\mb{b}}U\rho (V_O+V_O^\dag)\rho U^\dagger\ket{\mb{b}}.
\end{aligned}
\end{equation}

Following similarly as Eq.~\eqref{ap:eq:P3Zest}, one has the expectation value for any one term from $\widehat{o_3}'$ in Eq.~\eqref{eq:o3infty} as (here we omit the labels of $i,j,j'$ for simplicity)
\begin{equation}
\begin{aligned}
  &\norm{O}_\infty\mathbb{E}_U\sum_{\mb{b},\mb{b}'}[\mathrm{Pr}(0,\mb{b}|U)-\mathrm{Pr}(1,\mb{b}|U)][\mathrm{Pr}(0,\mb{b}|U)+\mathrm{Pr}(1,\mb{b}|U)]X(\mb{b},\mb{b}')\\
  =&\mathbb{E}_U\sum_{\mb{b},\mb{b}'}X(\mb{b},\mb{b}')\bra{\mb{b}}U\rho U^\dagger\ket{\mb{b}}\bra{\mb{b}}U \rho \left[\frac{1}{2}\norm{O}_\infty(V_O+V_O^\dagger)\right] \rho U^\dagger\ket{\mb{b}}\\
  =&\mathbb{E}_U\sum_{\mb{b},\mb{b}'}X(\mb{b},\mb{b}')\bra{\mb{b}}U\rho U^\dagger\ket{\mb{b}}\bra{\mb{b}}U \rho O \rho U^\dagger\ket{\mb{b}}\\
=&\tr\left[S_2\rho\otimes(\rho O \rho)\right]=\tr(O\rho^3).
\end{aligned}
\end{equation}
where in the third line we use the decomposition of $O$.
\end{proof}
The proof of the unbiasedness of $\widehat{o_4}'$ in Eq.~\eqref{eq:o4Papp} is quite similar to that of $\widehat{P_4}'$ in Eq.~\eqref{4mEstZ}, and we do not elaborate it here.

The variance of $\widehat{o_3}'$ and $\widehat{o_4}'$ can be analysed in a similar way as that of $\widehat{P_3}'$, and the leading term shows the same scaling behaviour as in Eq.~\eqref{ap:randomClvar} and \eqref{ap:randomLvar}.
We also numerically study the scalings of the statistical error with the qubit number $n$ using the Pauli measurement in Fig.~\ref{fig:varianceobsApp} (b) and (c), 
which are summarized in the third column of Table \ref{o3table} and Table \ref{o4table}, respectively. The operator $O$ there is taken to be Pauli operator, and thus the control unitary in Fig.~\ref{fig:Odecomp} (a) is directly the Controlled-$O$ operation.

\begin{table}[htb]
\begin{tabular}
{|c|c|c|c|}
\hline
    $o_4$ & OS & HS & HR  \\
    \hline
     Anal.& $\mathcal{O}(\frac{d^{4}}{M^{2}})$ \cite{elben2020mixedstate} & $\mathcal{O}(\frac{d^{1.5}}{M})$ & $\mathcal{O}(\frac{d^{0.79}}{\sqrt{M}K})$ 
 \\ \hline
     Numer.& \diagbox{}{}& $\mathcal{O}(\frac{d^{1.07}}{M})$ [blue]& $\mathcal{O}(\frac{d^{0.72}}{\sqrt{M}K})$[green]\\
     \hline
\end{tabular}
\caption{The statistical errors for estimating $o_4$ with different protocols using the Pauli measurement. The numerical results are from Fig.~\ref{fig:varianceobsApp} (c). The analytical results of HS and HR follow similarly as that of $P_3$ in Eq.~\eqref{ap:eq:P3Var} and \eqref{ap:randomLvar}, respectively. The numerical result of OS is not shown since the measurement and post-processing costs are too demanding for the $4$-degree function.}
\label{o4table}
\end{table}

\begin{figure}
    \centering
    \includegraphics[width=10cm]{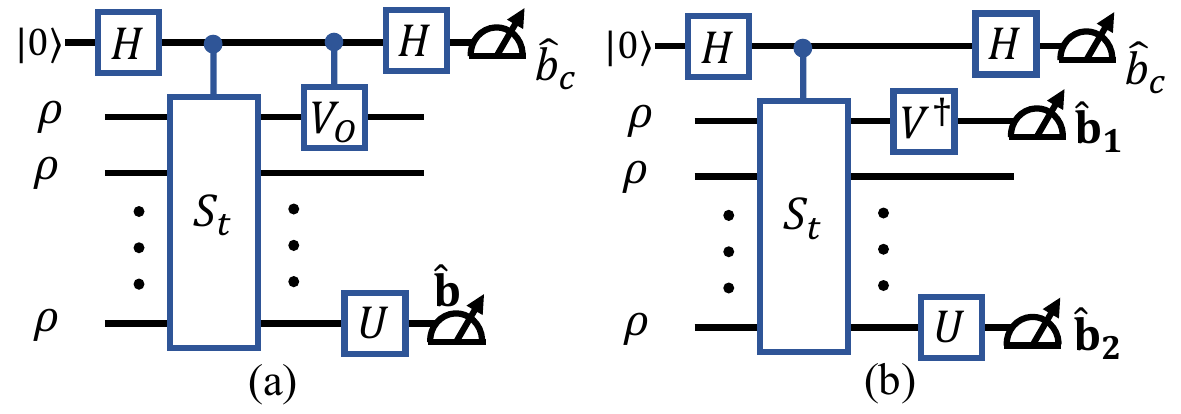}
    \caption{The quantum circuits of two protocols to measure $o_m=\tr(O\rho^m)$. 
    In (a), we decompose the observable as $O=\frac{1}{2}\norm{O}_\infty\left(V_O+V_O^\dagger\right)$, and perform a Controlled-$V_O$ operation on the control qubit and the first copy of $\rho$ after the Controlled-shift operation. Fig.~(4)  in main text is a special case of (a), when $O$ is also a unitary.
    In (b), we take the spectral decomposition of $O$ as $O=V\Lambda_OV^\dagger$, and perform the unitary $V^{\dag}$ on the control qubit. Besides the measurement on the control and final copy, we also conduct projective measurement on the first copy to get the result $\mb{b}_1$.}
    \label{fig:Odecomp}
\end{figure}

The second method is by directly measuring in the diagonal basis of $O$ on the first-copy, and taking into account the measurement result in the final estimator. The advantage compared to the first one is that here one do not need additional control operation. 

Any Hermitian $O$ can be diagonalized as $O=V\Lambda_O V^{\dag}$, with $V$ a (fixed) unitary and $\Lambda_O=\sum \lambda_{\mb{b}}\ketbra{\mb{b}}$ diagonal in the computational basis. Besides the RM on the second-copy with the result denoted by $\mb{b}_2$, we also add the unitary $V^\dag$ and measure the first-copy in the computational basis (essentially measure $O$) and denote the measurement result as $\mb{b}_1$. The corresponding quantum circuit is shown in Fig.~\ref{fig:Odecomp} (b) with $t=2$. Similar as in Eq.~\eqref{eq:fullmatrix}, the joint density matrix before all the projective measurements shows
\begin{equation}
\begin{aligned}
  \rho_{c,I,II}=&\frac{1}{2}\ketbra{0}{0}\otimes(V^{\dag}\otimes U)\rho^{\otimes 2}(V\otimes U^\dagger)+\frac{1}{2}\ketbra{1}{1}\otimes(V^{\dag}\otimes U)S_2\rho^{\otimes 2}S_2(V\otimes U^\dagger)\\ 
  &+\frac{1}{2}\ketbra{0}{1}\otimes(V^{\dag}\otimes U)\rho^{\otimes 2}S_2(V\otimes U^\dagger)+\frac{1}{2}\ketbra{1}{0}\otimes(V^{\dag}\otimes U)S_2\rho^{\otimes 2}(V\otimes U^\dagger).
\end{aligned}
\end{equation}
When the measurement outcomes of the two-copy are $\mb{b}_1,\mb{b}_2$, the (unnormalized) density matrix of the control qubit reads 
\begin{equation}
\begin{aligned}
  \rho_{c}=\frac{1}{2}\bra{\mb{b}_1}V^\dagger \rho V\ket{\mb{b}_1} \bra{\mb{b}_2}U\rho U^\dagger\ket{\mb{b}_2}\id_c+\frac{1}{2}\bra{\mb{b}_2}U\rho V\ketbra{\mb{b}_1}{\mb{b}_1}V^\dag \rho U^\dagger\ket{\mb{b}_2}X_c.
\end{aligned}
\end{equation}
and now the conditional probabilities in Eq.~\eqref{eq:probs} become 
\begin{equation}\label{eq:probsO2}
\begin{aligned}
  &\mathrm{Pr}(0,\mb{b_1},\mb{b_2}|U)=\frac{1}{2}\bra{\mb{b}_1}V^\dagger \rho V\ket{\mb{b}_1} \bra{\mb{b}_2}U\rho U^\dagger\ket{\mb{b}_2}+\frac{1}{2}\bra{\mb{b}_2}U\rho V\ketbra{\mb{b}_1}{\mb{b}_1}V^\dag \rho U^\dagger\ket{\mb{b}_2}\\
  &\mathrm{Pr}(1,\mb{b_1},\mb{b_2}|U)=\frac{1}{2}\bra{\mb{b}_1}V^\dagger \rho V\ket{\mb{b}_1} \bra{\mb{b}_2}U\rho U^\dagger\ket{\mb{b}_2}-\frac{1}{2}\bra{\mb{b}_2}U\rho V\ketbra{\mb{b}_1}{\mb{b}_1}V^\dag \rho U^\dagger\ket{\mb{b}_2}.
\end{aligned}
\end{equation}

Here we fist give the following estimator $(-1)^{\widehat{b}_c}\lambda_{\widehat{\mb{b}_1}}$, and prove that its expectation value on $b_c$ and $\mb{b}_1$ just corresponds to a single-shot RM on $\sigma:=\rho O \rho$.
\begin{equation}\label{eq:o3EstZProofapp}
    \begin{aligned}
       \mathbb{E}_{\{b_c,\mb{b}_1\}}\ (-1)^{\widehat{b}_c}\lambda_{\widehat{\mb{b}_1}}
       &=\sum_{b_c,\mb{b}_1} \mathrm{Pr}(b_c,\mb{b}_1,\mb{b}_2|U)\ (-1)^{b_c}\lambda_{\mb{b}_1}\\
       &=\sum_{\mb{b}_1}\lambda_{\mb{b}_1} [\mathrm{Pr}(0,\mb{b_1},\mb{b_2}|U)-\mathrm{Pr}(1,\mb{b_1},\mb{b_2}|U)]\\
&=\sum_{\mb{b}_1}\lambda_{\mb{b}_1} \bra{\mb{b}_2}U\rho V\ketbra{\mb{b}_1}{\mb{b}_1}V^\dag \rho U^\dagger\ket{\mb{b}_2}\\
&=\bra{\mb{b}_2}U\rho \ \left [V \sum_{\mb{b}_1}\lambda_{\mb{b}_1} \ketbra{\mb{b}_1}{\mb{b}_1}V^\dag\right]\  \rho U^\dagger\ket{\mb{b}_2}=\bra{\mb{b}_2}U\rho O \rho\ U^\dagger\ket{\mb{b}_2}.
    \end{aligned}
\end{equation}
where in the third line we insert the conditional probabilities in Eq.~\eqref{eq:probsO2}, and the final line is by the decomposition of $O$. 

Following Eq.~\eqref{ap:eq:P3Zest}, one can further show that $(-1)^{\widehat{b}_c}\lambda_{\widehat{\mb{b}_1}}X(\widehat{\mb{b}_2},\widehat{\mb{b}'_2})$ is the unbiased estimator of $o_3$, with the measurement result $\{\widehat{b}_c,\widehat{\mb{b}_1},\widehat{\mb{b}_2}\}$ from a single-shot, and $\widehat{\mb{b}'_2}$ from the other shot, under the same random unitary evolution. Finally, by averaging all possible combination of shots, the overall estimator shows
\begin{equation}\label{eq:o3EstZ}
\begin{aligned}
  \widehat{o_3}''=\frac{1}{MK(K-1)}\sum_{i\in[M]}\sum_{j\neq j'\in[K]}(-1)^{\widehat{b}_c^{(i,j)}}\lambda_{\widehat{\mb{b}_1}^{(i,j)}}\ X_{c\backslash p}(\widehat{\mb{b}_2}^{(i,j)},\widehat{\mb{b}_2}^{(i,j')}),
\end{aligned}
\end{equation}
where $\{\widehat{b}_c^{(i,j)},\widehat{\mb{b}_1}^{(i,j)}, \widehat{\mb{b}_2}^{(i,j)}\}$ denotes the measurement result for the control qubit, the first and the second copies of $j$-th shot under the $i$-the unitary evolution, collected from the quantum circuit in Fig.~\ref{fig:Odecomp} (b). We remark that this method can be applied to Sec.~\ref{ap:frame} to estimate function $F_m$ with general observable $\{O_i\}$, but may make the statistical variance larger. 

The estimator of $o_4$ can be constructed similarly by adding the information of the control qubit of $j'$-th shot as follows.
\begin{equation}\label{eq:o4EstZ}
\begin{aligned}
  \widehat{o_4}''=\frac{1}{MK(K-1)}\sum_{i\in[M]}\sum_{j\neq j'\in[K]}(-1)^{\widehat{b}_c^{(i,j)}+\widehat{b}_c^{(i,j')}}\lambda_{\widehat{\mb{b}_1}^{(i,j)}}\ X_{c\backslash p}(\widehat{\mb{b}_2}^{(i,j)},\widehat{\mb{b}_2}^{(i,j')}),
\end{aligned}
\end{equation}
and one can prove its unbiasedness by following Eq.~\eqref{4mEstZproof} and Eq.~\eqref{eq:o3EstZProofapp}.

\end{appendix}


\begin{thebibliography}{60}%
\makeatletter
\providecommand \@ifxundefined [1]{%
 \@ifx{#1\undefined}
}%
\providecommand \@ifnum [1]{%
 \ifnum #1\expandafter \@firstoftwo
 \else \expandafter \@secondoftwo
 \fi
}%
\providecommand \@ifx [1]{%
 \ifx #1\expandafter \@firstoftwo
 \else \expandafter \@secondoftwo
 \fi
}%
\providecommand \natexlab [1]{#1}%
\providecommand \enquote  [1]{``#1''}%
\providecommand \bibnamefont  [1]{#1}%
\providecommand \bibfnamefont [1]{#1}%
\providecommand \citenamefont [1]{#1}%
\providecommand \href@noop [0]{\@secondoftwo}%
\providecommand \href [0]{\begingroup \@sanitize@url \@href}%
\providecommand \@href[1]{\@@startlink{#1}\@@href}%
\providecommand \@@href[1]{\endgroup#1\@@endlink}%
\providecommand \@sanitize@url [0]{\catcode `\\12\catcode `\$12\catcode
  `\&12\catcode `\#12\catcode `\^12\catcode `\_12\catcode `\%12\relax}%
\providecommand \@@startlink[1]{}%
\providecommand \@@endlink[0]{}%
\providecommand \url  [0]{\begingroup\@sanitize@url \@url }%
\providecommand \@url [1]{\endgroup\@href {#1}{\urlprefix }}%
\providecommand \urlprefix  [0]{URL }%
\providecommand \Eprint [0]{\href }%
\providecommand \doibase [0]{https://doi.org/}%
\providecommand \selectlanguage [0]{\@gobble}%
\providecommand \bibinfo  [0]{\@secondoftwo}%
\providecommand \bibfield  [0]{\@secondoftwo}%
\providecommand \translation [1]{[#1]}%
\providecommand \BibitemOpen [0]{}%
\providecommand \bibitemStop [0]{}%
\providecommand \bibitemNoStop [0]{.\EOS\space}%
\providecommand \EOS [0]{\spacefactor3000\relax}%
\providecommand \BibitemShut  [1]{\csname bibitem#1\endcsname}%
\let\auto@bib@innerbib\@empty
\bibitem [{\citenamefont {Ekert}\ \emph {et~al.}(2002)\citenamefont {Ekert},
  \citenamefont {Alves}, \citenamefont {Oi}, \citenamefont {Horodecki},
  \citenamefont {Horodecki},\ and\ \citenamefont {Kwek}}]{Ekert2002Direct}%
  \BibitemOpen
  \bibfield  {author} {\bibinfo {author} {\bibfnamefont {A.~K.}\ \bibnamefont
  {Ekert}}, \bibinfo {author} {\bibfnamefont {C.~M.}\ \bibnamefont {Alves}},
  \bibinfo {author} {\bibfnamefont {D.~K.~L.}\ \bibnamefont {Oi}}, \bibinfo
  {author} {\bibfnamefont {M.}~\bibnamefont {Horodecki}}, \bibinfo {author}
  {\bibfnamefont {P.}~\bibnamefont {Horodecki}},\ and\ \bibinfo {author}
  {\bibfnamefont {L.~C.}\ \bibnamefont {Kwek}},\ }\href
  {https://doi.org/10.1103/PhysRevLett.88.217901} {\bibfield  {journal}
  {\bibinfo  {journal} {Phys. Rev. Lett.}\ }\textbf {\bibinfo {volume} {88}},\
  \bibinfo {pages} {217901} (\bibinfo {year} {2002})}\BibitemShut {NoStop}%
\bibitem [{\citenamefont {Horodecki}\ and\ \citenamefont
  {Ekert}(2002)}]{Horodecki2002Method}%
  \BibitemOpen
  \bibfield  {author} {\bibinfo {author} {\bibfnamefont {P.}~\bibnamefont
  {Horodecki}}\ and\ \bibinfo {author} {\bibfnamefont {A.}~\bibnamefont
  {Ekert}},\ }\href {https://doi.org/10.1103/PhysRevLett.89.127902} {\bibfield
  {journal} {\bibinfo  {journal} {Phys. Rev. Lett.}\ }\textbf {\bibinfo
  {volume} {89}},\ \bibinfo {pages} {127902} (\bibinfo {year}
  {2002})}\BibitemShut {NoStop}%
\bibitem [{\citenamefont {Bruni}(2004)}]{Todd2004Polynomial}%
  \BibitemOpen
  \bibfield  {author} {\bibinfo {author} {\bibfnamefont {T.~A.}\ \bibnamefont
  {Bruni}},\ }\href@noop {} {\bibfield  {journal} {\bibinfo  {journal} {Quantum
  Info. Comput.}\ }\textbf {\bibinfo {volume} {4}},\ \bibinfo {pages}
  {401–408} (\bibinfo {year} {2004})}\BibitemShut {NoStop}%
\bibitem [{\citenamefont {Garcia-Escartin}\ and\ \citenamefont
  {Chamorro-Posada}(2013)}]{Garcia2013equivalent}%
  \BibitemOpen
  \bibfield  {author} {\bibinfo {author} {\bibfnamefont {J.~C.}\ \bibnamefont
  {Garcia-Escartin}}\ and\ \bibinfo {author} {\bibfnamefont {P.}~\bibnamefont
  {Chamorro-Posada}},\ }\href {https://doi.org/10.1103/PhysRevA.87.052330}
  {\bibfield  {journal} {\bibinfo  {journal} {Phys. Rev. A}\ }\textbf {\bibinfo
  {volume} {87}},\ \bibinfo {pages} {052330} (\bibinfo {year}
  {2013})}\BibitemShut {NoStop}%
\bibitem [{\citenamefont {Islam}\ \emph {et~al.}(2015)\citenamefont {Islam},
  \citenamefont {Ma}, \citenamefont {Preiss}, \citenamefont {Eric~Tai},
  \citenamefont {Lukin}, \citenamefont {Rispoli},\ and\ \citenamefont
  {Greiner}}]{Islam2015Measuring}%
  \BibitemOpen
  \bibfield  {author} {\bibinfo {author} {\bibfnamefont {R.}~\bibnamefont
  {Islam}}, \bibinfo {author} {\bibfnamefont {R.}~\bibnamefont {Ma}}, \bibinfo
  {author} {\bibfnamefont {P.~M.}\ \bibnamefont {Preiss}}, \bibinfo {author}
  {\bibfnamefont {M.}~\bibnamefont {Eric~Tai}}, \bibinfo {author}
  {\bibfnamefont {A.}~\bibnamefont {Lukin}}, \bibinfo {author} {\bibfnamefont
  {M.}~\bibnamefont {Rispoli}},\ and\ \bibinfo {author} {\bibfnamefont
  {M.}~\bibnamefont {Greiner}},\ }\href {https://doi.org/10.1038/nature15750}
  {\bibfield  {journal} {\bibinfo  {journal} {Nature}\ }\textbf {\bibinfo
  {volume} {528}},\ \bibinfo {pages} {77} (\bibinfo {year} {2015})}\BibitemShut
  {NoStop}%
\bibitem [{\citenamefont {Kaufman}\ \emph {et~al.}(2016)\citenamefont
  {Kaufman}, \citenamefont {Tai}, \citenamefont {Lukin}, \citenamefont
  {Rispoli}, \citenamefont {Schittko}, \citenamefont {Preiss},\ and\
  \citenamefont {Greiner}}]{Kaufmanen2016tanglement}%
  \BibitemOpen
  \bibfield  {author} {\bibinfo {author} {\bibfnamefont {A.~M.}\ \bibnamefont
  {Kaufman}}, \bibinfo {author} {\bibfnamefont {M.~E.}\ \bibnamefont {Tai}},
  \bibinfo {author} {\bibfnamefont {A.}~\bibnamefont {Lukin}}, \bibinfo
  {author} {\bibfnamefont {M.}~\bibnamefont {Rispoli}}, \bibinfo {author}
  {\bibfnamefont {R.}~\bibnamefont {Schittko}}, \bibinfo {author}
  {\bibfnamefont {P.~M.}\ \bibnamefont {Preiss}},\ and\ \bibinfo {author}
  {\bibfnamefont {M.}~\bibnamefont {Greiner}},\ }\href
  {https://doi.org/10.1126/science.aaf6725} {\bibfield  {journal} {\bibinfo
  {journal} {Science}\ }\textbf {\bibinfo {volume} {353}},\ \bibinfo {pages}
  {794} (\bibinfo {year} {2016})}\BibitemShut {NoStop}%
\bibitem [{\citenamefont {Cotler}\ \emph {et~al.}(2019)\citenamefont {Cotler},
  \citenamefont {Choi}, \citenamefont {Lukin}, \citenamefont {Gharibyan},
  \citenamefont {Grover}, \citenamefont {Tai}, \citenamefont {Rispoli},
  \citenamefont {Schittko}, \citenamefont {Preiss}, \citenamefont {Kaufman},
  \citenamefont {Greiner}, \citenamefont {Pichler},\ and\ \citenamefont
  {Hayden}}]{cotler2019cooling}%
  \BibitemOpen
  \bibfield  {author} {\bibinfo {author} {\bibfnamefont {J.}~\bibnamefont
  {Cotler}}, \bibinfo {author} {\bibfnamefont {S.}~\bibnamefont {Choi}},
  \bibinfo {author} {\bibfnamefont {A.}~\bibnamefont {Lukin}}, \bibinfo
  {author} {\bibfnamefont {H.}~\bibnamefont {Gharibyan}}, \bibinfo {author}
  {\bibfnamefont {T.}~\bibnamefont {Grover}}, \bibinfo {author} {\bibfnamefont
  {M.~E.}\ \bibnamefont {Tai}}, \bibinfo {author} {\bibfnamefont
  {M.}~\bibnamefont {Rispoli}}, \bibinfo {author} {\bibfnamefont
  {R.}~\bibnamefont {Schittko}}, \bibinfo {author} {\bibfnamefont {P.~M.}\
  \bibnamefont {Preiss}}, \bibinfo {author} {\bibfnamefont {A.~M.}\
  \bibnamefont {Kaufman}}, \bibinfo {author} {\bibfnamefont {M.}~\bibnamefont
  {Greiner}}, \bibinfo {author} {\bibfnamefont {H.}~\bibnamefont {Pichler}},\
  and\ \bibinfo {author} {\bibfnamefont {P.}~\bibnamefont {Hayden}},\ }\href
  {https://doi.org/10.1103/PhysRevX.9.031013} {\bibfield  {journal} {\bibinfo
  {journal} {Phys. Rev. X}\ }\textbf {\bibinfo {volume} {9}},\ \bibinfo {pages}
  {031013} (\bibinfo {year} {2019})}\BibitemShut {NoStop}%
\bibitem [{\citenamefont {Huang}\ \emph
  {et~al.}(2022{\natexlab{a}})\citenamefont {Huang}, \citenamefont {Broughton},
  \citenamefont {Cotler}, \citenamefont {Chen}, \citenamefont {Li},
  \citenamefont {Mohseni}, \citenamefont {Neven}, \citenamefont {Babbush},
  \citenamefont {Kueng}, \citenamefont {Preskill} \emph
  {et~al.}}]{huang2022quantum}%
  \BibitemOpen
  \bibfield  {author} {\bibinfo {author} {\bibfnamefont {H.-Y.}\ \bibnamefont
  {Huang}}, \bibinfo {author} {\bibfnamefont {M.}~\bibnamefont {Broughton}},
  \bibinfo {author} {\bibfnamefont {J.}~\bibnamefont {Cotler}}, \bibinfo
  {author} {\bibfnamefont {S.}~\bibnamefont {Chen}}, \bibinfo {author}
  {\bibfnamefont {J.}~\bibnamefont {Li}}, \bibinfo {author} {\bibfnamefont
  {M.}~\bibnamefont {Mohseni}}, \bibinfo {author} {\bibfnamefont
  {H.}~\bibnamefont {Neven}}, \bibinfo {author} {\bibfnamefont
  {R.}~\bibnamefont {Babbush}}, \bibinfo {author} {\bibfnamefont
  {R.}~\bibnamefont {Kueng}}, \bibinfo {author} {\bibfnamefont
  {J.}~\bibnamefont {Preskill}}, \emph {et~al.},\ }\href
  {https://www.science.org/doi/full/10.1126/science.abn7293} {\bibfield
  {journal} {\bibinfo  {journal} {Science}\ }\textbf {\bibinfo {volume}
  {376}},\ \bibinfo {pages} {1182} (\bibinfo {year}
  {2022}{\natexlab{a}})}\BibitemShut {NoStop}%
\bibitem [{\citenamefont {Preskill}(2018)}]{preskill2018quantum}%
  \BibitemOpen
  \bibfield  {author} {\bibinfo {author} {\bibfnamefont {J.}~\bibnamefont
  {Preskill}},\ }\href {https://doi.org/10.22331/q-2018-08-06-79} {\bibfield
  {journal} {\bibinfo  {journal} {{Quantum}}\ }\textbf {\bibinfo {volume}
  {2}},\ \bibinfo {pages} {79} (\bibinfo {year} {2018})}\BibitemShut {NoStop}%
\bibitem [{\citenamefont {Elben}\ \emph {et~al.}(2023)\citenamefont {Elben},
  \citenamefont {Flammia}, \citenamefont {Huang}, \citenamefont {Kueng},
  \citenamefont {Preskill}, \citenamefont {Vermersch},\ and\ \citenamefont
  {Zoller}}]{elben2023randomized}%
  \BibitemOpen
  \bibfield  {author} {\bibinfo {author} {\bibfnamefont {A.}~\bibnamefont
  {Elben}}, \bibinfo {author} {\bibfnamefont {S.~T.}\ \bibnamefont {Flammia}},
  \bibinfo {author} {\bibfnamefont {H.-Y.}\ \bibnamefont {Huang}}, \bibinfo
  {author} {\bibfnamefont {R.}~\bibnamefont {Kueng}}, \bibinfo {author}
  {\bibfnamefont {J.}~\bibnamefont {Preskill}}, \bibinfo {author}
  {\bibfnamefont {B.}~\bibnamefont {Vermersch}},\ and\ \bibinfo {author}
  {\bibfnamefont {P.}~\bibnamefont {Zoller}},\ }\href
  {https://www.nature.com/articles/s42254-022-00535-2} {\bibfield  {journal}
  {\bibinfo  {journal} {Nature Reviews Physics}\ }\textbf {\bibinfo {volume}
  {5}},\ \bibinfo {pages} {9} (\bibinfo {year} {2023})}\BibitemShut {NoStop}%
\bibitem [{\citenamefont {Aaronson}(2019)}]{aaronson2019shadow}%
  \BibitemOpen
  \bibfield  {author} {\bibinfo {author} {\bibfnamefont {S.}~\bibnamefont
  {Aaronson}},\ }\href {https://epubs.siam.org/doi/abs/10.1137/18M120275X}
  {\bibfield  {journal} {\bibinfo  {journal} {SIAM Journal on Computing}\
  }\textbf {\bibinfo {volume} {49}},\ \bibinfo {pages} {STOC18} (\bibinfo
  {year} {2019})}\BibitemShut {NoStop}%
\bibitem [{\citenamefont {Huang}\ \emph {et~al.}(2020)\citenamefont {Huang},
  \citenamefont {Kueng},\ and\ \citenamefont {Preskill}}]{huang2020predicting}%
  \BibitemOpen
  \bibfield  {author} {\bibinfo {author} {\bibfnamefont {H.-Y.}\ \bibnamefont
  {Huang}}, \bibinfo {author} {\bibfnamefont {R.}~\bibnamefont {Kueng}},\ and\
  \bibinfo {author} {\bibfnamefont {J.}~\bibnamefont {Preskill}},\ }\href
  {https://doi.org/10.1038/s41567-020-0932-7} {\bibfield  {journal} {\bibinfo
  {journal} {Nature Physics}\ }\textbf {\bibinfo {volume} {16}},\ \bibinfo
  {pages} {1050} (\bibinfo {year} {2020})}\BibitemShut {NoStop}%
\bibitem [{\citenamefont {Elben}\ \emph
  {et~al.}(2020{\natexlab{a}})\citenamefont {Elben}, \citenamefont {Kueng},
  \citenamefont {Huang}, \citenamefont {van Bijnen}, \citenamefont {Kokail},
  \citenamefont {Dalmonte}, \citenamefont {Calabrese}, \citenamefont {Kraus},
  \citenamefont {Preskill}, \citenamefont {Zoller},\ and\ \citenamefont
  {Vermersch}}]{elben2020mixedstate}%
  \BibitemOpen
  \bibfield  {author} {\bibinfo {author} {\bibfnamefont {A.}~\bibnamefont
  {Elben}}, \bibinfo {author} {\bibfnamefont {R.}~\bibnamefont {Kueng}},
  \bibinfo {author} {\bibfnamefont {H.-Y.~R.}\ \bibnamefont {Huang}}, \bibinfo
  {author} {\bibfnamefont {R.}~\bibnamefont {van Bijnen}}, \bibinfo {author}
  {\bibfnamefont {C.}~\bibnamefont {Kokail}}, \bibinfo {author} {\bibfnamefont
  {M.}~\bibnamefont {Dalmonte}}, \bibinfo {author} {\bibfnamefont
  {P.}~\bibnamefont {Calabrese}}, \bibinfo {author} {\bibfnamefont
  {B.}~\bibnamefont {Kraus}}, \bibinfo {author} {\bibfnamefont
  {J.}~\bibnamefont {Preskill}}, \bibinfo {author} {\bibfnamefont
  {P.}~\bibnamefont {Zoller}},\ and\ \bibinfo {author} {\bibfnamefont
  {B.}~\bibnamefont {Vermersch}},\ }\href
  {https://doi.org/10.1103/PhysRevLett.125.200501} {\bibfield  {journal}
  {\bibinfo  {journal} {Phys. Rev. Lett.}\ }\textbf {\bibinfo {volume} {125}},\
  \bibinfo {pages} {200501} (\bibinfo {year} {2020}{\natexlab{a}})}\BibitemShut
  {NoStop}%
\bibitem [{\citenamefont {Zhou}\ \emph {et~al.}(2020)\citenamefont {Zhou},
  \citenamefont {Zeng},\ and\ \citenamefont {Liu}}]{singlezhou}%
  \BibitemOpen
  \bibfield  {author} {\bibinfo {author} {\bibfnamefont {Y.}~\bibnamefont
  {Zhou}}, \bibinfo {author} {\bibfnamefont {P.}~\bibnamefont {Zeng}},\ and\
  \bibinfo {author} {\bibfnamefont {Z.}~\bibnamefont {Liu}},\ }\href
  {https://doi.org/10.1103/PhysRevLett.125.200502} {\bibfield  {journal}
  {\bibinfo  {journal} {Phys. Rev. Lett.}\ }\textbf {\bibinfo {volume} {125}},\
  \bibinfo {pages} {200502} (\bibinfo {year} {2020})}\BibitemShut {NoStop}%
\bibitem [{\citenamefont {Elben}\ \emph
  {et~al.}(2020{\natexlab{b}})\citenamefont {Elben}, \citenamefont {Vermersch},
  \citenamefont {van Bijnen}, \citenamefont {Kokail}, \citenamefont {Brydges},
  \citenamefont {Maier}, \citenamefont {Joshi}, \citenamefont {Blatt},
  \citenamefont {Roos},\ and\ \citenamefont {Zoller}}]{elben2020cross}%
  \BibitemOpen
  \bibfield  {author} {\bibinfo {author} {\bibfnamefont {A.}~\bibnamefont
  {Elben}}, \bibinfo {author} {\bibfnamefont {B.}~\bibnamefont {Vermersch}},
  \bibinfo {author} {\bibfnamefont {R.}~\bibnamefont {van Bijnen}}, \bibinfo
  {author} {\bibfnamefont {C.}~\bibnamefont {Kokail}}, \bibinfo {author}
  {\bibfnamefont {T.}~\bibnamefont {Brydges}}, \bibinfo {author} {\bibfnamefont
  {C.}~\bibnamefont {Maier}}, \bibinfo {author} {\bibfnamefont {M.~K.}\
  \bibnamefont {Joshi}}, \bibinfo {author} {\bibfnamefont {R.}~\bibnamefont
  {Blatt}}, \bibinfo {author} {\bibfnamefont {C.~F.}\ \bibnamefont {Roos}},\
  and\ \bibinfo {author} {\bibfnamefont {P.}~\bibnamefont {Zoller}},\ }\href
  {https://doi.org/10.1103/PhysRevLett.124.010504} {\bibfield  {journal}
  {\bibinfo  {journal} {Phys. Rev. Lett.}\ }\textbf {\bibinfo {volume} {124}},\
  \bibinfo {pages} {010504} (\bibinfo {year} {2020}{\natexlab{b}})}\BibitemShut
  {NoStop}%
\bibitem [{\citenamefont {Liu}\ \emph {et~al.}(2022{\natexlab{a}})\citenamefont
  {Liu}, \citenamefont {Zeng}, \citenamefont {Zhou},\ and\ \citenamefont
  {Gu}}]{Zhenhuan2022correlation}%
  \BibitemOpen
  \bibfield  {author} {\bibinfo {author} {\bibfnamefont {Z.}~\bibnamefont
  {Liu}}, \bibinfo {author} {\bibfnamefont {P.}~\bibnamefont {Zeng}}, \bibinfo
  {author} {\bibfnamefont {Y.}~\bibnamefont {Zhou}},\ and\ \bibinfo {author}
  {\bibfnamefont {M.}~\bibnamefont {Gu}},\ }\href
  {https://doi.org/10.1103/PhysRevA.105.022407} {\bibfield  {journal} {\bibinfo
   {journal} {Phys. Rev. A}\ }\textbf {\bibinfo {volume} {105}},\ \bibinfo
  {pages} {022407} (\bibinfo {year} {2022}{\natexlab{a}})}\BibitemShut
  {NoStop}%
\bibitem [{\citenamefont {van Enk}\ and\ \citenamefont
  {Beenakker}(2012)}]{van2012Measuring}%
  \BibitemOpen
  \bibfield  {author} {\bibinfo {author} {\bibfnamefont {S.~J.}\ \bibnamefont
  {van Enk}}\ and\ \bibinfo {author} {\bibfnamefont {C.~W.~J.}\ \bibnamefont
  {Beenakker}},\ }\href {https://doi.org/10.1103/PhysRevLett.108.110503}
  {\bibfield  {journal} {\bibinfo  {journal} {Phys. Rev. Lett.}\ }\textbf
  {\bibinfo {volume} {108}},\ \bibinfo {pages} {110503} (\bibinfo {year}
  {2012})}\BibitemShut {NoStop}%
\bibitem [{\citenamefont {Brydges}\ \emph {et~al.}(2019)\citenamefont
  {Brydges}, \citenamefont {Elben}, \citenamefont {Jurcevic}, \citenamefont
  {Vermersch}, \citenamefont {Maier}, \citenamefont {Lanyon}, \citenamefont
  {Zoller}, \citenamefont {Blatt},\ and\ \citenamefont
  {Roos}}]{Brydges2019Probing}%
  \BibitemOpen
  \bibfield  {author} {\bibinfo {author} {\bibfnamefont {T.}~\bibnamefont
  {Brydges}}, \bibinfo {author} {\bibfnamefont {A.}~\bibnamefont {Elben}},
  \bibinfo {author} {\bibfnamefont {P.}~\bibnamefont {Jurcevic}}, \bibinfo
  {author} {\bibfnamefont {B.}~\bibnamefont {Vermersch}}, \bibinfo {author}
  {\bibfnamefont {C.}~\bibnamefont {Maier}}, \bibinfo {author} {\bibfnamefont
  {B.~P.}\ \bibnamefont {Lanyon}}, \bibinfo {author} {\bibfnamefont
  {P.}~\bibnamefont {Zoller}}, \bibinfo {author} {\bibfnamefont
  {R.}~\bibnamefont {Blatt}},\ and\ \bibinfo {author} {\bibfnamefont {C.~F.}\
  \bibnamefont {Roos}},\ }\href {https://doi.org/10.1126/science.aau4963}
  {\bibfield  {journal} {\bibinfo  {journal} {Science}\ }\textbf {\bibinfo
  {volume} {364}},\ \bibinfo {pages} {260} (\bibinfo {year}
  {2019})}\BibitemShut {NoStop}%
\bibitem [{\citenamefont {Rath}\ \emph {et~al.}(2021)\citenamefont {Rath},
  \citenamefont {Branciard}, \citenamefont {Minguzzi},\ and\ \citenamefont
  {Vermersch}}]{rath2021Fisher}%
  \BibitemOpen
  \bibfield  {author} {\bibinfo {author} {\bibfnamefont {A.}~\bibnamefont
  {Rath}}, \bibinfo {author} {\bibfnamefont {C.}~\bibnamefont {Branciard}},
  \bibinfo {author} {\bibfnamefont {A.}~\bibnamefont {Minguzzi}},\ and\
  \bibinfo {author} {\bibfnamefont {B.}~\bibnamefont {Vermersch}},\ }\href
  {https://doi.org/10.1103/PhysRevLett.127.260501} {\bibfield  {journal}
  {\bibinfo  {journal} {Phys. Rev. Lett.}\ }\textbf {\bibinfo {volume} {127}},\
  \bibinfo {pages} {260501} (\bibinfo {year} {2021})}\BibitemShut {NoStop}%
\bibitem [{\citenamefont {Vermersch}\ \emph {et~al.}(2019)\citenamefont
  {Vermersch}, \citenamefont {Elben}, \citenamefont {Sieberer}, \citenamefont
  {Yao},\ and\ \citenamefont {Zoller}}]{Vermersch2019Scrambling}%
  \BibitemOpen
  \bibfield  {author} {\bibinfo {author} {\bibfnamefont {B.}~\bibnamefont
  {Vermersch}}, \bibinfo {author} {\bibfnamefont {A.}~\bibnamefont {Elben}},
  \bibinfo {author} {\bibfnamefont {L.~M.}\ \bibnamefont {Sieberer}}, \bibinfo
  {author} {\bibfnamefont {N.~Y.}\ \bibnamefont {Yao}},\ and\ \bibinfo {author}
  {\bibfnamefont {P.}~\bibnamefont {Zoller}},\ }\href
  {https://doi.org/10.1103/PhysRevX.9.021061} {\bibfield  {journal} {\bibinfo
  {journal} {Phys. Rev. X}\ }\textbf {\bibinfo {volume} {9}},\ \bibinfo {pages}
  {021061} (\bibinfo {year} {2019})}\BibitemShut {NoStop}%
\bibitem [{\citenamefont {Garcia}\ \emph {et~al.}(2021)\citenamefont {Garcia},
  \citenamefont {Zhou},\ and\ \citenamefont {Jaffe}}]{garcia2021quantum}%
  \BibitemOpen
  \bibfield  {author} {\bibinfo {author} {\bibfnamefont {R.~J.}\ \bibnamefont
  {Garcia}}, \bibinfo {author} {\bibfnamefont {Y.}~\bibnamefont {Zhou}},\ and\
  \bibinfo {author} {\bibfnamefont {A.}~\bibnamefont {Jaffe}},\ }\href
  {https://doi.org/10.1103/PhysRevResearch.3.033155} {\bibfield  {journal}
  {\bibinfo  {journal} {Phys. Rev. Research}\ }\textbf {\bibinfo {volume}
  {3}},\ \bibinfo {pages} {033155} (\bibinfo {year} {2021})}\BibitemShut
  {NoStop}%
\bibitem [{\citenamefont {Elben}\ \emph
  {et~al.}(2020{\natexlab{c}})\citenamefont {Elben}, \citenamefont {Yu},
  \citenamefont {Zhu}, \citenamefont {Hafezi}, \citenamefont {Pollmann},
  \citenamefont {Zoller},\ and\ \citenamefont
  {Vermersch}}]{Elben2020topological}%
  \BibitemOpen
  \bibfield  {author} {\bibinfo {author} {\bibfnamefont {A.}~\bibnamefont
  {Elben}}, \bibinfo {author} {\bibfnamefont {J.}~\bibnamefont {Yu}}, \bibinfo
  {author} {\bibfnamefont {G.}~\bibnamefont {Zhu}}, \bibinfo {author}
  {\bibfnamefont {M.}~\bibnamefont {Hafezi}}, \bibinfo {author} {\bibfnamefont
  {F.}~\bibnamefont {Pollmann}}, \bibinfo {author} {\bibfnamefont
  {P.}~\bibnamefont {Zoller}},\ and\ \bibinfo {author} {\bibfnamefont
  {B.}~\bibnamefont {Vermersch}},\ }\href
  {https://doi.org/10.1126/sciadv.aaz3666} {\bibfield  {journal} {\bibinfo
  {journal} {Science Advances}\ }\textbf {\bibinfo {volume} {6}},\ \bibinfo
  {pages} {eaaz3666} (\bibinfo {year} {2020}{\natexlab{c}})}\BibitemShut
  {NoStop}%
\bibitem [{\citenamefont {Cian}\ \emph {et~al.}(2021)\citenamefont {Cian},
  \citenamefont {Dehghani}, \citenamefont {Elben}, \citenamefont {Vermersch},
  \citenamefont {Zhu}, \citenamefont {Barkeshli}, \citenamefont {Zoller},\ and\
  \citenamefont {Hafezi}}]{Cian2020Chern}%
  \BibitemOpen
  \bibfield  {author} {\bibinfo {author} {\bibfnamefont {Z.-P.}\ \bibnamefont
  {Cian}}, \bibinfo {author} {\bibfnamefont {H.}~\bibnamefont {Dehghani}},
  \bibinfo {author} {\bibfnamefont {A.}~\bibnamefont {Elben}}, \bibinfo
  {author} {\bibfnamefont {B.}~\bibnamefont {Vermersch}}, \bibinfo {author}
  {\bibfnamefont {G.}~\bibnamefont {Zhu}}, \bibinfo {author} {\bibfnamefont
  {M.}~\bibnamefont {Barkeshli}}, \bibinfo {author} {\bibfnamefont
  {P.}~\bibnamefont {Zoller}},\ and\ \bibinfo {author} {\bibfnamefont
  {M.}~\bibnamefont {Hafezi}},\ }\href
  {https://doi.org/10.1103/PhysRevLett.126.050501} {\bibfield  {journal}
  {\bibinfo  {journal} {Phys. Rev. Lett.}\ }\textbf {\bibinfo {volume} {126}},\
  \bibinfo {pages} {050501} (\bibinfo {year} {2021})}\BibitemShut {NoStop}%
\bibitem [{\citenamefont {Struchalin}\ \emph {et~al.}(2021)\citenamefont
  {Struchalin}, \citenamefont {Zagorovskii}, \citenamefont {Kovlakov},
  \citenamefont {Straupe},\ and\ \citenamefont
  {Kulik}}]{Struchalin2021Experimental}%
  \BibitemOpen
  \bibfield  {author} {\bibinfo {author} {\bibfnamefont {G.}~\bibnamefont
  {Struchalin}}, \bibinfo {author} {\bibfnamefont {Y.~A.}\ \bibnamefont
  {Zagorovskii}}, \bibinfo {author} {\bibfnamefont {E.}~\bibnamefont
  {Kovlakov}}, \bibinfo {author} {\bibfnamefont {S.}~\bibnamefont {Straupe}},\
  and\ \bibinfo {author} {\bibfnamefont {S.}~\bibnamefont {Kulik}},\ }\href
  {https://doi.org/10.1103/PRXQuantum.2.010307} {\bibfield  {journal} {\bibinfo
   {journal} {PRX Quantum}\ }\textbf {\bibinfo {volume} {2}},\ \bibinfo {pages}
  {010307} (\bibinfo {year} {2021})}\BibitemShut {NoStop}%
\bibitem [{\citenamefont {Zhang}\ \emph {et~al.}(2021)\citenamefont {Zhang},
  \citenamefont {Sun}, \citenamefont {Fang}, \citenamefont {Zhang},
  \citenamefont {Yuan},\ and\ \citenamefont {Lu}}]{zhang2021experimental}%
  \BibitemOpen
  \bibfield  {author} {\bibinfo {author} {\bibfnamefont {T.}~\bibnamefont
  {Zhang}}, \bibinfo {author} {\bibfnamefont {J.}~\bibnamefont {Sun}}, \bibinfo
  {author} {\bibfnamefont {X.-X.}\ \bibnamefont {Fang}}, \bibinfo {author}
  {\bibfnamefont {X.-M.}\ \bibnamefont {Zhang}}, \bibinfo {author}
  {\bibfnamefont {X.}~\bibnamefont {Yuan}},\ and\ \bibinfo {author}
  {\bibfnamefont {H.}~\bibnamefont {Lu}},\ }\href
  {https://doi.org/10.1103/PhysRevLett.127.200501} {\bibfield  {journal}
  {\bibinfo  {journal} {Phys. Rev. Lett.}\ }\textbf {\bibinfo {volume} {127}},\
  \bibinfo {pages} {200501} (\bibinfo {year} {2021})}\BibitemShut {NoStop}%
\bibitem [{\citenamefont {Yu}\ \emph {et~al.}(2021{\natexlab{a}})\citenamefont
  {Yu}, \citenamefont {Li}, \citenamefont {Wang}, \citenamefont {Chu},
  \citenamefont {Yang}, \citenamefont {Gong}, \citenamefont {Goldman},\ and\
  \citenamefont {Cai}}]{Yu2021Fisher}%
  \BibitemOpen
  \bibfield  {author} {\bibinfo {author} {\bibfnamefont {M.}~\bibnamefont
  {Yu}}, \bibinfo {author} {\bibfnamefont {D.}~\bibnamefont {Li}}, \bibinfo
  {author} {\bibfnamefont {J.}~\bibnamefont {Wang}}, \bibinfo {author}
  {\bibfnamefont {Y.}~\bibnamefont {Chu}}, \bibinfo {author} {\bibfnamefont
  {P.}~\bibnamefont {Yang}}, \bibinfo {author} {\bibfnamefont {M.}~\bibnamefont
  {Gong}}, \bibinfo {author} {\bibfnamefont {N.}~\bibnamefont {Goldman}},\ and\
  \bibinfo {author} {\bibfnamefont {J.}~\bibnamefont {Cai}},\ }\href
  {https://doi.org/10.1103/PhysRevResearch.3.043122} {\bibfield  {journal}
  {\bibinfo  {journal} {Phys. Rev. Research}\ }\textbf {\bibinfo {volume}
  {3}},\ \bibinfo {pages} {043122} (\bibinfo {year}
  {2021}{\natexlab{a}})}\BibitemShut {NoStop}%
\bibitem [{\citenamefont {Hu}\ and\ \citenamefont
  {You}(2022)}]{Hu2022Hamiltonian}%
  \BibitemOpen
  \bibfield  {author} {\bibinfo {author} {\bibfnamefont {H.-Y.}\ \bibnamefont
  {Hu}}\ and\ \bibinfo {author} {\bibfnamefont {Y.-Z.}\ \bibnamefont {You}},\
  }\href {https://doi.org/10.1103/PhysRevResearch.4.013054} {\bibfield
  {journal} {\bibinfo  {journal} {Phys. Rev. Research}\ }\textbf {\bibinfo
  {volume} {4}},\ \bibinfo {pages} {013054} (\bibinfo {year}
  {2022})}\BibitemShut {NoStop}%
\bibitem [{\citenamefont {Hu}\ \emph {et~al.}(2021)\citenamefont {Hu},
  \citenamefont {Choi},\ and\ \citenamefont {You}}]{hu2022Locally}%
  \BibitemOpen
  \bibfield  {author} {\bibinfo {author} {\bibfnamefont {H.-Y.}\ \bibnamefont
  {Hu}}, \bibinfo {author} {\bibfnamefont {S.}~\bibnamefont {Choi}},\ and\
  \bibinfo {author} {\bibfnamefont {Y.-Z.}\ \bibnamefont {You}},\ }\href
  {https://arxiv.org/abs/2107.04817} {\bibfield  {journal} {\bibinfo  {journal}
  {arXiv preprint arXiv:2107.04817}\ } (\bibinfo {year} {2021})}\BibitemShut
  {NoStop}%
\bibitem [{\citenamefont {Ohliger}\ \emph {et~al.}(2013)\citenamefont
  {Ohliger}, \citenamefont {Nesme},\ and\ \citenamefont
  {Eisert}}]{ohliger2013efficient}%
  \BibitemOpen
  \bibfield  {author} {\bibinfo {author} {\bibfnamefont {M.}~\bibnamefont
  {Ohliger}}, \bibinfo {author} {\bibfnamefont {V.}~\bibnamefont {Nesme}},\
  and\ \bibinfo {author} {\bibfnamefont {J.}~\bibnamefont {Eisert}},\ }\href
  {https://iopscience.iop.org/article/10.1088/1367-2630/15/1/015024/meta}
  {\bibfield  {journal} {\bibinfo  {journal} {New Journal of Physics}\ }\textbf
  {\bibinfo {volume} {15}},\ \bibinfo {pages} {015024} (\bibinfo {year}
  {2013})}\BibitemShut {NoStop}%
\bibitem [{\citenamefont {Bu}\ \emph {et~al.}(2022)\citenamefont {Bu},
  \citenamefont {Koh}, \citenamefont {Garcia},\ and\ \citenamefont
  {Jaffe}}]{bu2022classical}%
  \BibitemOpen
  \bibfield  {author} {\bibinfo {author} {\bibfnamefont {K.}~\bibnamefont
  {Bu}}, \bibinfo {author} {\bibfnamefont {D.~E.}\ \bibnamefont {Koh}},
  \bibinfo {author} {\bibfnamefont {R.~J.}\ \bibnamefont {Garcia}},\ and\
  \bibinfo {author} {\bibfnamefont {A.}~\bibnamefont {Jaffe}},\ }\href
  {https://arxiv.org/abs/2202.03272} {\bibfield  {journal} {\bibinfo  {journal}
  {arXiv preprint arXiv:2202.03272}\ } (\bibinfo {year} {2022})}\BibitemShut
  {NoStop}%
\bibitem [{sup()}]{supplementary}%
  \BibitemOpen
  \href@noop {} {}\bibinfo {note} {See Supplementary Materials for proofs of
  main results, further discussions, and more applications.}\BibitemShut
  {Stop}%
\bibitem [{\citenamefont {Elben}\ \emph {et~al.}(2019)\citenamefont {Elben},
  \citenamefont {Vermersch}, \citenamefont {Roos},\ and\ \citenamefont
  {Zoller}}]{Elben2019toolbox}%
  \BibitemOpen
  \bibfield  {author} {\bibinfo {author} {\bibfnamefont {A.}~\bibnamefont
  {Elben}}, \bibinfo {author} {\bibfnamefont {B.}~\bibnamefont {Vermersch}},
  \bibinfo {author} {\bibfnamefont {C.~F.}\ \bibnamefont {Roos}},\ and\
  \bibinfo {author} {\bibfnamefont {P.}~\bibnamefont {Zoller}},\ }\href
  {https://doi.org/10.1103/PhysRevA.99.052323} {\bibfield  {journal} {\bibinfo
  {journal} {Phys. Rev. A}\ }\textbf {\bibinfo {volume} {99}},\ \bibinfo
  {pages} {052323} (\bibinfo {year} {2019})}\BibitemShut {NoStop}%
\bibitem [{\citenamefont {Huggins}\ \emph {et~al.}(2021)\citenamefont
  {Huggins}, \citenamefont {McArdle}, \citenamefont {O'Brien}, \citenamefont
  {Lee}, \citenamefont {Rubin}, \citenamefont {Boixo}, \citenamefont {Whaley},
  \citenamefont {Babbush},\ and\ \citenamefont {McClean}}]{Huggins2021Virtual}%
  \BibitemOpen
  \bibfield  {author} {\bibinfo {author} {\bibfnamefont {W.~J.}\ \bibnamefont
  {Huggins}}, \bibinfo {author} {\bibfnamefont {S.}~\bibnamefont {McArdle}},
  \bibinfo {author} {\bibfnamefont {T.~E.}\ \bibnamefont {O'Brien}}, \bibinfo
  {author} {\bibfnamefont {J.}~\bibnamefont {Lee}}, \bibinfo {author}
  {\bibfnamefont {N.~C.}\ \bibnamefont {Rubin}}, \bibinfo {author}
  {\bibfnamefont {S.}~\bibnamefont {Boixo}}, \bibinfo {author} {\bibfnamefont
  {K.~B.}\ \bibnamefont {Whaley}}, \bibinfo {author} {\bibfnamefont
  {R.}~\bibnamefont {Babbush}},\ and\ \bibinfo {author} {\bibfnamefont {J.~R.}\
  \bibnamefont {McClean}},\ }\href {https://doi.org/10.1103/PhysRevX.11.041036}
  {\bibfield  {journal} {\bibinfo  {journal} {Phys. Rev. X}\ }\textbf {\bibinfo
  {volume} {11}},\ \bibinfo {pages} {041036} (\bibinfo {year}
  {2021})}\BibitemShut {NoStop}%
\bibitem [{\citenamefont {Koczor}(2021)}]{Koczor2021Exponential}%
  \BibitemOpen
  \bibfield  {author} {\bibinfo {author} {\bibfnamefont {B.}~\bibnamefont
  {Koczor}},\ }\href {https://doi.org/10.1103/PhysRevX.11.031057} {\bibfield
  {journal} {\bibinfo  {journal} {Phys. Rev. X}\ }\textbf {\bibinfo {volume}
  {11}},\ \bibinfo {pages} {031057} (\bibinfo {year} {2021})}\BibitemShut
  {NoStop}%
\bibitem [{\citenamefont {Temme}\ \emph {et~al.}(2017)\citenamefont {Temme},
  \citenamefont {Bravyi},\ and\ \citenamefont {Gambetta}}]{temme2017error}%
  \BibitemOpen
  \bibfield  {author} {\bibinfo {author} {\bibfnamefont {K.}~\bibnamefont
  {Temme}}, \bibinfo {author} {\bibfnamefont {S.}~\bibnamefont {Bravyi}},\ and\
  \bibinfo {author} {\bibfnamefont {J.~M.}\ \bibnamefont {Gambetta}},\ }\href
  {https://journals.aps.org/prl/abstract/10.1103/PhysRevLett.119.180509}
  {\bibfield  {journal} {\bibinfo  {journal} {Physical review letters}\
  }\textbf {\bibinfo {volume} {119}},\ \bibinfo {pages} {180509} (\bibinfo
  {year} {2017})}\BibitemShut {NoStop}%
\bibitem [{\citenamefont {Endo}\ \emph {et~al.}(2018)\citenamefont {Endo},
  \citenamefont {Benjamin},\ and\ \citenamefont {Li}}]{Suguru2018Practical}%
  \BibitemOpen
  \bibfield  {author} {\bibinfo {author} {\bibfnamefont {S.}~\bibnamefont
  {Endo}}, \bibinfo {author} {\bibfnamefont {S.~C.}\ \bibnamefont {Benjamin}},\
  and\ \bibinfo {author} {\bibfnamefont {Y.}~\bibnamefont {Li}},\ }\href
  {https://doi.org/10.1103/PhysRevX.8.031027} {\bibfield  {journal} {\bibinfo
  {journal} {Phys. Rev. X}\ }\textbf {\bibinfo {volume} {8}},\ \bibinfo {pages}
  {031027} (\bibinfo {year} {2018})}\BibitemShut {NoStop}%
\bibitem [{\citenamefont {Kandala}\ \emph {et~al.}(2019)\citenamefont
  {Kandala}, \citenamefont {Temme}, \citenamefont {C{\'o}rcoles}, \citenamefont
  {Mezzacapo}, \citenamefont {Chow},\ and\ \citenamefont
  {Gambetta}}]{kandala2019error}%
  \BibitemOpen
  \bibfield  {author} {\bibinfo {author} {\bibfnamefont {A.}~\bibnamefont
  {Kandala}}, \bibinfo {author} {\bibfnamefont {K.}~\bibnamefont {Temme}},
  \bibinfo {author} {\bibfnamefont {A.~D.}\ \bibnamefont {C{\'o}rcoles}},
  \bibinfo {author} {\bibfnamefont {A.}~\bibnamefont {Mezzacapo}}, \bibinfo
  {author} {\bibfnamefont {J.~M.}\ \bibnamefont {Chow}},\ and\ \bibinfo
  {author} {\bibfnamefont {J.~M.}\ \bibnamefont {Gambetta}},\ }\href
  {https://www.nature.com/articles/s41586-019-1040-7} {\bibfield  {journal}
  {\bibinfo  {journal} {Nature}\ }\textbf {\bibinfo {volume} {567}},\ \bibinfo
  {pages} {491} (\bibinfo {year} {2019})}\BibitemShut {NoStop}%
\bibitem [{\citenamefont {Endo}\ \emph {et~al.}(2021)\citenamefont {Endo},
  \citenamefont {Cai}, \citenamefont {Benjamin},\ and\ \citenamefont
  {Yuan}}]{endo2021hybrid}%
  \BibitemOpen
  \bibfield  {author} {\bibinfo {author} {\bibfnamefont {S.}~\bibnamefont
  {Endo}}, \bibinfo {author} {\bibfnamefont {Z.}~\bibnamefont {Cai}}, \bibinfo
  {author} {\bibfnamefont {S.~C.}\ \bibnamefont {Benjamin}},\ and\ \bibinfo
  {author} {\bibfnamefont {X.}~\bibnamefont {Yuan}},\ }\href
  {https://journals.jps.jp/doi/full/10.7566/JPSJ.90.032001} {\bibfield
  {journal} {\bibinfo  {journal} {Journal of the Physical Society of Japan}\
  }\textbf {\bibinfo {volume} {90}},\ \bibinfo {pages} {032001} (\bibinfo
  {year} {2021})}\BibitemShut {NoStop}%
\bibitem [{\citenamefont {Seif}\ \emph {et~al.}(2023)\citenamefont {Seif},
  \citenamefont {Cian}, \citenamefont {Zhou}, \citenamefont {Chen},\ and\
  \citenamefont {Jiang}}]{Seif2022Shadow}%
  \BibitemOpen
  \bibfield  {author} {\bibinfo {author} {\bibfnamefont {A.}~\bibnamefont
  {Seif}}, \bibinfo {author} {\bibfnamefont {Z.-P.}\ \bibnamefont {Cian}},
  \bibinfo {author} {\bibfnamefont {S.}~\bibnamefont {Zhou}}, \bibinfo {author}
  {\bibfnamefont {S.}~\bibnamefont {Chen}},\ and\ \bibinfo {author}
  {\bibfnamefont {L.}~\bibnamefont {Jiang}},\ }\href
  {https://doi.org/10.1103/PRXQuantum.4.010303} {\bibfield  {journal} {\bibinfo
   {journal} {PRX Quantum}\ }\textbf {\bibinfo {volume} {4}},\ \bibinfo {pages}
  {010303} (\bibinfo {year} {2023})}\BibitemShut {NoStop}%
\bibitem [{\citenamefont {Hu}\ \emph {et~al.}(2022)\citenamefont {Hu},
  \citenamefont {LaRose}, \citenamefont {You}, \citenamefont {Rieffel},\ and\
  \citenamefont {Wang}}]{Hu2022Logical}%
  \BibitemOpen
  \bibfield  {author} {\bibinfo {author} {\bibfnamefont {H.-Y.}\ \bibnamefont
  {Hu}}, \bibinfo {author} {\bibfnamefont {R.}~\bibnamefont {LaRose}}, \bibinfo
  {author} {\bibfnamefont {Y.-Z.}\ \bibnamefont {You}}, \bibinfo {author}
  {\bibfnamefont {E.}~\bibnamefont {Rieffel}},\ and\ \bibinfo {author}
  {\bibfnamefont {Z.}~\bibnamefont {Wang}},\ }\href
  {https://arxiv.org/abs/2203.07263} {\bibfield  {journal} {\bibinfo  {journal}
  {arXiv preprint arXiv:2203.07263}\ } (\bibinfo {year} {2022})}\BibitemShut
  {NoStop}%
\bibitem [{\citenamefont {O'Brien}\ \emph {et~al.}(2022)\citenamefont
  {O'Brien}, \citenamefont {Anselmetti}, \citenamefont {Gkritsis},
  \citenamefont {Elfving}, \citenamefont {Polla}, \citenamefont {Huggins},
  \citenamefont {Oumarou}, \citenamefont {Kechedzhi}, \citenamefont {Abanin},
  \citenamefont {Acharya} \emph {et~al.}}]{Brien2022purification}%
  \BibitemOpen
  \bibfield  {author} {\bibinfo {author} {\bibfnamefont {T.~E.}\ \bibnamefont
  {O'Brien}}, \bibinfo {author} {\bibfnamefont {G.}~\bibnamefont {Anselmetti}},
  \bibinfo {author} {\bibfnamefont {F.}~\bibnamefont {Gkritsis}}, \bibinfo
  {author} {\bibfnamefont {V.}~\bibnamefont {Elfving}}, \bibinfo {author}
  {\bibfnamefont {S.}~\bibnamefont {Polla}}, \bibinfo {author} {\bibfnamefont
  {W.~J.}\ \bibnamefont {Huggins}}, \bibinfo {author} {\bibfnamefont
  {O.}~\bibnamefont {Oumarou}}, \bibinfo {author} {\bibfnamefont
  {K.}~\bibnamefont {Kechedzhi}}, \bibinfo {author} {\bibfnamefont
  {D.}~\bibnamefont {Abanin}}, \bibinfo {author} {\bibfnamefont
  {R.}~\bibnamefont {Acharya}}, \emph {et~al.},\ }\href
  {https://arxiv.org/abs/2210.10799} {\bibfield  {journal} {\bibinfo  {journal}
  {arXiv preprint arXiv:2210.10799}\ } (\bibinfo {year} {2022})}\BibitemShut
  {NoStop}%
\bibitem [{\citenamefont {McArdle}\ \emph {et~al.}(2020)\citenamefont
  {McArdle}, \citenamefont {Endo}, \citenamefont {Aspuru-Guzik}, \citenamefont
  {Benjamin},\ and\ \citenamefont {Yuan}}]{McArdle2020chemistry}%
  \BibitemOpen
  \bibfield  {author} {\bibinfo {author} {\bibfnamefont {S.}~\bibnamefont
  {McArdle}}, \bibinfo {author} {\bibfnamefont {S.}~\bibnamefont {Endo}},
  \bibinfo {author} {\bibfnamefont {A.}~\bibnamefont {Aspuru-Guzik}}, \bibinfo
  {author} {\bibfnamefont {S.~C.}\ \bibnamefont {Benjamin}},\ and\ \bibinfo
  {author} {\bibfnamefont {X.}~\bibnamefont {Yuan}},\ }\href
  {https://doi.org/10.1103/RevModPhys.92.015003} {\bibfield  {journal}
  {\bibinfo  {journal} {Rev. Mod. Phys.}\ }\textbf {\bibinfo {volume} {92}},\
  \bibinfo {pages} {015003} (\bibinfo {year} {2020})}\BibitemShut {NoStop}%
\bibitem [{\citenamefont {Arute}\ \emph {et~al.}(2020)\citenamefont {Arute},
  \citenamefont {Arya}, \citenamefont {Babbush}, \citenamefont {Bacon},
  \citenamefont {Bardin}, \citenamefont {Barends}, \citenamefont {Boixo},
  \citenamefont {Broughton}, \citenamefont {Buckley}, \citenamefont {Buell},
  \citenamefont {Burkett}, \citenamefont {Bushnell}, \citenamefont {Chen},
  \citenamefont {Chen}, \citenamefont {Chiaro}, \citenamefont {Collins},
  \citenamefont {Courtney}, \citenamefont {Demura}, \citenamefont {Dunsworth},
  \citenamefont {Farhi}, \citenamefont {Fowler}, \citenamefont {Foxen},
  \citenamefont {Gidney}, \citenamefont {Giustina}, \citenamefont {Graff},
  \citenamefont {Habegger}, \citenamefont {Harrigan}, \citenamefont {Ho},
  \citenamefont {Hong}, \citenamefont {Huang}, \citenamefont {Huggins},
  \citenamefont {Ioffe}, \citenamefont {Isakov}, \citenamefont {Jeffrey},
  \citenamefont {Jiang}, \citenamefont {Jones}, \citenamefont {Kafri},
  \citenamefont {Kechedzhi}, \citenamefont {Kelly}, \citenamefont {Kim},
  \citenamefont {Klimov}, \citenamefont {Korotkov}, \citenamefont {Kostritsa},
  \citenamefont {Landhuis}, \citenamefont {Laptev}, \citenamefont {Lindmark},
  \citenamefont {Lucero}, \citenamefont {Martin}, \citenamefont {Martinis},
  \citenamefont {McClean}, \citenamefont {McEwen}, \citenamefont {Megrant},
  \citenamefont {Mi}, \citenamefont {Mohseni}, \citenamefont {Mruczkiewicz},
  \citenamefont {Mutus}, \citenamefont {Naaman}, \citenamefont {Neeley},
  \citenamefont {Neill}, \citenamefont {Neven}, \citenamefont {Niu},
  \citenamefont {O'Brien}, \citenamefont {Ostby}, \citenamefont {Petukhov},
  \citenamefont {Putterman}, \citenamefont {Quintana}, \citenamefont {Roushan},
  \citenamefont {Rubin}, \citenamefont {Sank}, \citenamefont {Satzinger},
  \citenamefont {Smelyanskiy}, \citenamefont {Strain}, \citenamefont {Sung},
  \citenamefont {Szalay}, \citenamefont {Takeshita}, \citenamefont
  {Vainsencher}, \citenamefont {White}, \citenamefont {Wiebe}, \citenamefont
  {Yao}, \citenamefont {Yeh},\ and\ \citenamefont
  {Zalcman}}]{google2020Hartree}%
  \BibitemOpen
  \bibfield  {author} {\bibinfo {author} {\bibfnamefont {F.}~\bibnamefont
  {Arute}}, \bibinfo {author} {\bibfnamefont {K.}~\bibnamefont {Arya}},
  \bibinfo {author} {\bibfnamefont {R.}~\bibnamefont {Babbush}}, \bibinfo
  {author} {\bibfnamefont {D.}~\bibnamefont {Bacon}}, \bibinfo {author}
  {\bibfnamefont {J.~C.}\ \bibnamefont {Bardin}}, \bibinfo {author}
  {\bibfnamefont {R.}~\bibnamefont {Barends}}, \bibinfo {author} {\bibfnamefont
  {S.}~\bibnamefont {Boixo}}, \bibinfo {author} {\bibfnamefont
  {M.}~\bibnamefont {Broughton}}, \bibinfo {author} {\bibfnamefont {B.~B.}\
  \bibnamefont {Buckley}}, \bibinfo {author} {\bibfnamefont {D.~A.}\
  \bibnamefont {Buell}}, \bibinfo {author} {\bibfnamefont {B.}~\bibnamefont
  {Burkett}}, \bibinfo {author} {\bibfnamefont {N.}~\bibnamefont {Bushnell}},
  \bibinfo {author} {\bibfnamefont {Y.}~\bibnamefont {Chen}}, \bibinfo {author}
  {\bibfnamefont {Z.}~\bibnamefont {Chen}}, \bibinfo {author} {\bibfnamefont
  {B.}~\bibnamefont {Chiaro}}, \bibinfo {author} {\bibfnamefont
  {R.}~\bibnamefont {Collins}}, \bibinfo {author} {\bibfnamefont
  {W.}~\bibnamefont {Courtney}}, \bibinfo {author} {\bibfnamefont
  {S.}~\bibnamefont {Demura}}, \bibinfo {author} {\bibfnamefont
  {A.}~\bibnamefont {Dunsworth}}, \bibinfo {author} {\bibfnamefont
  {E.}~\bibnamefont {Farhi}}, \bibinfo {author} {\bibfnamefont
  {A.}~\bibnamefont {Fowler}}, \bibinfo {author} {\bibfnamefont
  {B.}~\bibnamefont {Foxen}}, \bibinfo {author} {\bibfnamefont
  {C.}~\bibnamefont {Gidney}}, \bibinfo {author} {\bibfnamefont
  {M.}~\bibnamefont {Giustina}}, \bibinfo {author} {\bibfnamefont
  {R.}~\bibnamefont {Graff}}, \bibinfo {author} {\bibfnamefont
  {S.}~\bibnamefont {Habegger}}, \bibinfo {author} {\bibfnamefont {M.~P.}\
  \bibnamefont {Harrigan}}, \bibinfo {author} {\bibfnamefont {A.}~\bibnamefont
  {Ho}}, \bibinfo {author} {\bibfnamefont {S.}~\bibnamefont {Hong}}, \bibinfo
  {author} {\bibfnamefont {T.}~\bibnamefont {Huang}}, \bibinfo {author}
  {\bibfnamefont {W.~J.}\ \bibnamefont {Huggins}}, \bibinfo {author}
  {\bibfnamefont {L.}~\bibnamefont {Ioffe}}, \bibinfo {author} {\bibfnamefont
  {S.~V.}\ \bibnamefont {Isakov}}, \bibinfo {author} {\bibfnamefont
  {E.}~\bibnamefont {Jeffrey}}, \bibinfo {author} {\bibfnamefont
  {Z.}~\bibnamefont {Jiang}}, \bibinfo {author} {\bibfnamefont
  {C.}~\bibnamefont {Jones}}, \bibinfo {author} {\bibfnamefont
  {D.}~\bibnamefont {Kafri}}, \bibinfo {author} {\bibfnamefont
  {K.}~\bibnamefont {Kechedzhi}}, \bibinfo {author} {\bibfnamefont
  {J.}~\bibnamefont {Kelly}}, \bibinfo {author} {\bibfnamefont
  {S.}~\bibnamefont {Kim}}, \bibinfo {author} {\bibfnamefont {P.~V.}\
  \bibnamefont {Klimov}}, \bibinfo {author} {\bibfnamefont {A.}~\bibnamefont
  {Korotkov}}, \bibinfo {author} {\bibfnamefont {F.}~\bibnamefont {Kostritsa}},
  \bibinfo {author} {\bibfnamefont {D.}~\bibnamefont {Landhuis}}, \bibinfo
  {author} {\bibfnamefont {P.}~\bibnamefont {Laptev}}, \bibinfo {author}
  {\bibfnamefont {M.}~\bibnamefont {Lindmark}}, \bibinfo {author}
  {\bibfnamefont {E.}~\bibnamefont {Lucero}}, \bibinfo {author} {\bibfnamefont
  {O.}~\bibnamefont {Martin}}, \bibinfo {author} {\bibfnamefont {J.~M.}\
  \bibnamefont {Martinis}}, \bibinfo {author} {\bibfnamefont {J.~R.}\
  \bibnamefont {McClean}}, \bibinfo {author} {\bibfnamefont {M.}~\bibnamefont
  {McEwen}}, \bibinfo {author} {\bibfnamefont {A.}~\bibnamefont {Megrant}},
  \bibinfo {author} {\bibfnamefont {X.}~\bibnamefont {Mi}}, \bibinfo {author}
  {\bibfnamefont {M.}~\bibnamefont {Mohseni}}, \bibinfo {author} {\bibfnamefont
  {W.}~\bibnamefont {Mruczkiewicz}}, \bibinfo {author} {\bibfnamefont
  {J.}~\bibnamefont {Mutus}}, \bibinfo {author} {\bibfnamefont
  {O.}~\bibnamefont {Naaman}}, \bibinfo {author} {\bibfnamefont
  {M.}~\bibnamefont {Neeley}}, \bibinfo {author} {\bibfnamefont
  {C.}~\bibnamefont {Neill}}, \bibinfo {author} {\bibfnamefont
  {H.}~\bibnamefont {Neven}}, \bibinfo {author} {\bibfnamefont {M.~Y.}\
  \bibnamefont {Niu}}, \bibinfo {author} {\bibfnamefont {T.~E.}\ \bibnamefont
  {O'Brien}}, \bibinfo {author} {\bibfnamefont {E.}~\bibnamefont {Ostby}},
  \bibinfo {author} {\bibfnamefont {A.}~\bibnamefont {Petukhov}}, \bibinfo
  {author} {\bibfnamefont {H.}~\bibnamefont {Putterman}}, \bibinfo {author}
  {\bibfnamefont {C.}~\bibnamefont {Quintana}}, \bibinfo {author}
  {\bibfnamefont {P.}~\bibnamefont {Roushan}}, \bibinfo {author} {\bibfnamefont
  {N.~C.}\ \bibnamefont {Rubin}}, \bibinfo {author} {\bibfnamefont
  {D.}~\bibnamefont {Sank}}, \bibinfo {author} {\bibfnamefont {K.~J.}\
  \bibnamefont {Satzinger}}, \bibinfo {author} {\bibfnamefont {V.}~\bibnamefont
  {Smelyanskiy}}, \bibinfo {author} {\bibfnamefont {D.}~\bibnamefont {Strain}},
  \bibinfo {author} {\bibfnamefont {K.~J.}\ \bibnamefont {Sung}}, \bibinfo
  {author} {\bibfnamefont {M.}~\bibnamefont {Szalay}}, \bibinfo {author}
  {\bibfnamefont {T.~Y.}\ \bibnamefont {Takeshita}}, \bibinfo {author}
  {\bibfnamefont {A.}~\bibnamefont {Vainsencher}}, \bibinfo {author}
  {\bibfnamefont {T.}~\bibnamefont {White}}, \bibinfo {author} {\bibfnamefont
  {N.}~\bibnamefont {Wiebe}}, \bibinfo {author} {\bibfnamefont {Z.~J.}\
  \bibnamefont {Yao}}, \bibinfo {author} {\bibfnamefont {P.}~\bibnamefont
  {Yeh}},\ and\ \bibinfo {author} {\bibfnamefont {A.}~\bibnamefont {Zalcman}},\
  }\href {https://doi.org/10.1126/science.abb9811} {\bibfield  {journal}
  {\bibinfo  {journal} {Science}\ }\textbf {\bibinfo {volume} {369}},\ \bibinfo
  {pages} {1084} (\bibinfo {year} {2020})}\BibitemShut {NoStop}%
\bibitem [{\citenamefont {Ketterer}\ \emph {et~al.}(2019)\citenamefont
  {Ketterer}, \citenamefont {Wyderka},\ and\ \citenamefont
  {G\"uhne}}]{ketterer2019characterizing}%
  \BibitemOpen
  \bibfield  {author} {\bibinfo {author} {\bibfnamefont {A.}~\bibnamefont
  {Ketterer}}, \bibinfo {author} {\bibfnamefont {N.}~\bibnamefont {Wyderka}},\
  and\ \bibinfo {author} {\bibfnamefont {O.}~\bibnamefont {G\"uhne}},\ }\href
  {https://doi.org/10.1103/PhysRevLett.122.120505} {\bibfield  {journal}
  {\bibinfo  {journal} {Phys. Rev. Lett.}\ }\textbf {\bibinfo {volume} {122}},\
  \bibinfo {pages} {120505} (\bibinfo {year} {2019})}\BibitemShut {NoStop}%
\bibitem [{\citenamefont {Neven}\ \emph {et~al.}(2021)\citenamefont {Neven},
  \citenamefont {Carrasco}, \citenamefont {Vitale}, \citenamefont {Kokail},
  \citenamefont {Elben}, \citenamefont {Dalmonte}, \citenamefont {Calabrese},
  \citenamefont {Zoller}, \citenamefont {Vermersch}, \citenamefont {Kueng}
  \emph {et~al.}}]{neven2021symmetry}%
  \BibitemOpen
  \bibfield  {author} {\bibinfo {author} {\bibfnamefont {A.}~\bibnamefont
  {Neven}}, \bibinfo {author} {\bibfnamefont {J.}~\bibnamefont {Carrasco}},
  \bibinfo {author} {\bibfnamefont {V.}~\bibnamefont {Vitale}}, \bibinfo
  {author} {\bibfnamefont {C.}~\bibnamefont {Kokail}}, \bibinfo {author}
  {\bibfnamefont {A.}~\bibnamefont {Elben}}, \bibinfo {author} {\bibfnamefont
  {M.}~\bibnamefont {Dalmonte}}, \bibinfo {author} {\bibfnamefont
  {P.}~\bibnamefont {Calabrese}}, \bibinfo {author} {\bibfnamefont
  {P.}~\bibnamefont {Zoller}}, \bibinfo {author} {\bibfnamefont
  {B.}~\bibnamefont {Vermersch}}, \bibinfo {author} {\bibfnamefont
  {R.}~\bibnamefont {Kueng}}, \emph {et~al.},\ }\href
  {https://www.nature.com/articles/s41534-021-00487-y} {\bibfield  {journal}
  {\bibinfo  {journal} {npj Quantum Information}\ }\textbf {\bibinfo {volume}
  {7}},\ \bibinfo {pages} {1} (\bibinfo {year} {2021})}\BibitemShut {NoStop}%
\bibitem [{\citenamefont {Yu}\ \emph {et~al.}(2021{\natexlab{b}})\citenamefont
  {Yu}, \citenamefont {Imai},\ and\ \citenamefont {G\"uhne}}]{Yu2021Optimal}%
  \BibitemOpen
  \bibfield  {author} {\bibinfo {author} {\bibfnamefont {X.-D.}\ \bibnamefont
  {Yu}}, \bibinfo {author} {\bibfnamefont {S.}~\bibnamefont {Imai}},\ and\
  \bibinfo {author} {\bibfnamefont {O.}~\bibnamefont {G\"uhne}},\ }\href
  {https://doi.org/10.1103/PhysRevLett.127.060504} {\bibfield  {journal}
  {\bibinfo  {journal} {Phys. Rev. Lett.}\ }\textbf {\bibinfo {volume} {127}},\
  \bibinfo {pages} {060504} (\bibinfo {year} {2021}{\natexlab{b}})}\BibitemShut
  {NoStop}%
\bibitem [{\citenamefont {Liu}\ \emph {et~al.}(2022{\natexlab{b}})\citenamefont
  {Liu}, \citenamefont {Tang}, \citenamefont {Dai}, \citenamefont {Liu},
  \citenamefont {Chen},\ and\ \citenamefont {Ma}}]{liu2022detecting}%
  \BibitemOpen
  \bibfield  {author} {\bibinfo {author} {\bibfnamefont {Z.}~\bibnamefont
  {Liu}}, \bibinfo {author} {\bibfnamefont {Y.}~\bibnamefont {Tang}}, \bibinfo
  {author} {\bibfnamefont {H.}~\bibnamefont {Dai}}, \bibinfo {author}
  {\bibfnamefont {P.}~\bibnamefont {Liu}}, \bibinfo {author} {\bibfnamefont
  {S.}~\bibnamefont {Chen}},\ and\ \bibinfo {author} {\bibfnamefont
  {X.}~\bibnamefont {Ma}},\ }\href
  {https://doi.org/10.1103/PhysRevLett.129.260501} {\bibfield  {journal}
  {\bibinfo  {journal} {Phys. Rev. Lett.}\ }\textbf {\bibinfo {volume} {129}},\
  \bibinfo {pages} {260501} (\bibinfo {year} {2022}{\natexlab{b}})}\BibitemShut
  {NoStop}%
\bibitem [{\citenamefont {Joshi}\ \emph {et~al.}(2022)\citenamefont {Joshi},
  \citenamefont {Elben}, \citenamefont {Vikram}, \citenamefont {Vermersch},
  \citenamefont {Galitski},\ and\ \citenamefont {Zoller}}]{Joshi2022Probing}%
  \BibitemOpen
  \bibfield  {author} {\bibinfo {author} {\bibfnamefont {L.~K.}\ \bibnamefont
  {Joshi}}, \bibinfo {author} {\bibfnamefont {A.}~\bibnamefont {Elben}},
  \bibinfo {author} {\bibfnamefont {A.}~\bibnamefont {Vikram}}, \bibinfo
  {author} {\bibfnamefont {B.}~\bibnamefont {Vermersch}}, \bibinfo {author}
  {\bibfnamefont {V.}~\bibnamefont {Galitski}},\ and\ \bibinfo {author}
  {\bibfnamefont {P.}~\bibnamefont {Zoller}},\ }\href
  {https://doi.org/10.1103/PhysRevX.12.011018} {\bibfield  {journal} {\bibinfo
  {journal} {Phys. Rev. X}\ }\textbf {\bibinfo {volume} {12}},\ \bibinfo
  {pages} {011018} (\bibinfo {year} {2022})}\BibitemShut {NoStop}%
\bibitem [{\citenamefont {McGinley}\ \emph {et~al.}(2022)\citenamefont
  {McGinley}, \citenamefont {Leontica}, \citenamefont {Garratt}, \citenamefont
  {Jovanovic},\ and\ \citenamefont {Simon}}]{McGinley2022scrambling}%
  \BibitemOpen
  \bibfield  {author} {\bibinfo {author} {\bibfnamefont {M.}~\bibnamefont
  {McGinley}}, \bibinfo {author} {\bibfnamefont {S.}~\bibnamefont {Leontica}},
  \bibinfo {author} {\bibfnamefont {S.~J.}\ \bibnamefont {Garratt}}, \bibinfo
  {author} {\bibfnamefont {J.}~\bibnamefont {Jovanovic}},\ and\ \bibinfo
  {author} {\bibfnamefont {S.~H.}\ \bibnamefont {Simon}},\ }\href
  {https://arxiv.org/abs/2202.05132} {\bibfield  {journal} {\bibinfo  {journal}
  {arXiv preprint arXiv:2202.05132}\ } (\bibinfo {year} {2022})}\BibitemShut
  {NoStop}%
\bibitem [{\citenamefont {Yuan}\ \emph {et~al.}(2021)\citenamefont {Yuan},
  \citenamefont {Sun}, \citenamefont {Liu}, \citenamefont {Zhao},\ and\
  \citenamefont {Zhou}}]{Xiao2021Hybrid}%
  \BibitemOpen
  \bibfield  {author} {\bibinfo {author} {\bibfnamefont {X.}~\bibnamefont
  {Yuan}}, \bibinfo {author} {\bibfnamefont {J.}~\bibnamefont {Sun}}, \bibinfo
  {author} {\bibfnamefont {J.}~\bibnamefont {Liu}}, \bibinfo {author}
  {\bibfnamefont {Q.}~\bibnamefont {Zhao}},\ and\ \bibinfo {author}
  {\bibfnamefont {Y.}~\bibnamefont {Zhou}},\ }\href
  {https://doi.org/10.1103/PhysRevLett.127.040501} {\bibfield  {journal}
  {\bibinfo  {journal} {Phys. Rev. Lett.}\ }\textbf {\bibinfo {volume} {127}},\
  \bibinfo {pages} {040501} (\bibinfo {year} {2021})}\BibitemShut {NoStop}%
\bibitem [{\citenamefont {Lubasch}\ \emph {et~al.}(2020)\citenamefont
  {Lubasch}, \citenamefont {Joo}, \citenamefont {Moinier}, \citenamefont
  {Kiffner},\ and\ \citenamefont {Jaksch}}]{Lubasch2020nonlinear}%
  \BibitemOpen
  \bibfield  {author} {\bibinfo {author} {\bibfnamefont {M.}~\bibnamefont
  {Lubasch}}, \bibinfo {author} {\bibfnamefont {J.}~\bibnamefont {Joo}},
  \bibinfo {author} {\bibfnamefont {P.}~\bibnamefont {Moinier}}, \bibinfo
  {author} {\bibfnamefont {M.}~\bibnamefont {Kiffner}},\ and\ \bibinfo {author}
  {\bibfnamefont {D.}~\bibnamefont {Jaksch}},\ }\href
  {https://doi.org/10.1103/PhysRevA.101.010301} {\bibfield  {journal} {\bibinfo
   {journal} {Phys. Rev. A}\ }\textbf {\bibinfo {volume} {101}},\ \bibinfo
  {pages} {010301} (\bibinfo {year} {2020})}\BibitemShut {NoStop}%
\bibitem [{\citenamefont {Yamamoto}\ \emph {et~al.}(2021)\citenamefont
  {Yamamoto}, \citenamefont {Endo}, \citenamefont {Hakoshima}, \citenamefont
  {Matsuzaki},\ and\ \citenamefont {Tokunaga}}]{Yamamoto2021metrology}%
  \BibitemOpen
  \bibfield  {author} {\bibinfo {author} {\bibfnamefont {K.}~\bibnamefont
  {Yamamoto}}, \bibinfo {author} {\bibfnamefont {S.}~\bibnamefont {Endo}},
  \bibinfo {author} {\bibfnamefont {H.}~\bibnamefont {Hakoshima}}, \bibinfo
  {author} {\bibfnamefont {Y.}~\bibnamefont {Matsuzaki}},\ and\ \bibinfo
  {author} {\bibfnamefont {Y.}~\bibnamefont {Tokunaga}},\ }\href
  {https://doi.org/10.48550/ARXIV.2112.01850} {\bibinfo {title}
  {Error-mitigated quantum metrology via virtual purification}} (\bibinfo
  {year} {2021})\BibitemShut {NoStop}%
\bibitem [{\citenamefont {Chen}\ \emph {et~al.}(2022)\citenamefont {Chen},
  \citenamefont {Cotler}, \citenamefont {Huang},\ and\ \citenamefont
  {Li}}]{chen2022exponential}%
  \BibitemOpen
  \bibfield  {author} {\bibinfo {author} {\bibfnamefont {S.}~\bibnamefont
  {Chen}}, \bibinfo {author} {\bibfnamefont {J.}~\bibnamefont {Cotler}},
  \bibinfo {author} {\bibfnamefont {H.-Y.}\ \bibnamefont {Huang}},\ and\
  \bibinfo {author} {\bibfnamefont {J.}~\bibnamefont {Li}},\ }in\ \href
  {https://ieeexplore.ieee.org/abstract/document/9719827/} {\emph {\bibinfo
  {booktitle} {2021 IEEE 62nd Annual Symposium on Foundations of Computer
  Science (FOCS)}}}\ (\bibinfo {organization} {IEEE},\ \bibinfo {year} {2022})\
  pp.\ \bibinfo {pages} {574--585}\BibitemShut {NoStop}%
\bibitem [{\citenamefont {Huang}\ \emph
  {et~al.}(2022{\natexlab{b}})\citenamefont {Huang}, \citenamefont {Broughton},
  \citenamefont {Cotler}, \citenamefont {Chen}, \citenamefont {Li},
  \citenamefont {Mohseni}, \citenamefont {Neven}, \citenamefont {Babbush},
  \citenamefont {Kueng}, \citenamefont {Preskill} \emph
  {et~al.}}]{huang2022Science}%
  \BibitemOpen
  \bibfield  {author} {\bibinfo {author} {\bibfnamefont {H.-Y.}\ \bibnamefont
  {Huang}}, \bibinfo {author} {\bibfnamefont {M.}~\bibnamefont {Broughton}},
  \bibinfo {author} {\bibfnamefont {J.}~\bibnamefont {Cotler}}, \bibinfo
  {author} {\bibfnamefont {S.}~\bibnamefont {Chen}}, \bibinfo {author}
  {\bibfnamefont {J.}~\bibnamefont {Li}}, \bibinfo {author} {\bibfnamefont
  {M.}~\bibnamefont {Mohseni}}, \bibinfo {author} {\bibfnamefont
  {H.}~\bibnamefont {Neven}}, \bibinfo {author} {\bibfnamefont
  {R.}~\bibnamefont {Babbush}}, \bibinfo {author} {\bibfnamefont
  {R.}~\bibnamefont {Kueng}}, \bibinfo {author} {\bibfnamefont
  {J.}~\bibnamefont {Preskill}}, \emph {et~al.},\ }\href
  {https://www.science.org/doi/abs/10.1126/science.abn7293} {\bibfield
  {journal} {\bibinfo  {journal} {Science}\ }\textbf {\bibinfo {volume}
  {376}},\ \bibinfo {pages} {1182} (\bibinfo {year}
  {2022}{\natexlab{b}})}\BibitemShut {NoStop}%
\bibitem [{\citenamefont {Chen}\ \emph {et~al.}(2021)\citenamefont {Chen},
  \citenamefont {Yu}, \citenamefont {Zeng},\ and\ \citenamefont
  {Flammia}}]{Chen2021Robust}%
  \BibitemOpen
  \bibfield  {author} {\bibinfo {author} {\bibfnamefont {S.}~\bibnamefont
  {Chen}}, \bibinfo {author} {\bibfnamefont {W.}~\bibnamefont {Yu}}, \bibinfo
  {author} {\bibfnamefont {P.}~\bibnamefont {Zeng}},\ and\ \bibinfo {author}
  {\bibfnamefont {S.~T.}\ \bibnamefont {Flammia}},\ }\href
  {https://doi.org/10.1103/PRXQuantum.2.030348} {\bibfield  {journal} {\bibinfo
   {journal} {PRX Quantum}\ }\textbf {\bibinfo {volume} {2}},\ \bibinfo {pages}
  {030348} (\bibinfo {year} {2021})}\BibitemShut {NoStop}%
\bibitem [{\citenamefont {Helsen}\ \emph {et~al.}(2021)\citenamefont {Helsen},
  \citenamefont {Ioannou}, \citenamefont {Roth}, \citenamefont {Kitzinger},
  \citenamefont {Onorati}, \citenamefont {Werner},\ and\ \citenamefont
  {Eisert}}]{Helsen2021Estimating}%
  \BibitemOpen
  \bibfield  {author} {\bibinfo {author} {\bibfnamefont {J.}~\bibnamefont
  {Helsen}}, \bibinfo {author} {\bibfnamefont {M.}~\bibnamefont {Ioannou}},
  \bibinfo {author} {\bibfnamefont {I.}~\bibnamefont {Roth}}, \bibinfo {author}
  {\bibfnamefont {J.}~\bibnamefont {Kitzinger}}, \bibinfo {author}
  {\bibfnamefont {E.}~\bibnamefont {Onorati}}, \bibinfo {author} {\bibfnamefont
  {A.~H.}\ \bibnamefont {Werner}},\ and\ \bibinfo {author} {\bibfnamefont
  {J.}~\bibnamefont {Eisert}},\ }\href {https://arxiv.org/abs/2110.13178}
  {\bibfield  {journal} {\bibinfo  {journal} {arXiv preprint arXiv:2110.13178}\
  } (\bibinfo {year} {2021})}\BibitemShut {NoStop}%
\bibitem [{\citenamefont {Kunjummen}\ \emph {et~al.}(2021)\citenamefont
  {Kunjummen}, \citenamefont {Tran}, \citenamefont {Carney},\ and\
  \citenamefont {Taylor}}]{Kunjummen2021process}%
  \BibitemOpen
  \bibfield  {author} {\bibinfo {author} {\bibfnamefont {J.}~\bibnamefont
  {Kunjummen}}, \bibinfo {author} {\bibfnamefont {M.~C.}\ \bibnamefont {Tran}},
  \bibinfo {author} {\bibfnamefont {D.}~\bibnamefont {Carney}},\ and\ \bibinfo
  {author} {\bibfnamefont {J.~M.}\ \bibnamefont {Taylor}},\ }\href
  {https://arxiv.org/abs/2110.03629} {\bibfield  {journal} {\bibinfo  {journal}
  {arXiv preprint arXiv:2110.03629}\ } (\bibinfo {year} {2021})}\BibitemShut
  {NoStop}%
\bibitem [{\citenamefont {Levy}\ \emph {et~al.}(2021)\citenamefont {Levy},
  \citenamefont {Luo},\ and\ \citenamefont {Clark}}]{Levy2021Process}%
  \BibitemOpen
  \bibfield  {author} {\bibinfo {author} {\bibfnamefont {R.}~\bibnamefont
  {Levy}}, \bibinfo {author} {\bibfnamefont {D.}~\bibnamefont {Luo}},\ and\
  \bibinfo {author} {\bibfnamefont {B.~K.}\ \bibnamefont {Clark}},\ }\href
  {https://arxiv.org/abs/2110.02965} {\bibfield  {journal} {\bibinfo  {journal}
  {arXiv preprint arXiv:2110.02965}\ } (\bibinfo {year} {2021})}\BibitemShut
  {NoStop}%
\bibitem [{\citenamefont {Elben}\ \emph {et~al.}(2018)\citenamefont {Elben},
  \citenamefont {Vermersch}, \citenamefont {Dalmonte}, \citenamefont {Cirac},\
  and\ \citenamefont {Zoller}}]{Elben2018Random}%
  \BibitemOpen
  \bibfield  {author} {\bibinfo {author} {\bibfnamefont {A.}~\bibnamefont
  {Elben}}, \bibinfo {author} {\bibfnamefont {B.}~\bibnamefont {Vermersch}},
  \bibinfo {author} {\bibfnamefont {M.}~\bibnamefont {Dalmonte}}, \bibinfo
  {author} {\bibfnamefont {J.~I.}\ \bibnamefont {Cirac}},\ and\ \bibinfo
  {author} {\bibfnamefont {P.}~\bibnamefont {Zoller}},\ }\href
  {https://doi.org/10.1103/PhysRevLett.120.050406} {\bibfield  {journal}
  {\bibinfo  {journal} {Phys. Rev. Lett.}\ }\textbf {\bibinfo {volume} {120}},\
  \bibinfo {pages} {050406} (\bibinfo {year} {2018})}\BibitemShut {NoStop}%
\bibitem [{\citenamefont {Zhao}\ \emph {et~al.}(2021)\citenamefont {Zhao},
  \citenamefont {Rubin},\ and\ \citenamefont {Miyake}}]{Zhao2021Fermionic}%
  \BibitemOpen
  \bibfield  {author} {\bibinfo {author} {\bibfnamefont {A.}~\bibnamefont
  {Zhao}}, \bibinfo {author} {\bibfnamefont {N.~C.}\ \bibnamefont {Rubin}},\
  and\ \bibinfo {author} {\bibfnamefont {A.}~\bibnamefont {Miyake}},\ }\href
  {https://doi.org/10.1103/PhysRevLett.127.110504} {\bibfield  {journal}
  {\bibinfo  {journal} {Phys. Rev. Lett.}\ }\textbf {\bibinfo {volume} {127}},\
  \bibinfo {pages} {110504} (\bibinfo {year} {2021})}\BibitemShut {NoStop}%
\end{thebibliography}
\end{document}